\documentclass[ejs]{imsart}

\usepackage[T1]{fontenc}
\usepackage{lmodern}
\RequirePackage{amsthm,amsmath,amssymb}
\RequirePackage[numbers]{natbib}
\RequirePackage[colorlinks,citecolor=blue,urlcolor=blue]{hyperref}
\RequirePackage{graphicx}
\usepackage{float}
\usepackage{amsfonts}
\usepackage{xcolor}
\usepackage{mathbbol}
\usepackage{comment}
\usepackage{lipsum}
\DeclareMathOperator*{\argmax}{argmax}

\usepackage{xcolor}
\usepackage{subcaption} 
\usepackage{adjustbox}  
\usepackage{enumitem}
\usepackage{tikz}
\usetikzlibrary{positioning} 
\usepackage{tabularx}
\usepackage{booktabs}
\newcolumntype{R}{>{\centering\arraybackslash}p{0.22\textwidth}}

\DeclareSymbolFontAlphabet{\amsmathbb}{AMSb}%

\startlocaldefs

\DeclareSymbolFontAlphabet{\amsmathbb}{AMSb}%
\newcommand{\pr}{{\amsmathbb{P}}}
\newcommand{\E}{\amsmathbb{E}}
\newcommand*\dd{\mathop{}\!\mathrm{d}}

\numberwithin{equation}{section}
\theoremstyle{plain}
\newtheorem{theorem}{Theorem}
\newtheorem{lemma}[theorem]{Lemma}

\newtheorem{corollary}[theorem]{Corollary}
\newtheorem{assumption}[theorem]{Assumption}

\newtheorem{definition}{Definition}
\newtheorem{remark}[theorem]{Remark}

\endlocaldefs

\begin{document}
\begin{frontmatter}
\title{Scoring Rules with Normalized Upper Order Statistics for Tail Inference
}
\runtitle{Scoring Rules with Normalized Upper Order Statistics for Tail Inference}

\begin{aug}

\author{%
  \fnms{Martin} \snm{Bladt}\ead[label=e1]{martinbladt@math.ku.dk}
  \and
  \fnms{Christoffer} \snm{Øhlenschlæger}\ead[label=e2]{choh@math.ku.dk}
}
\address{%
  Department of Mathematical Sciences, University of Copenhagen\\
  Universitetsparken 5, 2100 Copenhagen Ø, Denmark \\
  \printead{e1,e2}
}

\end{aug}

\begin{abstract}
This paper proposes a scoring-rule-based method for ranking predictive distributions in the Fréchet domain that is able to distinguish between different tail indices. The approach is built on normalized order statistics and exploits proper scoring rules to compare tail limit distributions in a distributional framework, with direct relevance for insurance claim-severity tails. On the theoretical side, consistency and asymptotic normality for empirical tail scores based on normalized upper order statistics are obtained through residual estimation theory. Simulation results demonstrate that the scoring-rule-based approach is capable of discriminating between different tail behaviors in finite samples and that
systematic scale variation has only a minor impact on stability.
We further show that optimizing scoring rules yields consistent tail-index estimators and that the classical Hill estimator arises as a special case. The performance of the Energy Score tail-index estimator is investigated and compared with the Hill estimator across a range of tail indices. Lastly, we analyze an automobile claim-severity data set to demonstrate how scoring rules can be used to rank predictive models based on tail predictions in actuarial settings.
\end{abstract}

\end{frontmatter}


\section{Introduction}

Floods, earthquakes, and other catastrophic events occur infrequently, yet they can generate substantial financial losses. A sound understanding of both the frequency and the severity of such extremes is therefore essential for actuarial risk assessment. To model such phenomena, \emph{Extreme Value Theory} (EVT) provides a fundamental framework and is a standard tool in actuarial science and risk management (e.g., \cite{EmbrechtsKluppelbergMikosch1997}). 
In practice, however, risk assessments have traditionally relied primarily on point forecasts. By adopting predictive distributions, one obtains a more complete description of risk, and a growing body of literature has therefore focused on distributional approaches to uncertainty quantification (see, e.g., \cite{Benedetti2010, RuizAbellonFernandezJimenezGuillamonGabaldon2024, KleinSmithNott2023}). Assessment of predictive distributions typically involves two complementary aspects: \emph{calibration}, which examines whether predicted probabilities are statistically consistent with observed outcomes, and \emph{scoring rules}, which provide a principled way to compare and rank different predictive distributions (\cite{Gneit_Raft, GneitingBalabdaouiRaftery2007}).
In actuarial applications, model selection rarely involves choosing between a clearly correct and a clearly incorrect specification. Practitioners typically face several competing models, so a data-driven ranking provides an objective and reproducible basis for pricing, reserving, and risk communication. In addition, the full ranking can guide the selection of alternative models when practical considerations such as interpretability, robustness, or regulatory constraints are taken into account.
Accordingly, this paper proposes a method for ranking predictive distributions in the Fréchet domain based on proper scoring rules and normalized upper order statistics.

Because this approach relies on scoring-rule comparisons, it is crucial that the scoring rule remains \emph{proper}, i.e.,
\[
\amsmathbb{E}_{Y \sim F} [ S(F, Y) ] \geq \amsmathbb{E}_{Y \sim F} [ S(G, Y) ],
\]
for all predictive distributions \( G \), where \( S(\cdot, \cdot) \) denotes the scoring rule and \( Y \sim F \) represents the true data-generating distribution.  
Properness ensures that the forecaster is incentivized to issue their genuine predictive distribution \cite{Gneit_Raft}.  
The dangers of using improper scoring rules can be seen in \cite{LerchThorarinsdottirRavazzoloGneiting2017}, which demonstrate how misassessment of extreme events can lead forecasters to issue predictions that are more extreme than they actually believe.

{Against this background, there is a vast literature on the evaluation of predictive distributions with a focus on extreme events.} One common approach is to employ weighted scoring rules, which allow greater emphasis to be placed on large observations, for example by using a weight function of the form \( w(y) = \mathbf{1}\{y \geq t\} \) for some threshold \( t \).
Studies investigating proper weighted scoring rules include \cite{GneitingRanjan2011}, \cite{HolzmannKlar2017}, \cite{Allen2024}, \cite{WesselFerroEvansKwasniok2025} and \cite{OlafsdottirRootzenBolin2024}. In addition to their theoretical contributions, several of these studies provide empirical illustrations of how weighted scoring rules can be applied to evaluate tail performance. {However, a central concern in the literature is the ability of predictive distributions to differentiate between tail indices.} \cite{TaillardatFougeresNaveauDeFondeville2023} shows that the widely used Continuous Ranked Probability Score (CRPS) fails to distinguish between predictive distributions with different tail indices. \cite{score-tail-kristin} extend the result by proving that no proper scoring rule can discriminate between tail indices, i.e., it is always possible to construct two competing predictive distributions with distinct tail indices whose scoring rule values are arbitrarily close. \cite{Allen2025TailCalibration} strengthen this finding by showing that this limitation cannot be overcome by taking the maximum over multiple scoring rules. These studies conclude that scoring rules are unsuitable when the objective is to assess or compare tail indices. 

{Our approach addresses this limitation by changing the object on which the scoring rule is applied. Rather than evaluating predictive distributions on the full sample, we base the comparison on normalized upper order statistics. Let \(Y_1,\ldots,Y_n\) denote the observations, with order statistics \(Y_{n,1}\leq \cdots \leq Y_{n,n}\). For a given number \(k\) of upper order statistics, the threshold is \(Y_{n,n-k}\), and each exceedance is evaluated relative to this threshold. Thus, the observations entering the score are
\[
    Y_{n,n-i+1}/Y_{n,n-k}, \qquad i=1,\ldots,k .
\]
This normalization leads to a limiting comparison in terms of the tail distribution.
Accordingly, the score is applied to a tail object obtained by restricting attention to
the upper tail and normalizing the exceedances, rather than directly to the full
predictive distribution on the original sample space. The relevant propriety statement
therefore concerns this tail object. In practice, the tail object may be chosen as either
the finite-threshold tail counterpart of the predictive distribution, which assesses
conditional tail fit above a high threshold, or the corresponding Pareto limit, which
targets the asymptotic tail distribution and hence the tail index; see
Remark~\ref{rem:gg}.}

{This perspective is also closely related to weighted scoring rules.
In \cite{HolzmannKlar2017} a study on how scoring rules can be transformed into
weighted scoring rules while preserving propriety is provided. In particular, their
Theorem~3 constructs weighted scoring rules by first normalizing the predictive
distribution on the weighted region and then multiplying the resulting score by
the weight function. With a deterministic threshold \(t\), the corresponding analogue in our setting is obtained by considering the rescaled variable \(Y/t\).
On this scale, exceedances over \(t\) correspond exactly to the event \(Y/t\geq 1\), so the relevant weight is \(w(y)=\mathbf{1}\{y\geq 1\}\). The renormalized predictive distribution on this weighted region is then the conditional distribution of \(Y/t\) given \(Y\geq t\). This connection also highlights where our approach departs from standard fixed-threshold weighted scoring rules. Weighted scoring rules can focus the evaluation on observations above a chosen threshold, but this alone does not provide a theoretical connection to the tail index. If the threshold is fixed, the score still evaluates a fixed part of the distribution, and there is no guarantee that this region reflects the limiting tail behavior. Hence, fixed-threshold weighting does not by itself resolve the difficulty that scoring-rule comparisons may fail to distinguish between different tail indices. The key feature of our approach is that the threshold is chosen as an intermediate order statistic and increases with the sample size. The score is thus evaluated on normalized exceedances in an increasingly extreme region, where the limiting distribution is determined by the tail index.}

{A related but distinct approach to tail-focused forecast evaluation is provided by tail calibration.
For instance, \cite{Allen2025TailCalibration} introduce tail calibration as a tool for assessing calibration in
the extremes and show that it can help distinguish between distributions with different tail
behaviors. This perspective is complementary to the score-based ranking framework considered
here and applies beyond the Fréchet domain.}

Beyond model ranking, this paper also investigates scoring-rule optimization as a method for estimating the tail index. {The idea is to choose the parameter value whose predictive distribution attains the largest empirical score. More precisely, let $\{F_\theta:\theta\in\Theta\}$ be a parametric class of distributions, let $S$ be a strictly proper scoring rule, and let $(Y_1,\ldots,Y_n)$ be an i.i.d.\ sample. The scoring-rule estimator is then defined by
\begin{align*}
    \hat{\theta}_S
    =
    \arg\max_{\theta\in\Theta}
    \sum_{i=1}^n S(F_\theta,Y_i).
\end{align*}}
\cite{Dawid2016} shows that, under regularity conditions, this estimator is consistent and asymptotically Gaussian. In addition, the associated estimating equation is shown to be unbiased. This places minimum scoring rule inference within the broader framework of $M$-estimation (see, e.g., \cite{HuberRonchetti2009}). Scoring-rule estimators have also found applications in the applied sciences, for instance in weather forecasting \cite{FriederichsThorarinsdottir2012}, where CRPS is used. For a comprehensive review of scoring-rule inference, see \cite{WaghmareZiegel2026}.

{Finally, we investigate the proposed framework through simulation studies and an empirical application. As part of the simulation study, we examine the effect of heterogeneous scaling. Such scale heterogeneity is not covered by the asymptotic theory in Section~\ref{sec:math}, but it is relevant in forecasting applications. We therefore include it as a simulation-based robustness check. In the extreme value theory literature, related non-identically distributed settings have been studied theoretically by \cite{EinmahlDeHaanZhou2016}.
 For the empirical study, we use the \texttt{usautoBI} (USAutoBI) bodily injury claims data from the \texttt{CASdatasets} package in \textsf{R}. Five Pareto predictive distributions are compared and ranked using the proposed scoring-rule approach. We further consider two data partitions based on policyholder characteristics. While one split leads to the same ranking as in the full sample, the other results in a different ordering of the models, indicating that tail behavior varies across subpopulations.}

In summary, our contributions are threefold: first, we propose a scoring-rule framework for tail-model ranking based on normalized upper order statistics; second, we provide asymptotic justification and show how score optimization yields a consistent tail-index estimator with Hill as a special case; and third, we illustrate the framework in simulation studies and in a real claim-severity application.

The paper is organized as follows. Section \ref{sec:pre} introduces the core concepts from both the scoring rule and extreme value theory literature that are used throughout the paper. Section \ref{sec:method} introduces our proposed method and Section \ref{sec:math} presents the main theoretical results. Section \ref{sec:sim} investigates the finite-sample performance of our method through simulation studies, demonstrating promising results even for relatively small sample sizes. Finally, Section \ref{sec:Data} applies the proposed scoring rule framework to a claim severity dataset.

\section{Background and Preliminaries} \label{sec:pre}
\subsection{Proper Scoring Rules}

{Let \(\mathcal{Y}=\mathbb{R}\) be equipped with the Borel
\(\sigma\)-algebra \(\mathcal{A}\), and let \(\mathcal{F}\) denote a
convex class of distribution functions on \(\mathcal{Y}\).} Before defining the score function, we introduce the term quasi-integrable.

\begin{definition}[$\mathcal{F}$-quasi-integrable]
    $f$ is quasi-integrable with respect to $F\in\mathcal{F}$ if it is $\mathcal{A}$ measurable and if either $f^+$ or $f^-$ has a real integral.
    $f$ is $\mathcal{F}$-quasi-integrable if it is quasi-integrable with respect to all {$F\in \mathcal{F}$}.
\end{definition}

With this setup, a scoring rule is defined as

\begin{definition}[Scoring rule] \label{def: Score rule}
   Any $S: {\mathcal{F} \times \mathcal{Y}} \to \overline{\mathbb{R}}$ is a scoring rule if it is $\mathcal{F}$-quasi-integrable for all ${F\in \mathcal{F}}$. 
\end{definition}

With this notation, we write 
    \begin{align}
        S(F,G):=\E[S(F,{Y})]=\int S(F,y)\dd G(y) \label{eq:scoring exp}
    \end{align}
    where $F,G\in \mathcal{F}$, and {$Y$} is a random variable with {distribution function $G$.} 

A natural requirement of a scoring rule is that it is proper, which has the following definition. 

    \begin{definition}[Proper scoring rule]\label{def:Proper}
    Scoring rule $S$ is proper relative to $\mathcal{F}$ if
    \begin{align*}
        S(G,G)\geq S(F,G)
    \end{align*}
    for all $F,G\in \mathcal{F}$. It is strictly proper if the {inequality} holds with equality if and only if $F=G$.
\end{definition}
Accordingly, using proper scoring rules incentivizes the forecaster to report their honest belief about the distribution. Therefore, the forecaster has no incentive to distort or strategically alter their forecast. {Definition~\ref{def:Proper} is stated under the positive score orientation, meaning that larger score values correspond to better predictive performance.} 

Two widely used scoring rules are the logarithmic score (LogS) and the continuous ranked probability score (CRPS). 
The former is defined as
\[ 
\mathrm{LogS}(F,y) 
:= \log {f}(y)\
\]
where {\(f\)} denotes the density of {\(F\)}. CRPS {is defined as
\begin{align*}
    CRPS(F,y):=-
\int_{-\infty}^{\infty}
\left(F(x)-\mathbf{1}\{y\le x\}\right)^2\,dx.
\end{align*}}
CRPS can be viewed as a special case of the Energy score and, unlike the logarithmic score, it does not require the existence of a density. 
More generally, the Energy score with parameter \(\beta \in (0,2)\) is defined as
\[ 
ES_\beta(F,y)
:=
\frac12 \E|X-X'|^\beta - \E|X-y|^\beta,
\]
where \(X,X' \sim F\) are independent random variables. The CRPS coincides with \(\mathrm{ES}_\beta\) for \(\beta = 1\). {In this paper, we present results for the Energy score, of which the CRPS is a special case. 

Both the logarithmic score and the Energy score with \(\beta\in(0,2)\) are strictly proper; see, for example, \cite{Gneit_Raft}.}

\subsection{Extreme Value Theory Framework}
A central result in Extreme Value Theory is the Fisher--Tippett--Gnedenko theorem, which states that suitably normalized maxima converge in distribution to one of three possible types: Gumbel, Fr\'echet, or Weibull. These domains of attraction correspond to fundamentally different tail behaviors in the underlying distribution. See \cite{dehaan} for more information.

The Fr\'echet domain, also referred to as the heavy-tailed domain, encompasses distributions whose survival function decays according to a power law. Such distributions have infinite right endpoints and exhibit polynomially decreasing tails, implying that extreme events occur with non-negligible probability. Formally, a distribution function $G$ is said to belong to the Fr\'echet domain of attraction if its survival function is regularly varying. That is, there exists a constant $\gamma > 0$ such that
\begin{align}
\lim_{t \to \infty} \frac{1 - G(tx)}{1 - G(t)} = x^{-1/\gamma}, \quad x > 0. 
\label{eq:RV}
\end{align}
The parameter $\gamma$, called the \emph{tail index}, quantifies the heaviness of the tail. When \eqref{eq:RV} holds, we say that the survival function $1-G$ is regularly varying with index $-1/\gamma$, and write $1-G \in \mathcal{R}_{-1/\gamma}$. 

An equivalent characterization can be expressed in terms of the tail quantile function
\[
U(x) = \inf\{ y \in \mathbb{R} : G(y) \ge 1 - 1/x \}, \qquad x > 1.
\]
If $1-G \in \mathcal{R}_{-1/\gamma}$, then the associated tail quantile function $U$ is regularly varying with index $\gamma$, that is, $U\in \mathcal{R}_{\gamma}$

In this paper we let $(Y_1,Y_2, \ldots , Y_n)$ be an i.i.d.\ sample with common distribution function $G$. Throughout, we assume that $1-G\in \mathcal{R}_{-1/\gamma_G}$. This implies that $Y$ is in the Fréchet domain with tail index $\gamma_G$. We now introduce the tail counterpart of $G$, which is given by $$G^t(x)=\frac{G(tx)-G(t)}{1-G(t)}\quad \text{for $t\geq 1$.}$$ We denote $G^\circ$ as the limit distribution of $G^t$ and we call it the \textit{tail distribution}. $Y^\circ$ denotes a random variable with distribution $G^\circ$. As $Y$ is in the Fréchet domain, we obtain
$$1-G^\circ(x)=\lim_{t\to \infty}1-G^t(x)=x^{-1/\gamma_G},$$ i.e., $G^t$ converges to a Pareto distribution with parameter $1/\gamma_G$.

\section{Methodology for Tail-Model Ranking and Estimation} \label{sec:method}

{The purpose of this section is to turn the tail comparison described in the
introduction into a computable score-based procedure. Rather than evaluating
predictive distributions on the original sample space, we evaluate the score on
normalized exceedances above a high threshold. The distribution used in the score
is therefore a tail object associated with the predictive model. In the main
implementation, this tail object is the Pareto limit distribution implied by the
model, so the comparison directly targets differences in tail indices. More
generally, the same framework can also be applied to finite-threshold tail
counterparts when the aim is to assess tail fit above a fixed high threshold; see
Remark~3.1.}

Formally, for $F\in \mathcal{F}$ and any proper scoring rule, $S$, we have

\begin{align*}
    \E[S(G^\circ, Y^\circ)]\geq \E[S(F, Y^\circ)],
\end{align*}

where $Y^\circ \sim G^\circ$. If $S$ is a strictly proper scoring rule, we have equality if and only if $F=G^\circ$. As $G^\circ$ is a Pareto distribution it is entirely described by the tail index, $\gamma_G$, and hence
\begin{align*}
    \E[S(G^\circ, Y^\circ)]=\E[S(F, Y^\circ)]
\end{align*}
if and only if $F$ is a {Pareto distribution} with the same tail index as $G^\circ$.

To connect this population comparison with sample quantities, note that
\[
(V^t_{k,1}, V^t_{k,2}, \ldots , V^t_{k,k})\stackrel{d}{=} \left(\frac{Y_{n,n-k+1}}{Y_{n,n-k}}, \frac{Y_{n,n-k+2}}{Y_{n,n-k}}, \ldots , \frac{Y_{n,n}}{Y_{n,n-k}}\right)\Bigg| Y_{n,n-k}=t,
\]
where $(V_i^t)_{i=1}^k$ are i.i.d.\ random variables with distribution $G^t$ and {$Y_{n,1} \leq Y_{n,2} \leq \cdots \leq Y_{n,n}$ are the order statistics}. As $G^t\to G^\circ$ when $t\to \infty$, a natural estimator for $\E[S(F, Y^\circ)]$ is 
\begin{align}
    S_k(F):=\frac{1}{k}\sum_{i=1}^k S\left(F, \frac{Y_{n,n-i+1}}{Y_{n,n-k}}\right). \label{eq: estimator}
\end{align}
Consequently, this introduces an additional tuning parameter, $k$. It should be chosen small enough for $G^t$ to provide a good approximation to $G^\circ$, yet large enough to ensure that the reduced sample remains sufficiently large for reliable inference. The EVT literature offers several methods for selecting $k$, most of which are based on identifying a range of values over which the estimator exhibits stability. In the simulation studies, we investigate this stability behavior.

In practice, we implement the procedure in the following way. First, choose a candidate set \(\{F_1,\ldots,F_m\}\) and a grid \(\mathcal{K}\) of values for \(k\). Next, compute \(S_k(F_j)\) for each \(j\) and \(k\in\mathcal{K}\). Then identify a stability range \(\mathcal{K}_{\mathrm{stab}}\subseteq \mathcal{K}\), where the ordering is not driven by erratic local fluctuations. Finally, report the ranking based on both pointwise curves \(k\mapsto S_k(F_j)\) and the average score over the stable range,
\[
\bar S(F_j):=\frac{1}{|\mathcal{K}_{\mathrm{stab}}|}\sum_{k\in\mathcal{K}_{\mathrm{stab}}} S_k(F_j).
\]
This separates the statistical comparison (score level) from the threshold-sensitivity diagnostic (stability over \(k\)).

\begin{remark} \label{rem:gg}

An important question is how to choose the distribution \(F\) used in the score.
There are two natural choices, corresponding to different inferential targets.

{The first choice is to use the tail limit distribution of \(F\), that is, the
Pareto distribution with tail index \(\gamma\) implied by \(F\).} This is appropriate when
the object of interest is the asymptotic tail index itself. In this formulation, two
predictive distributions with the same tail index but different finite-threshold
behaviour or different second-order properties have the same Pareto limit and are
therefore asymptotically indistinguishable under the score.

The second choice is to use the finite-threshold tail counterpart \(F^t\), for instance
with \(t=Y_{n,n-k}\). This retains information about the predictive distribution above
a high but finite threshold and is therefore more appropriate when the objective is
to assess finite-threshold tail fit. In this case, the score evaluates the conditional
distribution of normalized exceedances above \(t\), rather than only the limiting
Pareto approximation.

In both cases, propriety should be understood relative to the transformed
tail object being scored, not to the full predictive distribution on the original sample
space. 
\end{remark}

\begin{remark} 
    A practical issue is the choice of scoring rule. Two of the most widely used scoring rules in the literature are $LogS$ and $CRPS$. Notably, the $LogS$ requires the predictive distribution to admit a density. This is unproblematic when $F$ is specified as a Pareto distribution, for which the density is available in closed form. However, if one instead works with $F^t$, evaluation of the corresponding density may be unavailable or computationally cumbersome. In such situations, the $CRPS$ provides a convenient alternative, since it can be computed directly from the distribution function and therefore does not require density evaluation. 
\end{remark}

{The same framework can also be used for tail-index estimation. In this case, we restrict attention to Pareto predictive distributions, since
the objective is to estimate the tail index rather than to compare general
predictive distributions. We propose a
grid-search version of the estimator based on the empirical score in \eqref{eq: estimator}.
Let {\(\Gamma := [\gamma_L,\gamma_U]\), with \(0<\gamma_L<\gamma_U\)},
and for each \(\gamma \in \Gamma\), let \(F_\gamma\) denote the Pareto
distribution with tail index \(\gamma\). For a given \(k\), define
\begin{align*}
\widehat{\gamma}_{k}(S):
=
\argmax_{\gamma\in\Gamma}
\frac{1}{k}\sum_{i=1}^{k}
S\left(
F_\gamma,
\frac{Y_{n,n-i+1}}{Y_{n,n-k}}
\right).
\end{align*}
}
{In numerical implementations, the maximization over \(\Gamma\) may be carried
out over a finite grid.}
\section{Asymptotic Results} \label{sec:math}

In this section, we present consistency and asymptotic normality results for
scoring rules based on the normalized order statistics introduced in
Section~\ref{sec:method}. We also establish consistency for a tail-index
estimator obtained by optimizing the empirical score.

To support the use of scoring rules for ranking tail models, it is important to establish their large-sample properties. Consistency is a fundamental requirement for any method aimed at tail evaluation, as it ensures that the correct model is identified in the limit. Although asymptotic normality is not strictly necessary, it is highly useful since it makes it possible to assess whether differences in scores between competing models are statistically significant.

\subsection{Scoring with Normalized Upper Order Statistics}

We begin with consistency. Consistency and asymptotic normality follow directly from \cite{segers} (Theorem~2.1 and Theorem~4.5). Here the assumptions and theorem are restated in our notation. 

\begin{assumption} \label{ass:segers_ass}
    Let $S:\mathcal{F} \times [1,\infty) \to \mathbb{R}$ be an a.e. continuous function for $F\in \mathcal{F}$, and such that $|S(F,y)|\leq Ay^{(1-\delta)/\gamma_{G}}$ for some $A>0$ and $\delta\in (0,1)$.
\end{assumption}

This assumption is mild and does not impose a substantial restriction, since commonly used scoring rules are typically continuous. {The domination condition provides the necessary tail control.}

\begin{theorem}
\label{thm:consistency}
Let Assumption~\ref{ass:segers_ass} hold. In addition, assume that $k,n \to \infty$ and $k/n \to 0$. Then
\begin{align*}
S_k(F)
\;\rightarrow\;
\E\!\left[S(F, Y^\circ)\right]
\end{align*}
in probability, where $Y^\circ$ has distribution $G^\circ$.
\end{theorem}

In practical terms, this result ensures that empirical tail scores converge to their population targets when \(k\) is chosen in the standard EVT regime.

We next turn to asymptotic normality, which follows as a direct consequence of Theorem~4.5 in \cite{segers}. To establish this result, it is necessary to control the growth rate of $k$. We therefore introduce the following notation. Let $\xi \in \mathbb{R}$, define the function
$h_\xi : (0,\infty) \to \mathbb{R}$ by
\[
h_\xi(x)
= \int_1^x y^{\xi-1}\,dy
=
\begin{cases}
\displaystyle \frac{x^\xi - 1}{\xi}, & \text{if } \xi \neq 0, \\[1ex]
\log x, & \text{if } \xi = 0.
\end{cases}
\]

We may now restate Theorem~4.5 of \cite{segers} in our notation.
\begin{theorem}\label{thm:normality}
Let $G$ be a {distribution function} with tail quantile function $U$ satisfying
\begin{equation}
\lim_{t \to \infty} \frac{U(tx)/U(t) - x^{\gamma_G}}{a(t)}
= \pm x^{\gamma_G} h_{\xi}(x), \qquad x \ge 1,
\end{equation}
for some $\gamma_G > 0$, $\xi \le 0$ and $a \in \mathcal{R}_{\xi}$ satisfying
$a(t) \to 0$ as $t \to \infty$.
Let $F\in\mathcal{P}$ be a fixed candidate distribution used in the score comparison.
Let $S :\mathcal{F} \times [1,\infty) \to \mathbb{R}$ be an absolutely continuous function for $F\in \mathcal{F}$ with
$\frac{\partial S}{\partial y}$ such that 
\begin{align*}
\Bigg|\frac{\partial S}{\partial y}\Bigg| \le A y^{\rho-1}, \qquad y \ge 1,
\end{align*}
for some $A > 0$ and some $\rho < (2\gamma_G)^{-1}$.
If $k$ satisfies
$\sqrt{k}\,a{\left(\frac{n+1}{k+1}\right)} \to 0$ when $k,n \to \infty$ and $k/n \to 0$, then
\begin{equation} \label{eq:normality}
\sqrt{k}\Biggl(S_k(F) - \amsmathbb{E}[S(F,Y^\circ)]\Biggr)
\xrightarrow{d} \mathcal{N}\!\left(0,\operatorname{Var}(S(F,Y^\circ))\right).
\end{equation}
\end{theorem}

{The derivative condition is a regularity condition ensuring that the scoring rule behaves smoothly in the tail. As $y\mapsto S(F,y)$ is assumed absolute continuous, the condition also ensures that the variance exists.} This normal approximation provides a direct basis for uncertainty quantification when comparing score levels across competing tail models.

The following corollary shows under which conditions Theorem~\ref{thm:consistency} and Theorem~\ref{thm:normality} are applicable for LogS and ES.  

\begin{corollary}\label{col:cons}
Let \(F_\gamma\) denote a Pareto distribution with tail index \(\gamma\). Then:
\begin{enumerate}
    \item The logarithmic score satisfies the assumptions of Theorems~\ref{thm:consistency}
    and~\ref{thm:normality} for all \(\gamma>0\). Moreover,
    \[
    \operatorname{Var}(\mathrm{LogS}(F_\gamma,Y^\circ))
    =
    \left(1+\frac{1}{\gamma}\right)^2\gamma_G^2 .
    \]

    \item The Energy score satisfies the assumptions of Theorem~\ref{thm:consistency}
    whenever
    \(
    \beta < \min\left\{\frac{1}{\gamma},\frac{1}{\gamma_G}\right\}.
    \)
    It satisfies the assumptions of Theorem~\ref{thm:normality} whenever, in addition,
    \(
    \beta < \frac{1}{2\gamma_G}.
    \)
    In particular, for \(\beta=1\), corresponding to the CRPS,
    \[
    \operatorname{Var}(\mathrm{ES}_1(F_\gamma,Y^\circ))
    =
    \frac{1}{1-2\gamma_G}
    -\frac{2a}{1-(1+p)\gamma_G}
    +\frac{a^2}{1-2p\gamma_G}
    -
    \left(
    \frac{1}{1-\gamma_G}
    -\frac{a}{1-p\gamma_G}
    \right)^2,
    \]
    where \(a=2\gamma/(\gamma-1)\) and \(p=1-1/\gamma\), whenever the displayed
    moments are finite.
\end{enumerate}
\end{corollary}

\begin{proof}
    We verify the conditions separately for LogS and \(\mathrm{ES}_\beta\).

    First consider LogS. If we denote the density of $F_\gamma$ as $f_\gamma$, we have
    \begin{align*}
        \mathrm{LogS}(F_\gamma,y)=\log(f_\gamma(y))=\log(1/\gamma)-\left(\frac{1}{\gamma}+1\right)\log(y).
    \end{align*}     
    $\mathrm{LogS}(F_\gamma,y)$  is absolute continuous for $y>1$ and for any $\delta\in (0,1)$ it is dominated by $Ay^{(1-\delta)/\gamma_G}$, where $A$ is some constant. Further, we have
    \begin{align*}
        \frac{\partial \mathrm{LogS}}{\partial y}=-\left(\frac{1}{\gamma}+1\right)y^{-1}
    \end{align*}
    {For any \(\rho \in (0,(2\gamma_G)^{-1})\), this derivative is dominated by \(Ay^{\rho-1}\).}

    Next consider \(\mathrm{ES}_\beta\):
    \begin{align*}
\mathrm{ES}_\beta(F_\gamma,y)
&= {\tfrac{1}{2} \int_1^\infty\int_1^\infty |x - x'|^\beta \, \dd F_\gamma(x) \dd F_\gamma(x')-\int_1^\infty |x - y|^\beta \, \dd F_\gamma(x)}.
\end{align*}
For $y\geq 1$, this expression is absolutely continuous. We first bound the first integral. {For $\beta<\frac{1}{\gamma}$, we have that}
\begin{align}
    \int_1^\infty |x - y|^\beta \, dF(x)&\leq A_1 \left(\int_1^y(y-1)^\beta x^{-1/\gamma -1}\dd x + \int_y^\infty x^{\beta-1/\gamma-1}\dd x\right) \label{eq:ES1}\\
    &{\leq A_2 (y^{\beta}+y^{\beta-1/\gamma-1})}\nonumber \\
    &{\le A_3y^\beta} \nonumber
\end{align}
for some constants $A_1,A_2,{A_3>0}$. { Hence we can conclude that \eqref{eq:ES1} is finite. Furthermore, for $\beta<\frac{1}{\gamma_G}$ there exist some $\delta\in(0,1)$ and $A$ such that it is dominated by $Ay^{(1-\delta)/\gamma_G}$}. Next, for the second integral, it is enough to verify finiteness since it does not depend on $y$:
\begin{align}
&\int_1^\infty\int_1^\infty |x - x'|^\beta \, \dd F_\gamma(x) \dd F_\gamma(x')\nonumber\\&\quad\leq \int_1^\infty x^{\beta-2/\gamma-1}-\left(A_4 + \int_1^\infty x'^{\beta-1/\gamma-1}\dd x'\right)x^{-1/\gamma-1}\dd x \label{eq:CRPS_finite}
\end{align}
for some constant $A_4>0$. The inner integral is finite if and only if $\beta < \frac{1}{\gamma}$. Moreover, note that
\begin{align*}
    \int_1^\infty x^{\beta - 2/\gamma - 1} - A_5 x^{-1/\gamma - 1} \, \dd x,
\end{align*}
where $A_5$ is a constant, is finite if and only if $\beta \leq \frac{2}{\gamma}$. Consequently, \eqref{eq:CRPS_finite} is finite if and only if $\beta < \frac{1}{\gamma}$. {Hence, we can conclude that $\mathrm{ES}_\beta$ satisfies Assumption \ref{ass:segers_ass} for \(\beta<\min\left\{1/\gamma,1/\gamma_G\right\}\).} Lastly, we have
    \begin{align*}
        \left|\frac{\partial \mathrm{ES}_\beta}{\partial y}\right|\leq {\beta \int_1^\infty |x-y|^{\beta-1}\dd F(x)} \leq A_6y^{{\beta-1}},
    \end{align*}
where $A_6$ is a constant. {If \(\beta < (2\gamma_G)^{-1}\), choose
\(\rho\) with \(\beta < \rho < (2\gamma_G)^{-1}\). Since \(y^{\beta-1} \leq
y^{\rho-1}\) for \(y \geq 1\), the derivative condition is satisfied.}

    Therefore both scoring rules satisfy the required conditions, which proves the claim.
\end{proof}

{This shows that the derivative condition mainly affects $\mathrm{ES}_\beta$, where it imposes an additional restriction on the score parameter $\beta$. For $\mathrm{LogS}$, the condition is automatically satisfied and therefore imposes no further restriction.}

\begin{remark} 
We highlight two practical aspects that are relevant for applying the asymptotic results.

\begin{enumerate}
    \item[(i)] The condition for \(ES_\beta\) depends on the true tail index \(\gamma_G\), which is
unknown in practice. This limitation should not be interpreted as excluding the Energy
score for practical applications. For the common CRPS case \(\beta=1\), the consistency condition is satisfied when the relevant tail indices are below one. This is a natural range in many insurance applications, since \(\gamma_G<1\) corresponds to finite mean claim sizes. The asymptotic normality condition is stronger and requires \(\gamma_G<1/2\), corresponding to finite variance; this is more restrictive, but still not uncommon in insurance applications. The logarithmic score avoids this restriction and therefore
provides a convenient default for tail-model ranking.

    \item[(ii)] The asymptotic variance in \eqref{eq:normality} depends on the unknown tail index
\(\gamma_G\). When an analytic variance formula is used, this parameter must
therefore be replaced by a plug-in value, for instance a tail-index estimate.
Alternatively, the variance may be estimated empirically.
\end{enumerate}
\end{remark}

\begin{remark}

The results in Theorems~\ref{thm:consistency} and~\ref{thm:normality}
are stated for an i.i.d. sample with a common tail distribution. Hence, they do not
cover covariate-dependent or otherwise non-identically distributed predictive
distributions. Extending the asymptotic theory to non-identically distributed predictive
distributions would require additional assumptions on the conditional tail
behavior, such as uniform regular variation or a common
tail index across covariate values; see, for example,
\cite{BladtOhlenschlaeger2025}. The effect of
heterogeneous scaling is investigated numerically in Section~\ref{subsec:scales}, but a detailed
theoretical treatment of such extensions is beyond the scope of the present paper.
\end{remark}

Taken together, Theorems~\ref{thm:consistency}--\ref{thm:normality} and Corollary~\ref{col:cons}: compute empirical tail scores over a range of (k), compare candidates by the score level, and use the normal approximation to assess whether observed score differences are practically meaningful.

\subsection{Tail Index Estimation via Score Optimization}

Beyond their use for ranking predictive distributions, we also investigate whether scoring rules can be employed for estimating the tail index. In this section, we assume that the predictive distribution is Pareto, denoted by $F_\gamma$, where $\gamma$ is the candidate tail index. Let 
\(\Gamma := [\gamma_L,\gamma_U]\), with {\(0<\gamma_L<\gamma_G<\gamma_U<\infty\)}. For shorthand, define
\[
S_k(\gamma):=\frac{1}{k}\sum_{i=1}^k S\left(F_\gamma, \frac{Y_{n,n-i+1}}{Y_{n,n-k}}\right).
\]
We are now in a position to state the result.
\begin{theorem} \label{thm:argmin}
	 Assume that $S$ is continuous on \(\Gamma\times[1,\infty)\) and a strictly proper scoring rule. Further assume Assumption \ref{ass:segers_ass} and that for $\gamma\in \Gamma$, {there exists
a measurable function \(D:[1,\infty)\to[0,\infty)\) such that
\begin{align}
\sup_{\gamma\in\Gamma}|S(F_\gamma,z)|\le D(z),
\qquad z\ge 1, \label{eq:uni_bound}
\end{align}
with $D(z)\le A z^{(1-\delta)/\gamma_G}$
for some \(A>0\) and \(\delta\in(0,1)\). In particular,
\(E[|S(F_\gamma,Y^\circ)|]<\infty\) for all \(\gamma\in\Gamma\).} Let $k\to\infty$ with $n/k\to\infty$. Then
    \begin{align}
        \label{eq:argmin_theo}\hat{\gamma}_k(S)=\argmax_{\gamma\in \Gamma} \frac{1}{k}\sum_{i=1}^k S\left(F_\gamma, \frac{Y_{n,n-i+1}}{Y_{n,n-k}}\right) \xrightarrow{\pr} \gamma_G.
    \end{align}
\end{theorem}

We first establish the following uniform convergence lemma.
\begin{lemma}
    {Assume the same conditions as in Theorem \ref{thm:argmin}. Then, as \(k\to\infty\) and \(n/k\to\infty\),
\[
\sup_{\gamma\in\Gamma}
\left|
S_k(\gamma)-E[S(F_\gamma,Y^\circ)]
\right|
\xrightarrow{P}0 .
\]}
    \label{lemma:sup}
\end{lemma}

\begin{proof}[Proof of Lemma \ref{lemma:sup}]
    Let
    \begin{align*}
        u(z,\theta,\tau):= \sup_{\gamma{\in \Gamma}:|\gamma-\theta|\leq \tau}|S(F_\gamma,z)-S(F_\theta,z)|.
    \end{align*}
    {Firstly, we argue that $z\mapsto u(z,\theta,\tau)$ is measurable. For fixed \(\theta\in\Gamma\) and \(\tau>0\), let
    \[
    \Gamma_{\theta,\tau}:=\{\gamma\in\Gamma:|\gamma-\theta|\le \tau\}.
    \]
    Since \(\Gamma_{\theta,\tau}\) is compact and
    \(\gamma\mapsto S(F_\gamma,z)\) is continuous for each fixed \(z\), the supremum over \(\Gamma_{\theta,\tau}\) may be taken over any countable dense subset \(Q_{\theta,\tau}\subset \Gamma_{\theta,\tau}\). Hence
\[
u(z,\theta,\tau)
=
\sup_{\gamma\in Q_{\theta,\tau}}
|S(F_\gamma,z)-S(F_\theta,z)|.
\]
For each fixed \(\gamma\), the map \(z\mapsto S(F_\gamma,z)\) is measurable, since \(S\) is a scoring rule. Therefore \(u(\cdot,\theta,\tau)\) is the pointwise supremum of countably many measurable functions, and is consequently
measurable.}

    {By the continuity assumption of $S$, we have that
    $u(Y^\circ,\theta,\tau)\to 0$ for $\tau \to 0$ from above}.  Since $S$ is dominated by $D$, by the {triangle-inequality we have that $u(Y^\circ,\theta,\tau)\leq 2D(Y^\circ)$}. {Furthermore, since \(Y^\circ\) is Pareto with tail index \(\gamma_G\), the bound
\(D(z)\leq A z^{(1-\delta)/\gamma_G}\) implies
\(E[D(Y^\circ)]<\infty\).} Hence {for a specific $\theta$}, we have by the dominated convergence theorem, $\E [u(Y^\circ,\theta,\tau)]\to 0$ as $\tau\to 0$.

    Now use compactness of the parameter space to cover it by finitely many balls \(B(\theta_m,\tau_m)\), \(m=1,\dots,M\), {where $\theta_m\in \Gamma$}. By choosing \(M\) sufficiently large, we can ensure that for each ball
\begin{align}
\mu_m:=\E\big[u(Y^\circ,\theta_m,\tau_m)\big] < \varepsilon \label{eq:mu}
\end{align}
    for any given \(\varepsilon > 0\). For a given $\gamma\in B_m$ we have
    \begin{align}
        |S_k(\gamma)-\E[S(F_\gamma,Y^\circ)]| \leq & |S_k(\gamma)-S_k(\theta_m)| \label{eq:in1}\\
        & +|S_k(\theta_m)-\E[S(F_{\theta_m},Y^\circ)]| \label{eq:in2}\\
        & + |\E[S(F_{\theta_m},Y^\circ)]-\E[S(F_\gamma,Y^\circ)]| \label{eq:in3}.
    \end{align}
    For \eqref{eq:in1}, note that
    \begin{align*}
        |S_k(\gamma)-S_k(\theta_m)| \leq \left|\frac{1}{k}\sum_{i=1}^k u\left(\frac{Y_{n,n-k+i}}{Y_{n,n-k}},\theta_m,\tau_m\right)-\mu_m\right|+\mu_m.
    \end{align*}
    {As argued earlier, we have 
    \begin{align*}
        u(z,\theta_m,\tau_m) \leq 2Az^{(1-\delta)/\gamma_G},
    \end{align*}
    and by continuity of \(S\) and compactness of
\(\Gamma_{\theta,\tau}\), we have that $z\mapsto u(z,\theta,\tau)$ is continuous. }  {Hence, we can use Theorem \ref{thm:consistency} on $z\mapsto u(z,\theta_m,\tau_m)$ to conclude that 
    \begin{align*}
        \frac{1}{k}\sum_{i=1}^k u\left(\frac{Y_{n,n-k+i}}{Y_{n,n-k}},\theta_m,\tau_m\right) \to \mu_m,
    \end{align*}
    in probability.}
    Since $\mu_m$ can be made arbitrarily small by selecting a large $M$, \eqref{eq:in1} can be made arbitrarily small. Likewise, \eqref{eq:in2} converges to 0 by Theorem \ref{thm:consistency}, {while \eqref{eq:in3} can be made arbitrarily small using
    \begin{align*}
        |\E[S(F_{\theta_m},Y^\circ)]-\E[S(F_\gamma,Y^\circ)]|\leq \E\big[u(Y^\circ,\theta_m,\tau_m)\big]
    \end{align*}
    and \eqref{eq:mu}}.

    Altogether, this shows that for any \(\gamma \in B(\theta_m,\tau_m)\), the quantity
\[
|S_k(\gamma)-\E[S(F_\gamma,Y^\circ)]|
\]
can be made arbitrarily small by choosing \(M\) sufficiently large. Since this holds for each of the finitely many balls in the cover, it also holds after taking the supremum over \(\gamma \in \Gamma\) for \(M\) sufficiently large, which concludes the proof.
\end{proof}

\begin{proof}[Proof of Theorem \ref{thm:argmin}]
    By Lemma~\ref{lemma:sup}, we have
    \begin{align} \label{eq:assum}
        \sup_{\gamma \in \Gamma}|S_k(\gamma)-\E[S(F_\gamma,Y^\circ)]|\xrightarrow{\pr}0
    \end{align}
    as $k,n/k\to\infty$. Using \eqref{eq:assum}, we get that
    \begin{align}\label{eq:SS}
        S_k(\hat{\gamma})\geq S_k(\gamma_G)\geq \E[S(F_{\gamma_G},Y^\circ)]-o_\pr(1).
    \end{align}
    Using \eqref{eq:SS} and \eqref{eq:assum}, we get that
    \begin{align}
        \E[S(F_{\gamma_G},Y^\circ)] - \E[S(F_{\hat{\gamma}},Y^\circ)]&\leq S_k(\hat{\gamma})+o_\pr(1)-\E[S(F_{\hat{\gamma}},Y^\circ)] \\
        & \leq \sup_{\gamma \in \Gamma} |S_k(\gamma)-\E[S(F_\gamma,Y^\circ)]|+o_\pr(1) \xrightarrow{\pr}0 \label{eq:conv0}
    \end{align}
    as $k,n/k\to \infty$.
    {Since \(Y^\circ\sim F_{\gamma_G}\), strict propriety of \(S\) implies that
\[
\mathbb{E}[S(F_\gamma,Y^\circ)]
<
\mathbb{E}[S(F_{\gamma_G},Y^\circ)]
\]
whenever \(F_\gamma\neq F_{\gamma_G}\). By identifiability of the Pareto family,
this holds for every \(\gamma\neq\gamma_G\). Moreover, the map $\gamma\mapsto \mathbb{E}[S(F_\gamma,Y^\circ)]$ 
is continuous by dominated convergence. Hence, for every \(\varepsilon>0\), the compact set $\{\gamma\in\Gamma:|\gamma-\gamma_G|\geq\varepsilon\}$
admits a maximizer. Since this maximizer differs from \(\gamma_G\), we obtain
\[
\sup_{\gamma:|\gamma-\gamma_G|\geq\varepsilon}
\mathbb{E}[S(F_\gamma,Y^\circ)]
<
\mathbb{E}[S(F_{\gamma_G},Y^\circ)].
\]}
    Hence, for every $\varepsilon>0$, there exists an $\eta>0$ such that
    \begin{align*}
        \E[S(F_{\hat{\gamma}},Y^\circ)]<\E[S(F_{\gamma_G},Y^\circ)]-\eta \quad \text{when } |\hat\gamma-\gamma_G|\geq \varepsilon .
    \end{align*}
    Therefore,
    \begin{align*}
        \{|\hat{\gamma}-\gamma_G|\geq \varepsilon\} \subset\{\E[S(F_{\hat{\gamma}},Y^\circ)]<\E[S(F_{\gamma_G},Y^\circ)]-\eta\},
    \end{align*}
    and as the right-hand side converges to 0 in probability by \eqref{eq:conv0}, the left-hand side also converges to 0.
\end{proof}

 {When the empirical maximizer lies on the boundary of \(\Gamma\), the estimate is
a boundary optimum induced by the restricted parameter space. It may therefore
reflect that the unconstrained empirical optimum lies outside \(\Gamma\), so such
a solution should be interpreted with care rather than as direct evidence that the
true tail index coincides with the boundary point. For the logarithmic score, this
boundary effect can be described explicitly: the unconstrained estimator is the
Hill estimator, while the constrained estimator over \(\Gamma\) is its projection
onto \(\Gamma\).} 
\begin{remark}\label{col:LogS} 
     {Let \(F_\gamma^\circ\) denote the Pareto distribution on
\([1,\infty)\) with tail index \(\gamma\), and let
\[
\widehat\gamma_H
=
\frac{1}{k}
\sum_{i=1}^k
\log\left(
\frac{Y_{n,n-i+1}}{Y_{n,n-k}}
\right)
\]
denote the Hill estimator. Then the logarithmic-score estimator over the
parameter space \(\Gamma\) is the projection of the Hill estimator
onto this interval, i.e.
\[
\argmax_{\gamma\in\Gamma}
\frac{1}{k}
\sum_{i=1}^k
\log\left[
f_\gamma^\circ\left(
\frac{Y_{n,n-i+1}}{Y_{n,n-k}}
\right)
\right]
=
\Pi_{\Gamma}(\widehat\gamma_H),
\]
where $\Pi_{\Gamma}(x)
=
\min\{\gamma_U,\max\{\gamma_L,x\}\}$. 
In particular, if \(\widehat\gamma_H\in\Gamma\), then the
logarithmic-score estimator coincides with the Hill estimator.}
\end{remark}

\begin{remark}\label{rem:var_Logs}
    \cite{segers} shows that, under mild conditions on \(h\), the estimator of the form
\[ 
\frac{1}{k}\sum_{i=1}^k h\!\left(\frac{Y_{n,n-i+1}}{Y_{n,n-k}}\right)
\]
attains the smallest asymptotic variance when \(h=\log\). This suggests that it is inherently difficult to construct a scoring-rule-based estimator \(\hat{\gamma}_k\) that outperforms the Hill estimator in terms of asymptotic efficiency.   
\end{remark}

{The following corollary verifies the assumptions of Theorem~\ref{thm:argmin} for LogS and
\(ES_\beta\).}
\begin{corollary}\label{cor:argmax}
    The assumptions of Theorem \ref{thm:argmin} are satisfied when
    \begin{enumerate}
        \item using \(\mathrm{LogS}\) for all $\gamma>0$;
        \item using $ES_\beta$ with $\beta<\frac{1}{\gamma_U}$.
    \end{enumerate}
\end{corollary}

\begin{proof}
    For \(\mathrm{LogS}\), the map \((\gamma,z)\mapsto \mathrm{LogS}(F_\gamma,z)\) is continuous on \(\Gamma\times[1,\infty)\), and \(\mathrm{LogS}\) is strictly proper. Moreover, the calculations in the proof of Corollary~\ref{col:cons} show that \[\sup_{\gamma\in\Gamma}|\mathrm{LogS}(F_\gamma,z)|\leq Az^{(1-\delta)/\gamma_G}.\] Hence the assumptions of Theorem~\ref{thm:argmin} are satisfied for \(\mathrm{LogS}\).

For \(ES_\beta\), the map \((\gamma,z)\mapsto ES_\beta(F_\gamma,z)\) is continuous on
\(\Gamma\times [1,\infty)\), and \(ES_\beta\) is strictly proper. By the
calculations in the proof of Corollary~\ref{col:cons}, for each fixed
\(\gamma\) with \(\beta<1/\gamma\) there exists a constant \(C_\gamma>0\) such that
\[
 |ES_\beta(F_\gamma,z)|\leq C_\gamma(1+z^\beta).
\]
Since \(\beta<1/\gamma_U\), we have \(\beta<1/\gamma\) for every
\(\gamma\in\Gamma\). Moreover, it is seen from the calculations in the proof of
Corollary~\ref{col:cons} that the constant \(C_\gamma\) can be chosen as
a continuous function of \(\gamma\). Since \(\Gamma\) is compact, \(C(\gamma)\) is
bounded on \(\Gamma\). Hence there exists \(C>0\) such that
\[
 \sup_{\gamma\in\Gamma}|ES_\beta(F_\gamma,z)|\leq C(1+z^\beta).
\]
\[
|ES_\beta(F_\gamma,z)|
\leq C_\gamma\left(1+z^\beta\right).
\]
Since \(\gamma_G\in\Gamma\), the condition
\(\beta<1/\gamma_U\) implies \(\beta<1/\gamma_G\). Hence we may choose
\(\delta\in(0,1)\) such that
\[
\beta<\frac{1-\delta}{\gamma_G}.
\]
It follows that, for \(z\geq 1\),
\[
C(1+z^\beta)\leq A z^{(1-\delta)/\gamma_G}
\]
for some \(A>0\). Therefore the uniform envelope condition in
Theorem~\ref{thm:argmin} is satisfied, and the result follows.
\end{proof}

\begin{remark}
{
Corollary~\ref{cor:argmax} gives a uniform sufficient condition for the
Energy-score estimator over the parameter interval
\(\Gamma=[\gamma_L,\gamma_U]\). For \(ES_\beta\), this condition requires
\(\beta < 1/\gamma_U\). Thus, allowing larger values of \(\gamma_U\) forces smaller
choices of \(\beta\). In the CRPS case, \(\beta=1\), this corresponds to the condition
\(\gamma_U < 1\). Hence, for score-based tail-index estimation, the Energy-score
estimator is most natural when the parameter interval can be chosen so that
commonly used values of \(\beta\) remain admissible.
}
\end{remark}

        
        

\section{Simulation Study} \label{sec:sim}

This section explores our theoretical results through three simulation studies. First, we investigate the ability of \eqref{eq: estimator} to differentiate between tail indices. Next, we assess robustness when the scaling varies systematically across observations. Lastly, we examine the finite-sample behavior of {$\widehat{\gamma}_{k}(S)$} for Energy-score-based tail-index estimation, compared with the classical Hill estimator.

Across all studies, the primary diagnostic is stability of the score or estimator curve over \(k\): reliable behavior should persist over a non-negligible \(k\)-range rather than at isolated grid points.

\subsection{Ranking Tail Models: Baseline Setting}
\label{subsec:homosce}

In the baseline setting, the purpose of the simulation study is to assess how well a $LogS$ rule can identify the true tail index of a heavy-tailed distribution, and how this depends on the threshold level and the sample size. 

Samples of size $n \in \{10^3, 10^4, 10^5\}$ are generated from heavy-tailed distributions with true tail index $\gamma_G = 1$. We consider two parametric families of distributions, namely the Fr\'echet and the Burr distributions.  
The Fr\'echet distribution is specified by the distribution function
\[
F(x) = \exp\!\left(-x^{-s}\right), \qquad x > 0,
\]
where the shape parameter is set to $s = 1/\gamma_G$.  The Burr distribution is given by
\[
F(x) = 1 - \bigl(1 + x^{c}\bigr)^{-t}, \qquad x > 0,
\]
with shape parameters chosen as $c = 1/\gamma_G$ and $t = 1$. With these parameterizations, both distributions are characterized by the same tail index $\gamma_G$.

Within each sample, the threshold parameter $k$ ranges from $50$ to $n/4$ and is evaluated on an evenly spaced grid of $100$ points. 
For each value of $k$, four logarithmic scores are computed by taking $F$ to be a Pareto distribution with tail index 
\[
\gamma \in \{0.8,\,1,\,1.2,\,1.5\}.
\]
Using this setup, the empirical criterion is computed via \eqref{eq: estimator}. In Figure~\ref{fig:LS_frechet_burr}, the resulting scores of a single realization are plotted as functions of $k$, with one curve corresponding to each value of $\gamma$. 
If the scoring rule is well aligned with the true tail behavior, the candidate value $\gamma = \gamma_G = 1$ should yield the highest log score across a broad range of $k$. To formalize this comparison, for each fixed \(k\) the candidate values are ranked according to their empirical logarithmic scores. In particular, the highest-ranked value is
\[
 \argmax_{\gamma \in \{0.8,1,1.2,1.5\}} \mathrm{LogS}_k(\gamma),
\]
with the remaining candidates ordered according to their corresponding scores. For the smaller sample sizes $n=10^3$ and $n=10^4$, and in particular for the Burr model, the curve corresponding to $\gamma=1$ is not systematically above the competing curves; instead, the candidate values $\gamma=0.8$ and $\gamma=1.2$ often attain similar or larger scores over substantial ranges of $k$. By contrast, clearer signals appear in the remaining cases, especially for $n=10^5$ and for the Fréchet model, where the curve corresponding to $\gamma=1$ tends to dominate over a noticeable region of smaller $k$.
This pattern is consistent with the expected bias-variance trade-off in tail inference: with larger \(n\), a wider lower-\(k\) region supports stable discrimination of the correct tail index.
This visual impression is further supported by a Monte Carlo experiment reported in Figure~\ref{fig:LS_mean}. For each configuration, 100 independent samples are generated and the empirical logarithmic score is computed for the candidate values $\gamma \in \{0.8,1,1.2,1.5\}$ across the same range of thresholds. For each value of $k/n$, we record whether the score is maximized at the true value $\gamma=1$, and plot the resulting proportion of such occurrences across the simulations. In both the Fr\'echet and Burr cases, when \(n=10^3\) there is no clear tendency for the proportion to approach one as $k/n$ goes to 0. As the sample size increases, however, the proportion rises in both models. For Fréchet data-generating distributions this increase is particularly strong. When \(n=10^5\) the proportion is equal to one across the entire range of \(k/n\), indicating that the correct tail index is almost always selected. In the Burr case, the increase is observed only for smaller values of k/n, indicating that the correct tail index is consistently favored only further into the tail. 

{It should be noted that both the Fréchet and Burr samples are evaluated against Pareto limit distributions. Although \(\gamma=1\) gives the correct asymptotic tail index in both cases, finite-threshold exceedances are not exactly Pareto, so larger values of \(k/n\) may still reflect pre-asymptotic distributional features. Hence, it is expected that the correct tail-index signal is clearest for smaller values of \(k/n\).}
 
\begin{figure}[H]
\centering

\begin{subfigure}{0.40\textwidth}
    \includegraphics[width=\linewidth, trim= 0.3in 0.6in 0.3in 0.6in,clip]{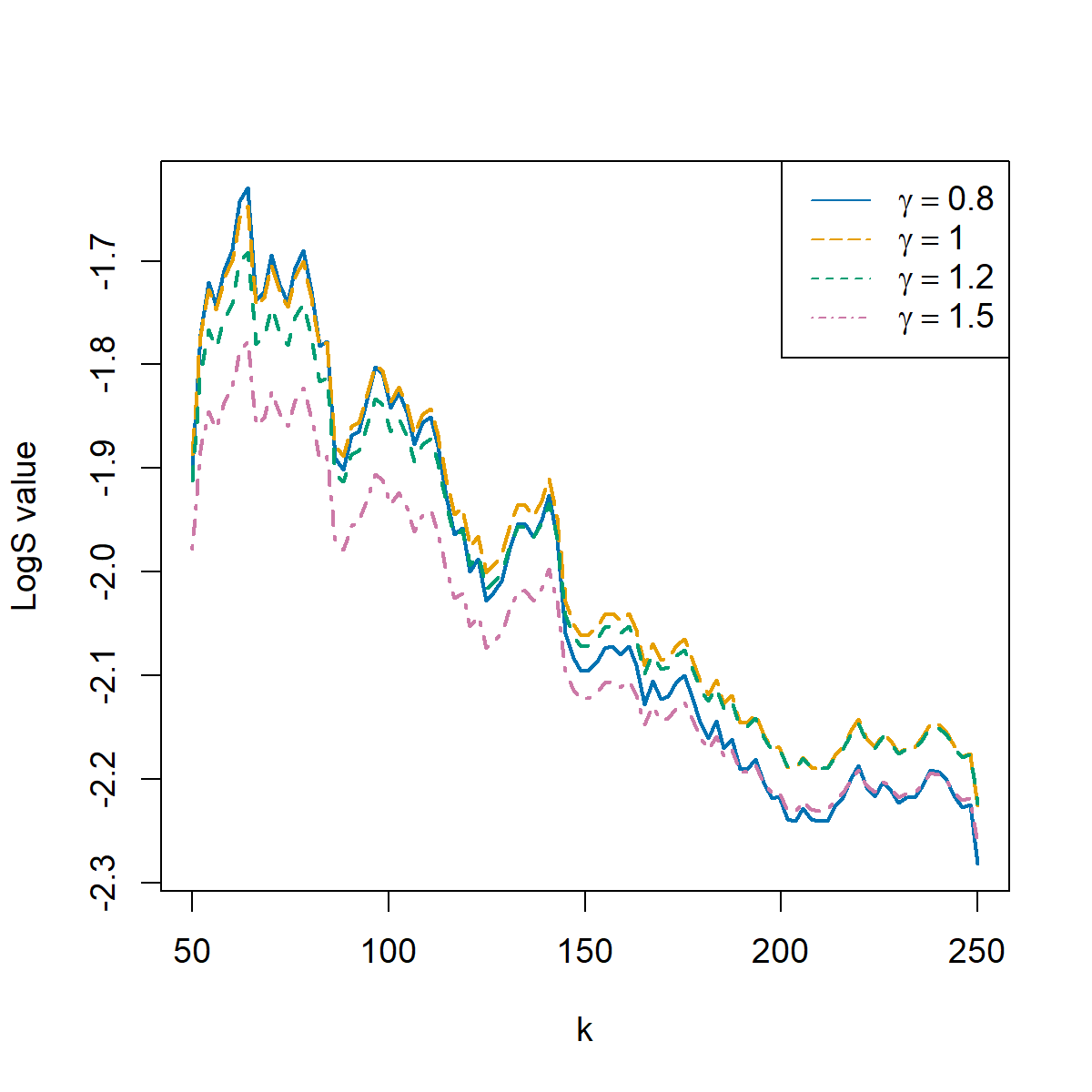}
    \caption{Fréchet DGP, $n = 10^3$}
\end{subfigure}
\begin{subfigure}{0.40\textwidth}
    \includegraphics[width=\linewidth, trim= 0.3in 0.6in 0.3in 0.6in,clip]{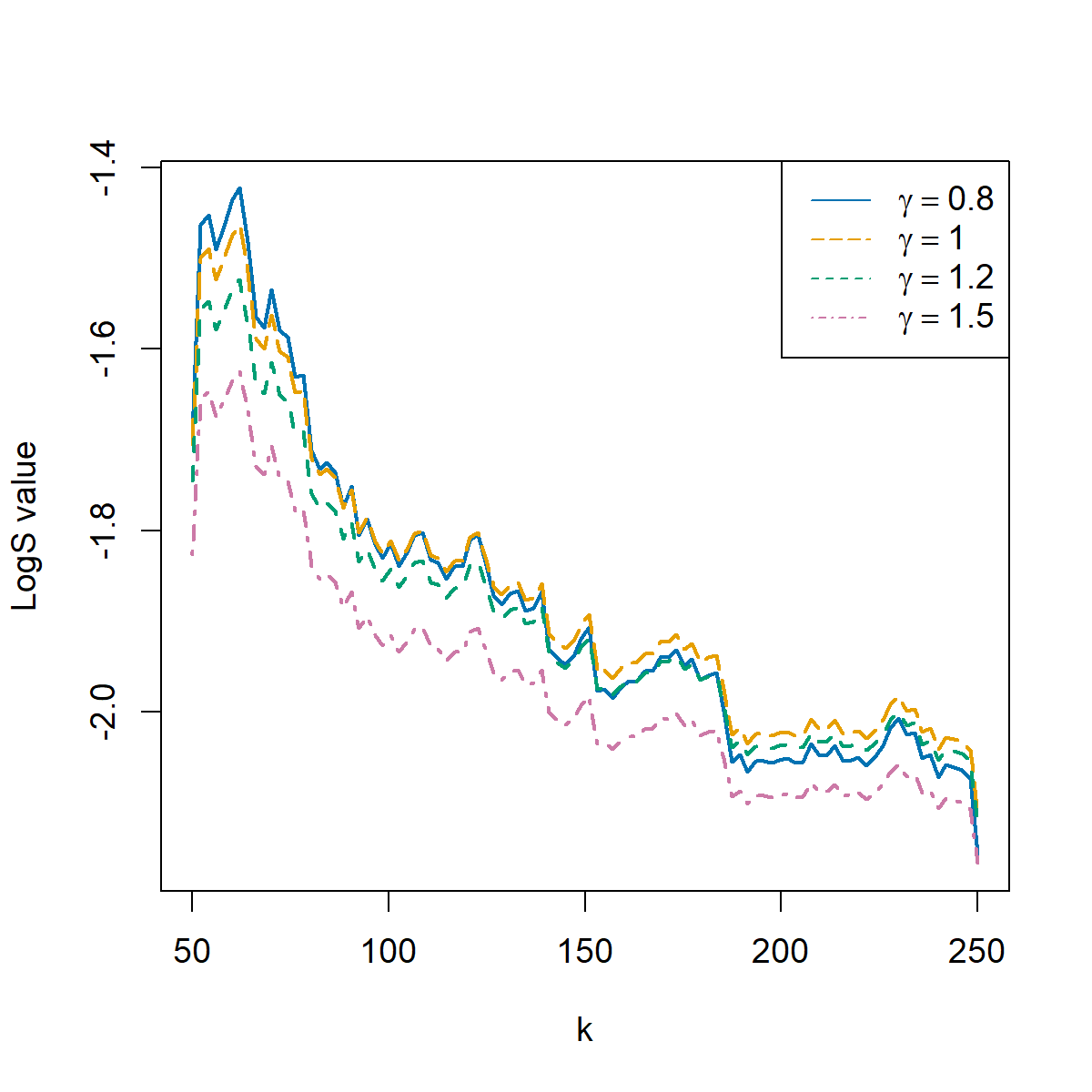}
    \caption{Burr DGP, $n = 10^3$}
\end{subfigure}\\
\begin{subfigure}{0.40\textwidth}
    \includegraphics[width=\linewidth, trim= 0.3in 0.6in 0.3in 0.6in,clip]{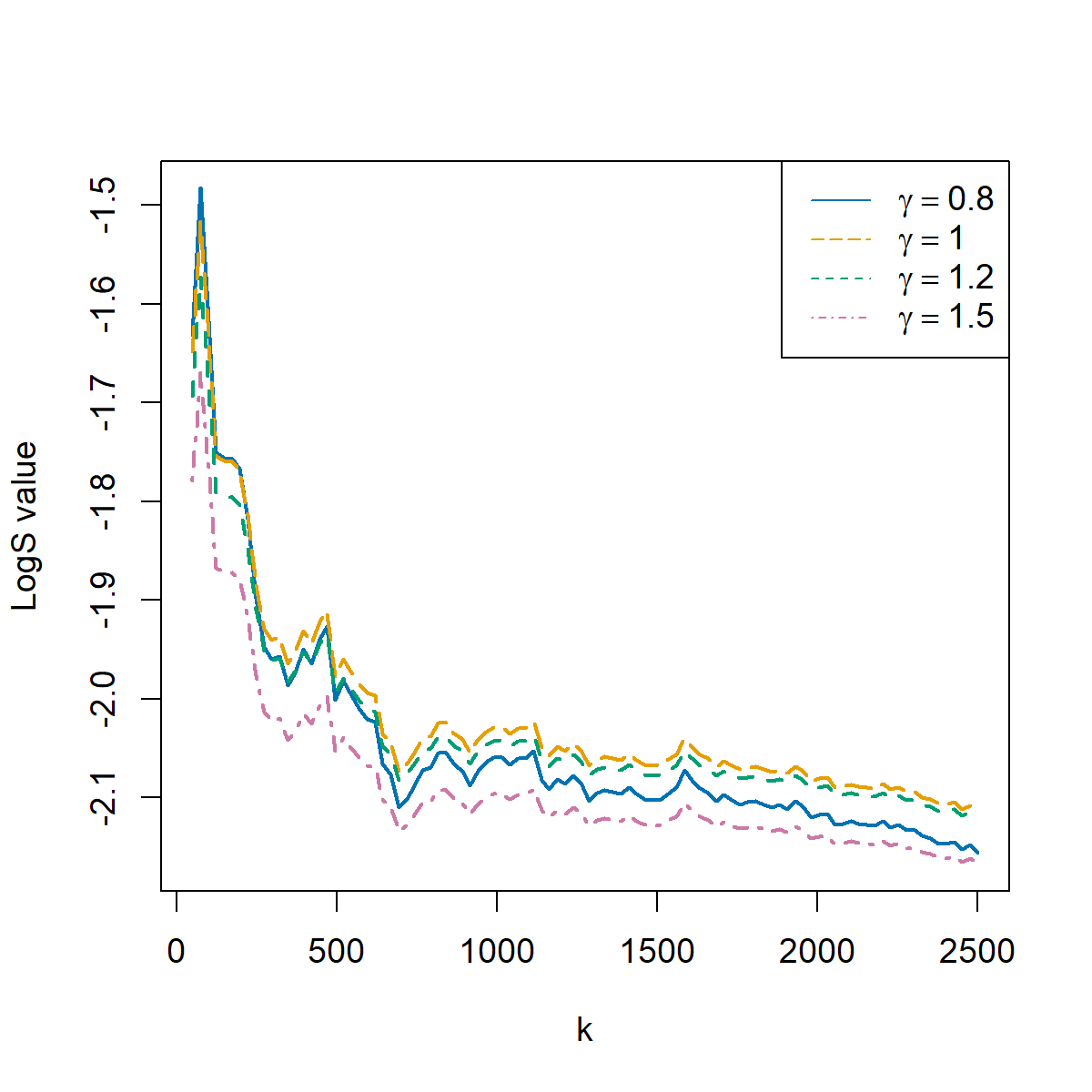}
    \caption{Fréchet DGP, $n = 10^4$}
\end{subfigure}
\begin{subfigure}{0.40\textwidth}
    \includegraphics[width=\linewidth, trim= 0.3in 0.6in 0.3in 0.6in,clip]{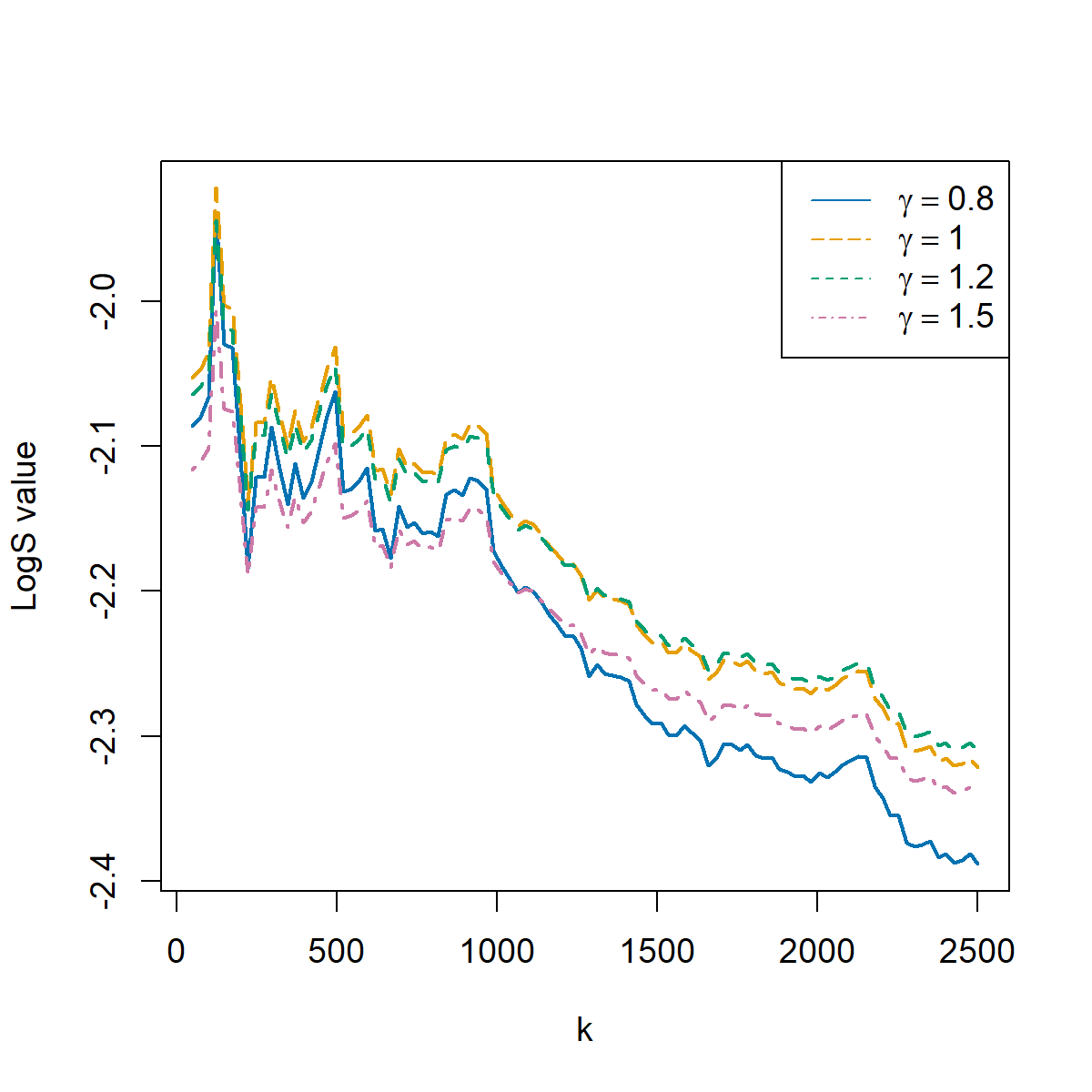}
    \caption{Burr DGP, $n = 10^4$}
\end{subfigure}\\
\begin{subfigure}{0.40\textwidth}
    \includegraphics[width=\linewidth, trim= 0.3in 0.6in 0.3in 0.6in,clip]{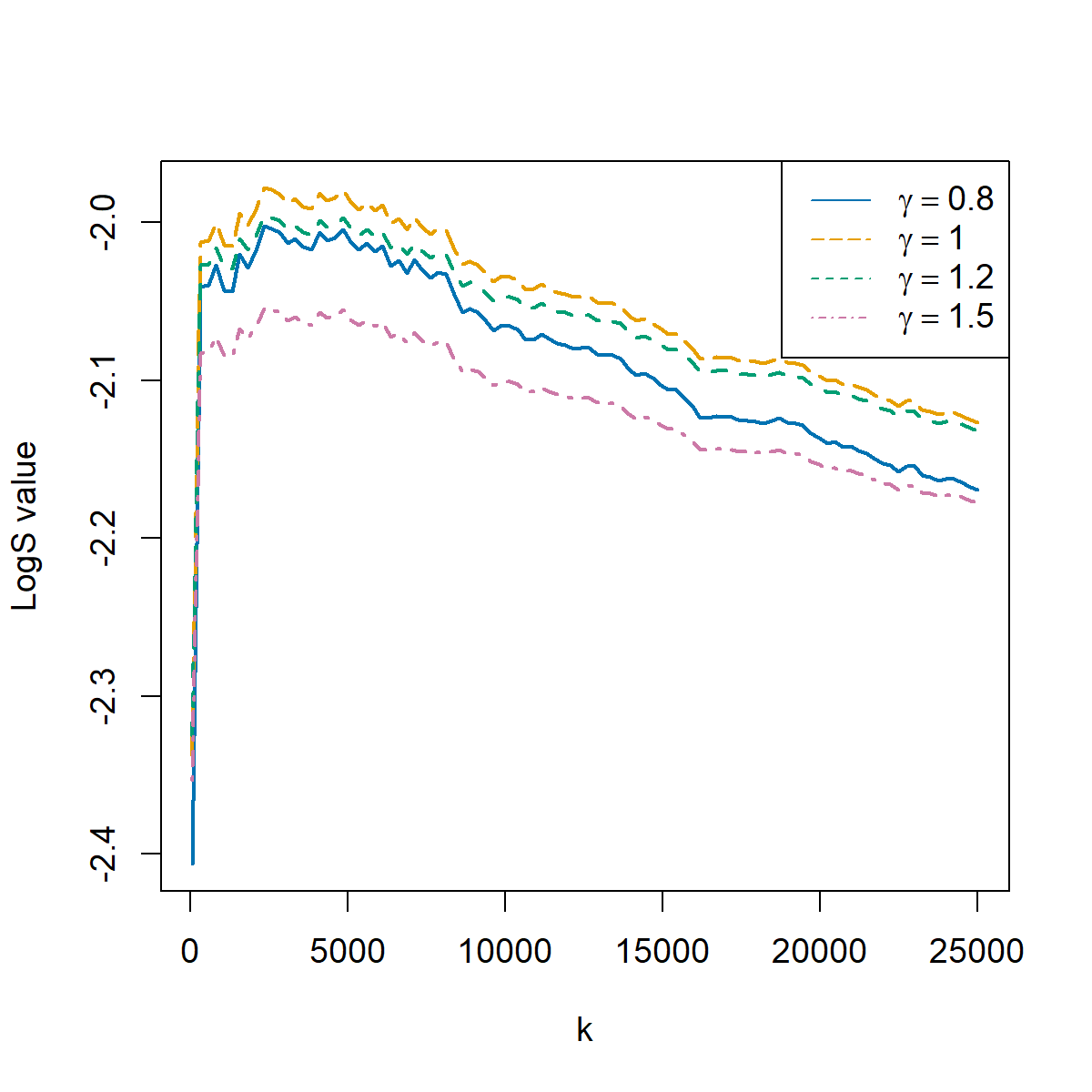}
    \caption{Fréchet DGP, $n = 10^5$}
\end{subfigure}
\begin{subfigure}{0.40\textwidth}
    \includegraphics[width=\linewidth, trim= 0.3in 0.6in 0.3in 0.6in,clip]{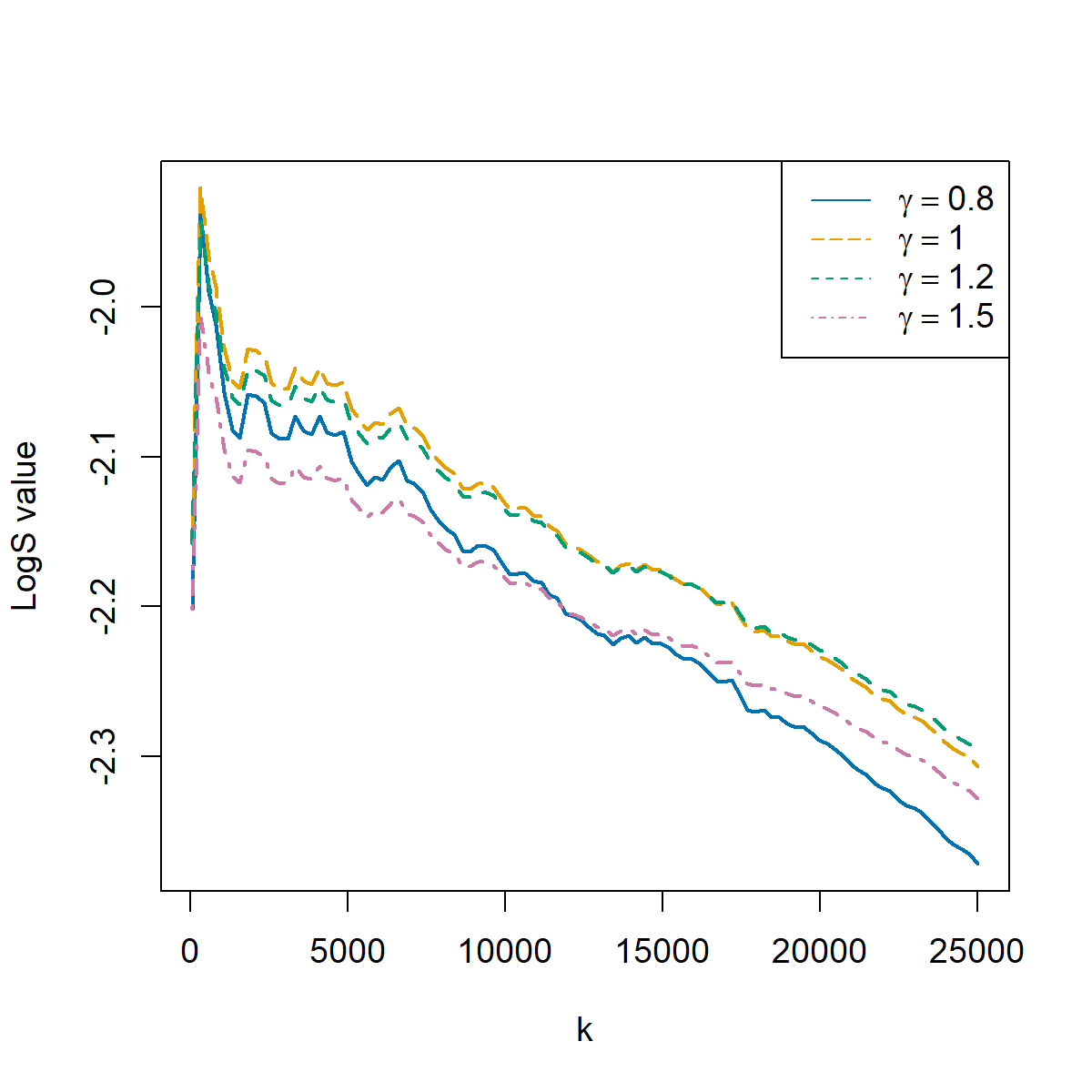}
    \caption{Burr DGP, $n = 10^5$}
\end{subfigure}

\caption{
Empirical logarithmic scores in \eqref{eq: estimator} (vertical axis) plotted against the number of upper order statistics $k$ (horizontal axis) for candidate tail indices $\gamma \in \{0.8, 1, 1.2, 1.5\}$. Left panels use Fréchet data-generating distributions, right panels use Burr distributions, and rows correspond to $n=10^3,10^4,10^5$. Higher curves indicate better tail fit; the true value is $\gamma_G=1$.}
\label{fig:LS_frechet_burr}
\end{figure}

\begin{figure}[H]
\centering

\begin{subfigure}{0.40\textwidth}
    \includegraphics[width=\linewidth, trim= 0.3in 0.6in 0.3in 0.6in,clip]{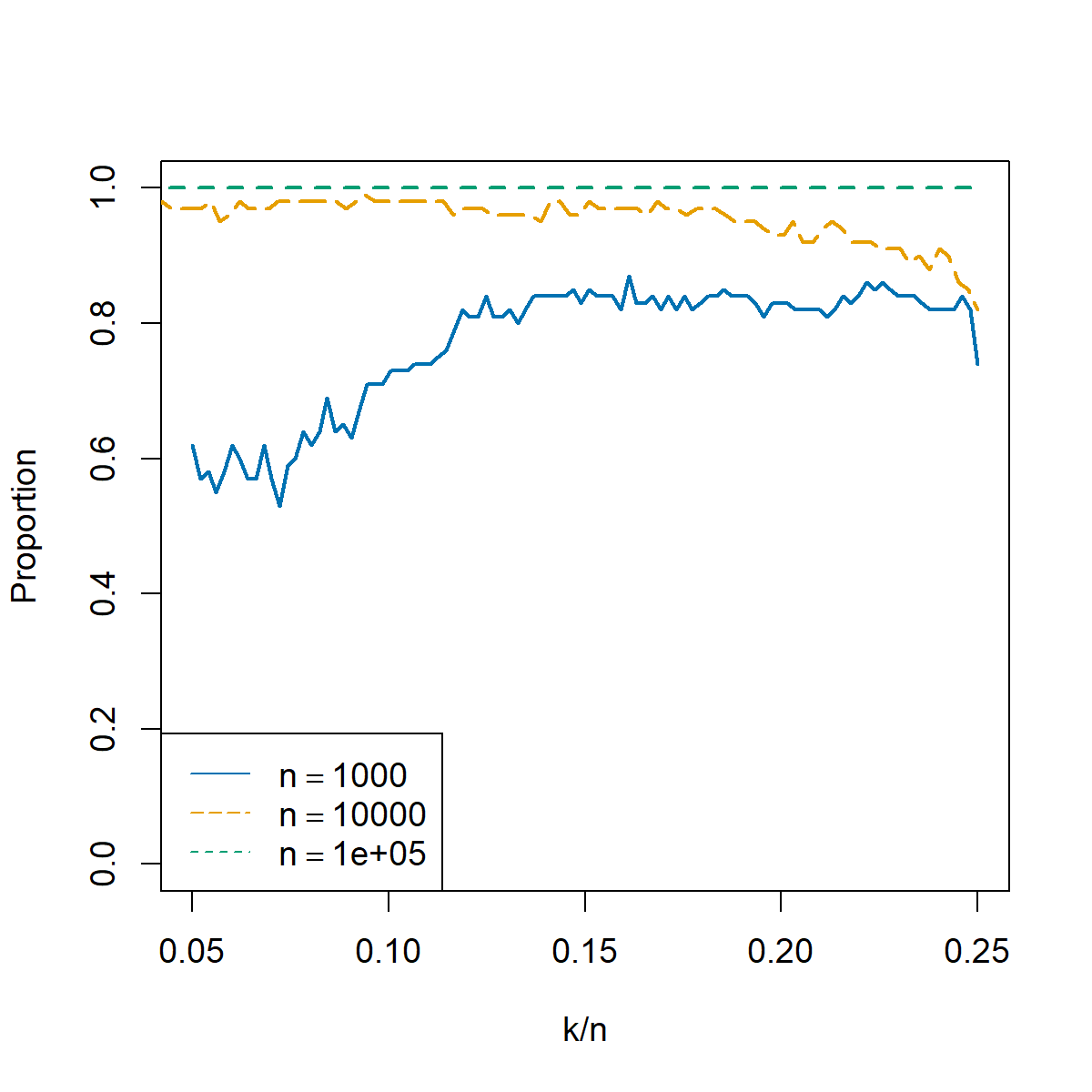}
    \caption{Fréchet DGP}
\end{subfigure}
\begin{subfigure}{0.40\textwidth}
    \includegraphics[width=\linewidth, trim= 0.3in 0.6in 0.3in 0.6in,clip]{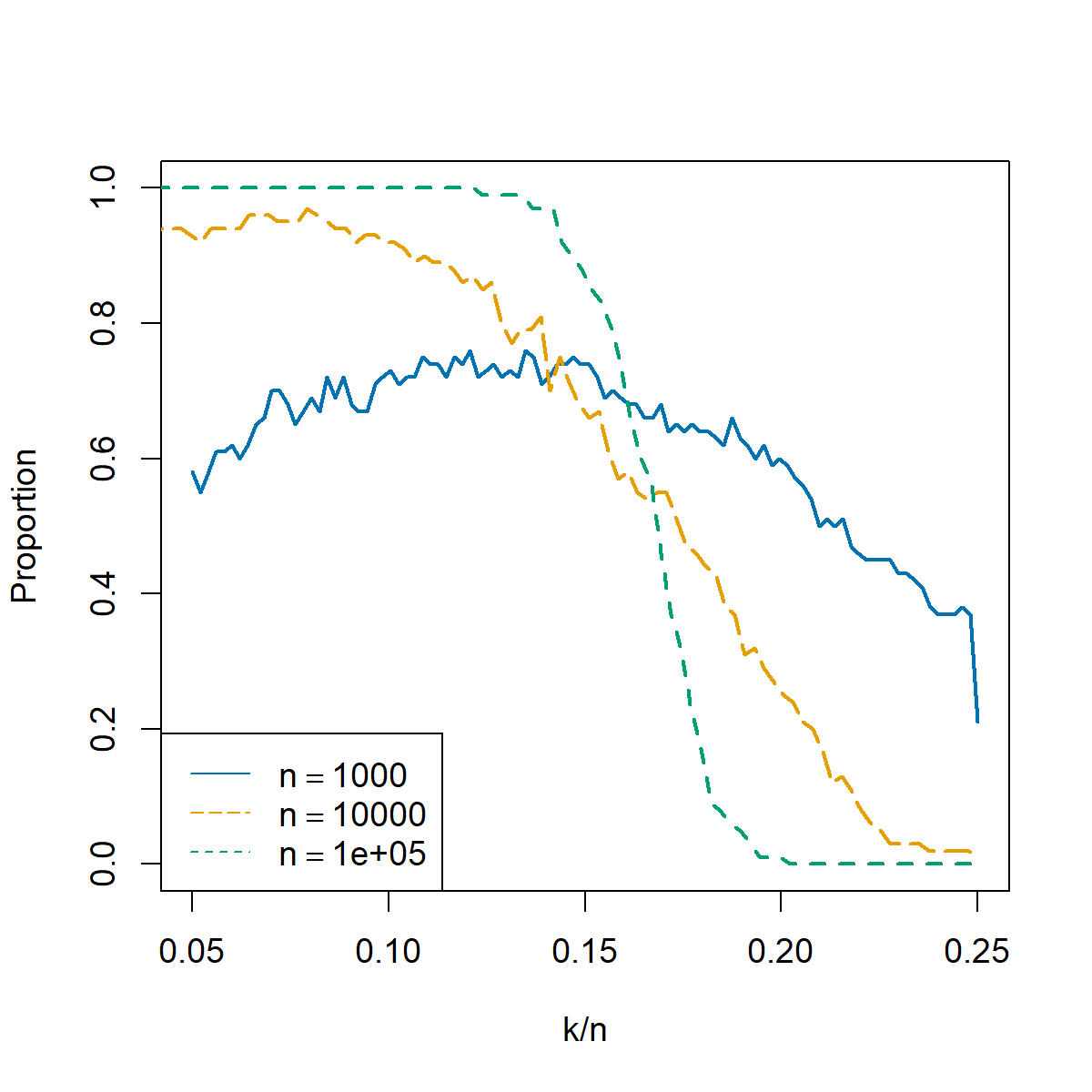}
    \caption{Burr DGP}
\end{subfigure}\\
\caption{
Proportion of simulations (based on 100 Monte Carlo replications) in which the empirical logarithmic score in \eqref{eq: estimator} is maximized at the true tail index $\gamma=1$ among candidate values $\gamma \in \{0.8, 1, 1.2, 1.5\}$, plotted against the relative number of upper order statistics $k/n$. The left plot uses Fréchet data-generating distributions and the right plot uses Burr distributions. Curves correspond to sample sizes $n=10^3,10^4,10^5$. Higher values indicate that the logarithmic score more frequently identifies the true tail index.}
\label{fig:LS_mean}
\end{figure}

\subsection{Ranking Tail Models under Heterogeneous Scaling}
\label{subsec:scales}

{A practically relevant situation is one in which the observations are not identically distributed. For example, covariates may affect the scale of the underlying distribution. Although such settings go beyond the asymptotic theory developed in Section~\ref{sec:math}, we include the following simulation experiment to examine the robustness of the proposed approach under systematic scale variation. }
Specifically, we simulate independent baseline heavy-tailed variables \(Z\) from Burr and Fr\'echet distributions and multiply each realization by a scaling factor. That is,
\[
Y_i = X_i Z_i .
\]

We follow the same setup as in Section~\ref{subsec:homosce} to generate the baseline heavy-tailed variable \(Z\). We consider two different structures for the scaling factors \(X\):

\begin{enumerate}[label=(\roman*)]
  \item \(X_i^1={1+}\frac{i}{n}\).
  \item \(X_i^2=1.5+0.5\sin\!\left(6\frac{i}{n}\pi\right)\).
\end{enumerate}

Both scaling structures imply that \(X_i\in[1,2]\). {These two structures represent qualitatively different forms of scale heterogeneity and should be viewed as simple heuristic stress tests. They are chosen to separate two effects that may arise in practice: persistent changes in scale and recurrent scale fluctuations. The monotone structure captures the former, with the largest scale occurring only at the end of the sample. The periodic structure captures the latter, producing repeated high-scale regions so that several parts of the sample contain observations with the largest scale.}

Importantly, in this simulation design the tail index is constant across all observations. The only source of heterogeneity arises from the scaling component. Consequently, the tail heaviness of the distribution remains unchanged, while only the scale varies across \(i\). Under this setup, the model corresponding to \(\gamma = 1\) correctly specifies the tail index of the data-generating process. Therefore, we would expect the logarithmic score (LogS) to attain its highest value at \(\gamma = 1\).

 Figures~\ref{fig:LS_frechet_hetero} and \ref{fig:LS_burr_hetero} present the results of a single realization, while Figure~\ref{fig:LS_mean_scale} reports the proportion results. The figures are constructed in the same manner as those in the previous section and are therefore not described in further detail. The same \(k\)-grid and candidate set as in Section~\ref{subsec:homosce} are used, so changes in behavior can be attributed to the scaling mechanism. Even in the presence of heterogeneous scaling, the qualitative conclusions of Section~\ref{subsec:homosce} remain valid. In particular, the logarithmic score continues to favor the correct tail index for sufficiently large sample sizes and appropriate choices of $k$, and the most reliable inference is still concentrated in a stable lower-\(k\) region.

\begin{figure}[H]
\centering

\begin{subfigure}{0.40\textwidth}
    \includegraphics[width=\linewidth, trim= 0.3in 0.6in 0.3in 0.6in,clip]{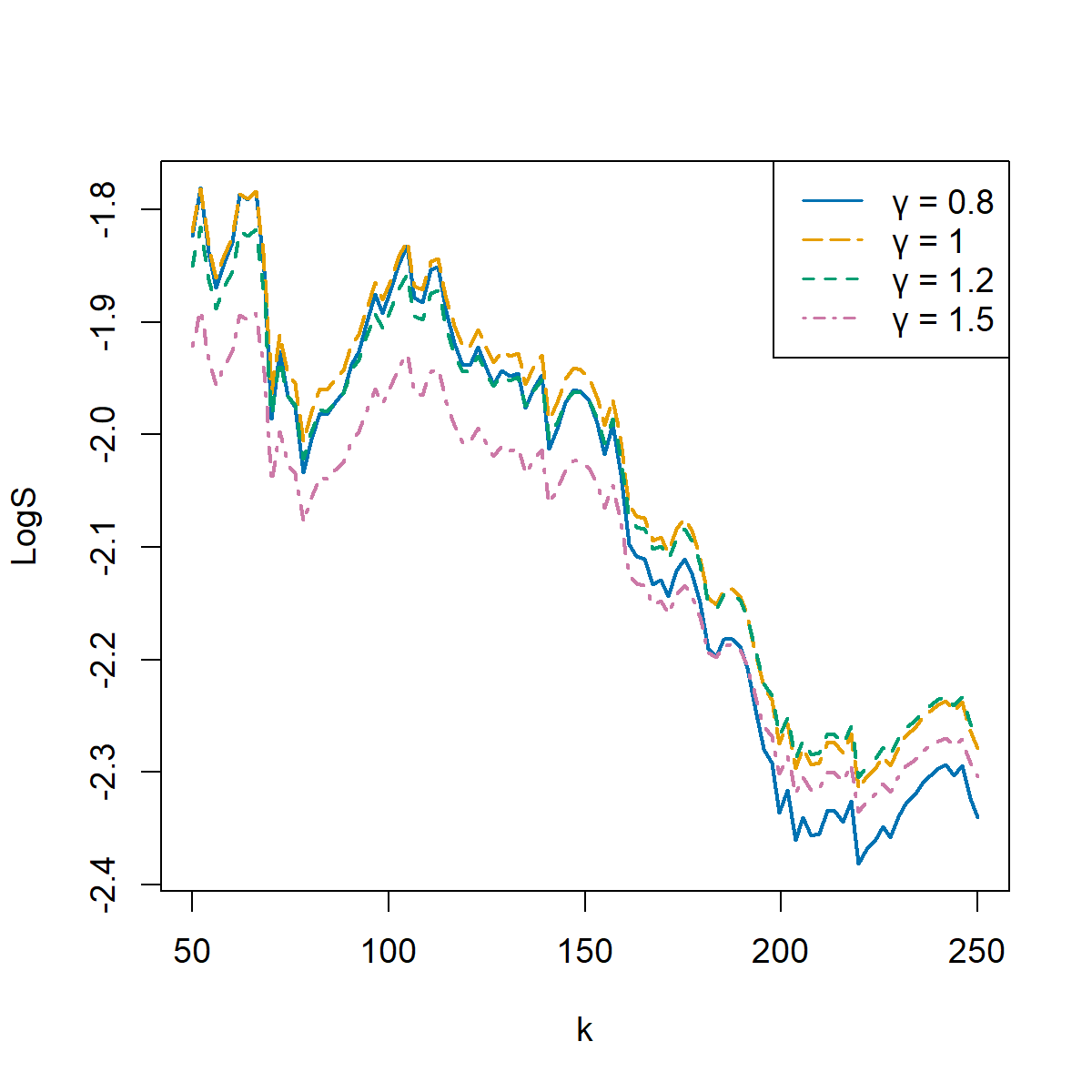}
    \caption{Linear scaling $X_i=X_i^1$, $n = 10^3$}
\end{subfigure}
\begin{subfigure}{0.40\textwidth}
    \includegraphics[width=\linewidth, trim= 0.3in 0.6in 0.3in 0.6in,clip]{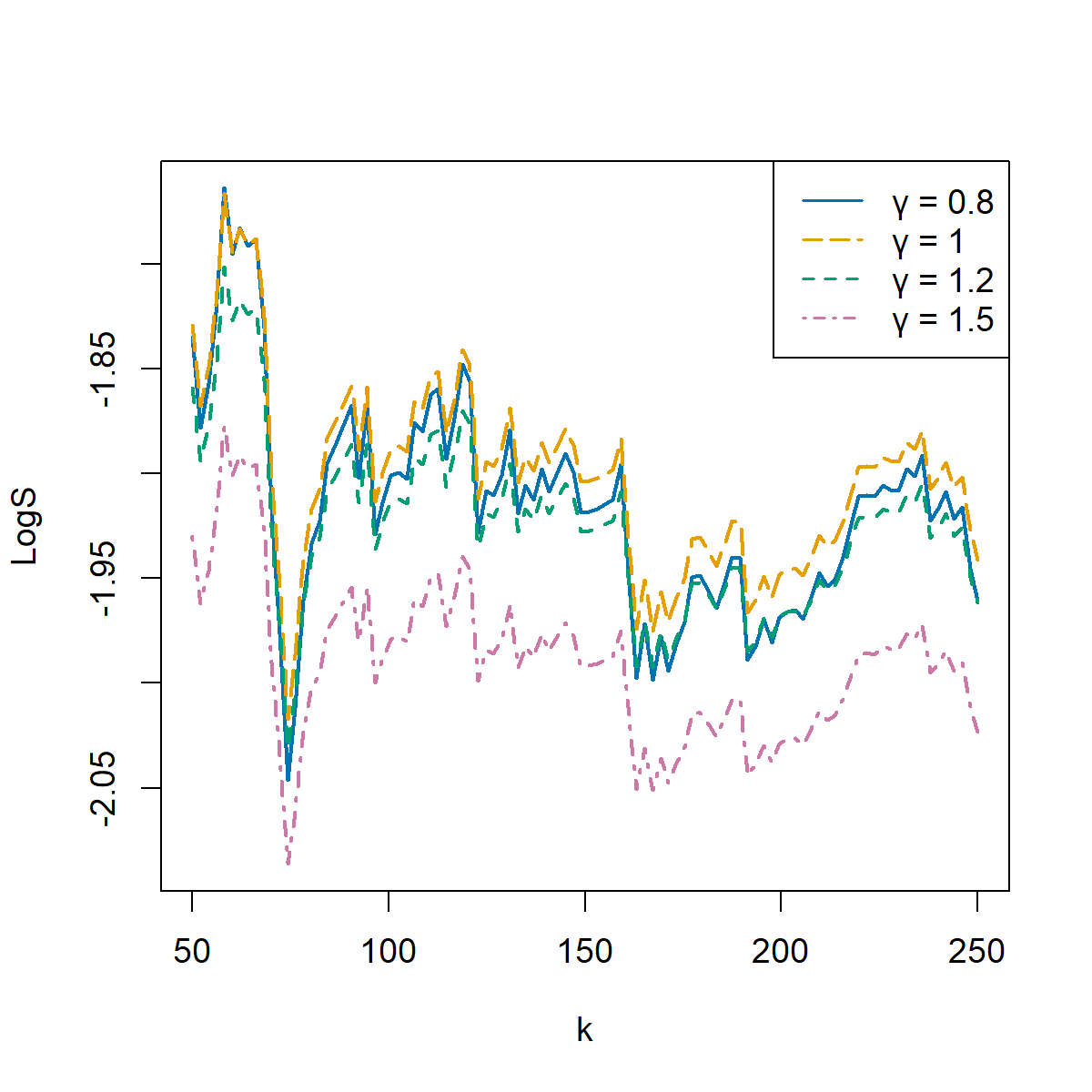}
    \caption{Sinusoidal scaling $X_i=X_i^2$, $n = 10^3$}
\end{subfigure}\\
\begin{subfigure}{0.40\textwidth}
    \includegraphics[width=\linewidth, trim= 0.3in 0.6in 0.3in 0.6in,clip]{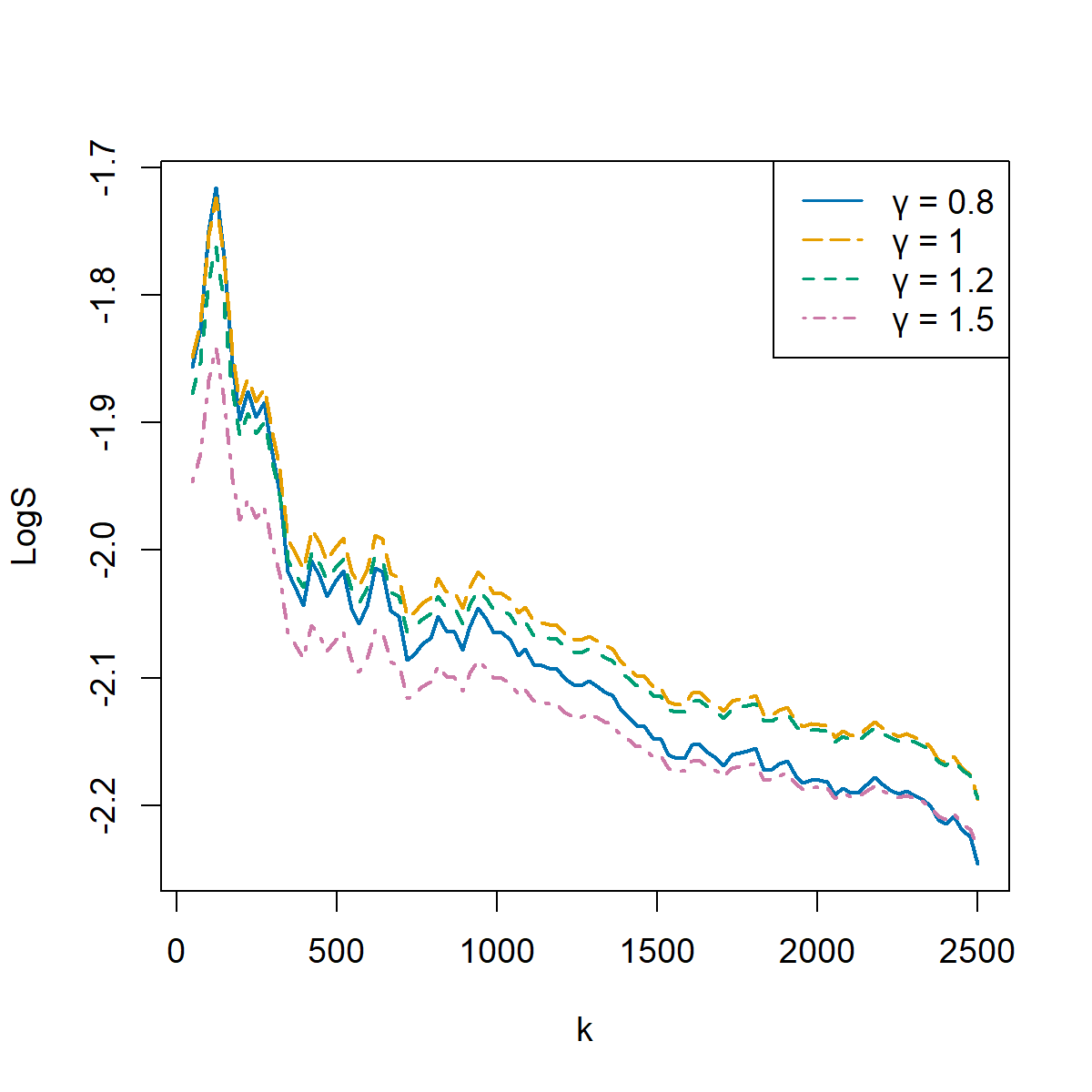}
    \caption{Linear scaling $X_i=X_i^1$, $n = 10^4$}
\end{subfigure}
\begin{subfigure}{0.40\textwidth}
    \includegraphics[width=\linewidth, trim= 0.3in 0.6in 0.3in 0.6in,clip]{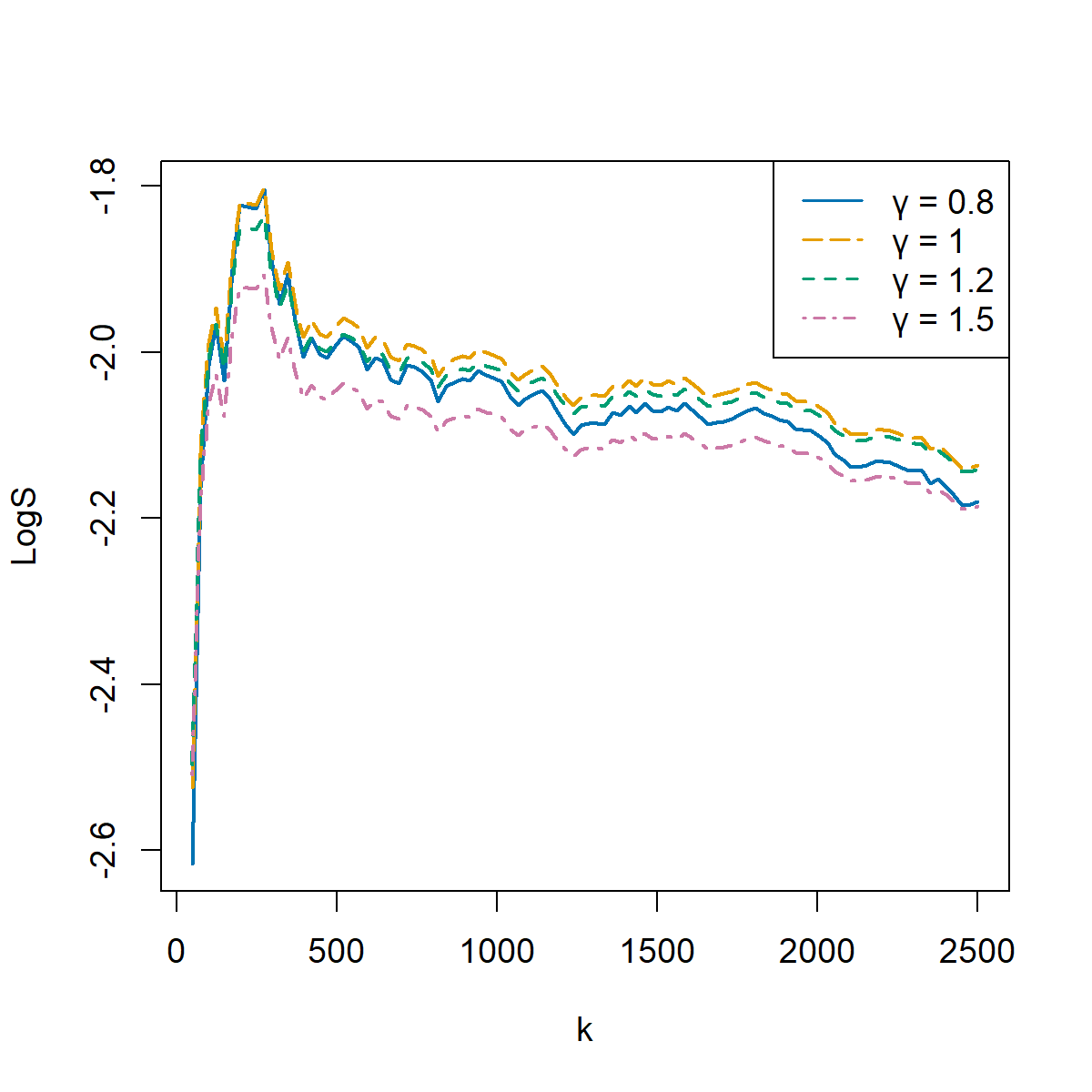}
    \caption{Sinusoidal scaling $X_i=X_i^2$, $n = 10^4$}
\end{subfigure}\\
\begin{subfigure}{0.40\textwidth}
    \includegraphics[width=\linewidth, trim= 0.3in 0.6in 0.3in 0.6in,clip]{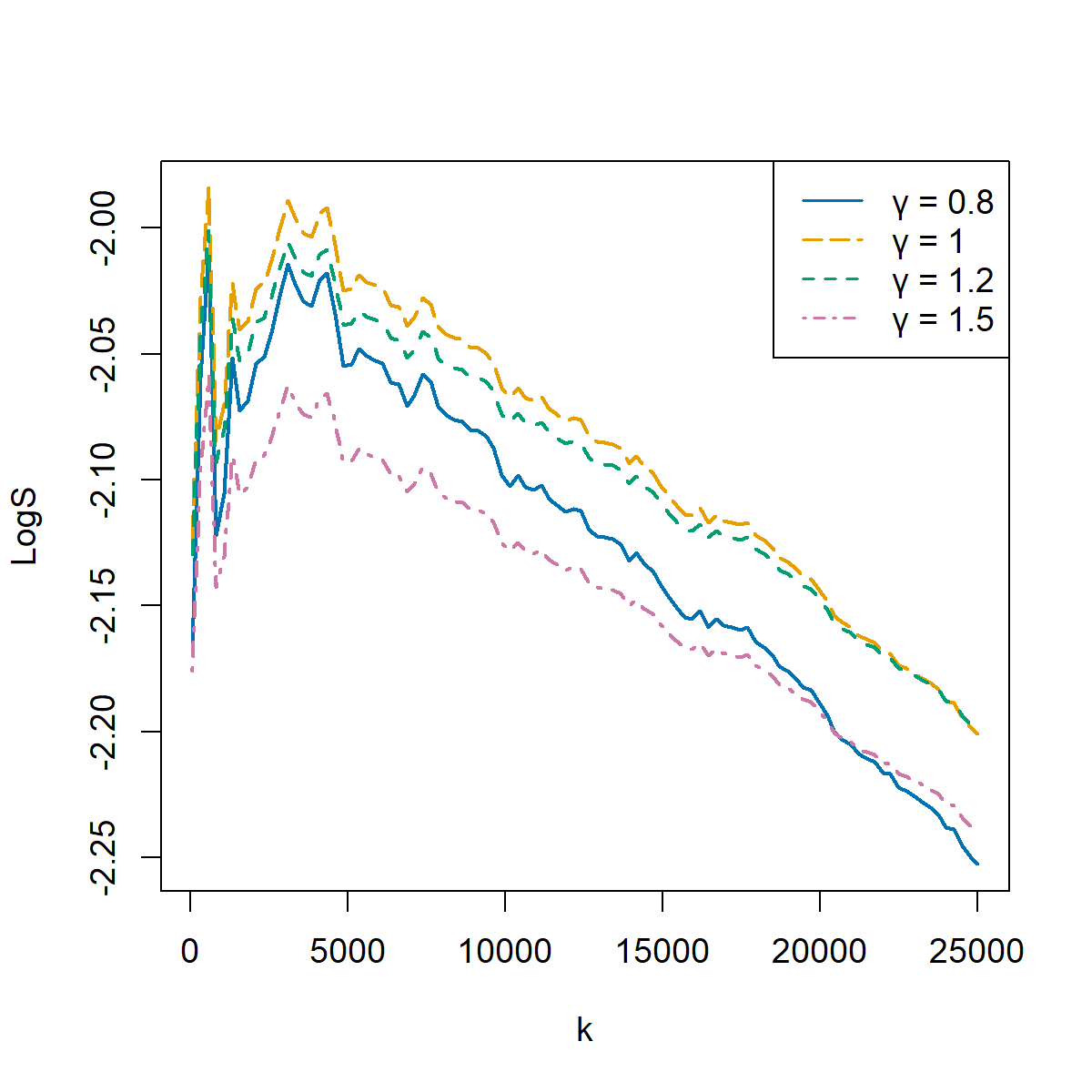}
    \caption{Linear scaling $X_i=X_i^1$, $n = 10^5$}
\end{subfigure}
\begin{subfigure}{0.40\textwidth}
    \includegraphics[width=\linewidth, trim= 0.3in 0.6in 0.3in 0.6in,clip]{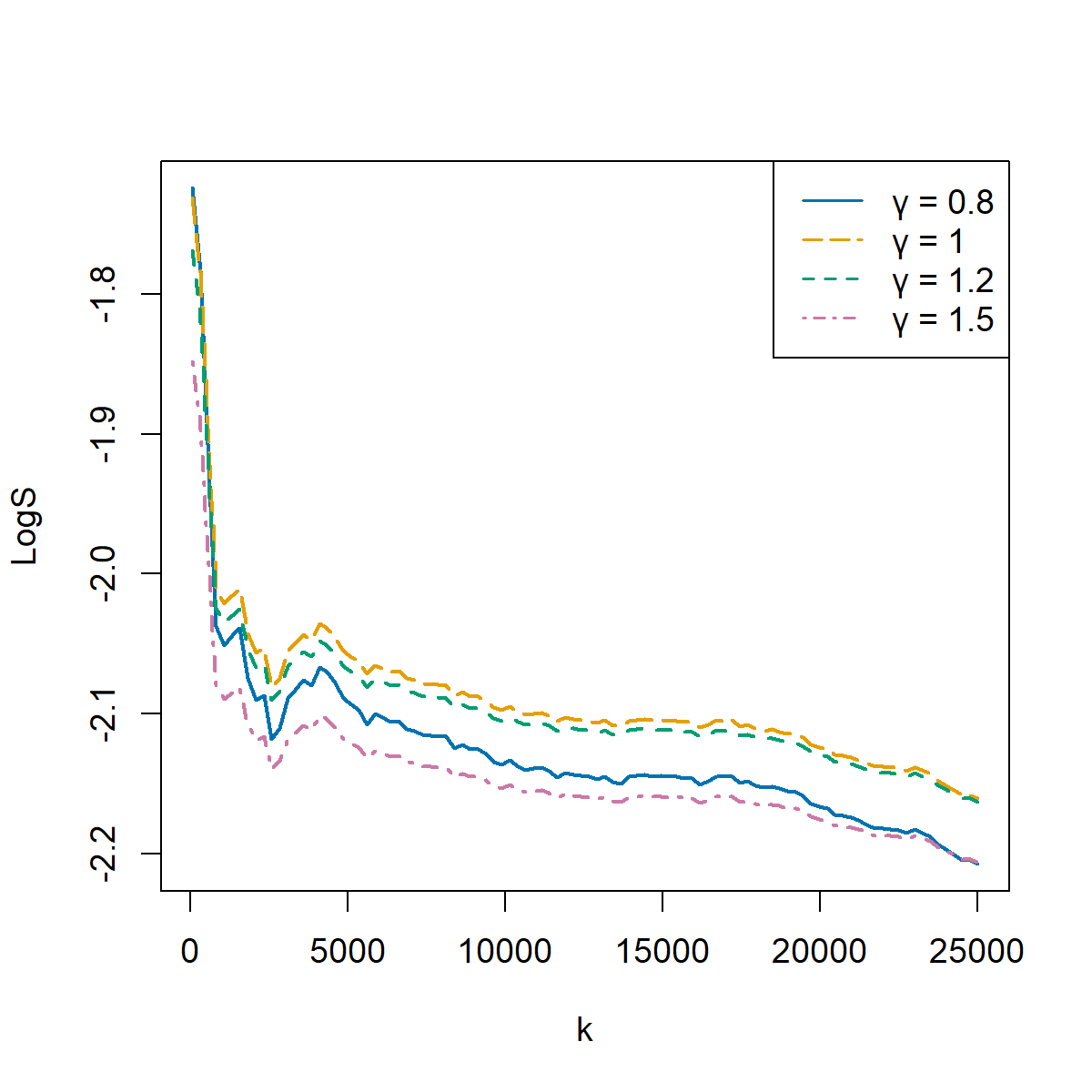}
    \caption{Sinusoidal scaling $X_i=X_i^2$, $n = 10^5$}
\end{subfigure}

\caption{Empirical logarithmic scores (vertical axis) versus $k$ (horizontal axis) for Fréchet baseline samples with heterogeneous scaling, $Y_i=X_iZ_i$, where $Z_i$ has tail index $\gamma_G=1$. Candidate tail indices are $\gamma \in \{0.8, 1, 1.2, 1.5\}$. Left panels use $X_i^1={1+}i/n$, right panels use $X_i^2=1.5+0.5\sin(6\pi i/n)$, and rows correspond to $n=10^3,10^4,10^5$.
}
\label{fig:LS_frechet_hetero}
\end{figure}

\begin{figure}[H]
\centering

\begin{subfigure}{0.40\textwidth}
    \includegraphics[width=\linewidth, trim= 0.3in 0.6in 0.3in 0.6in,clip]{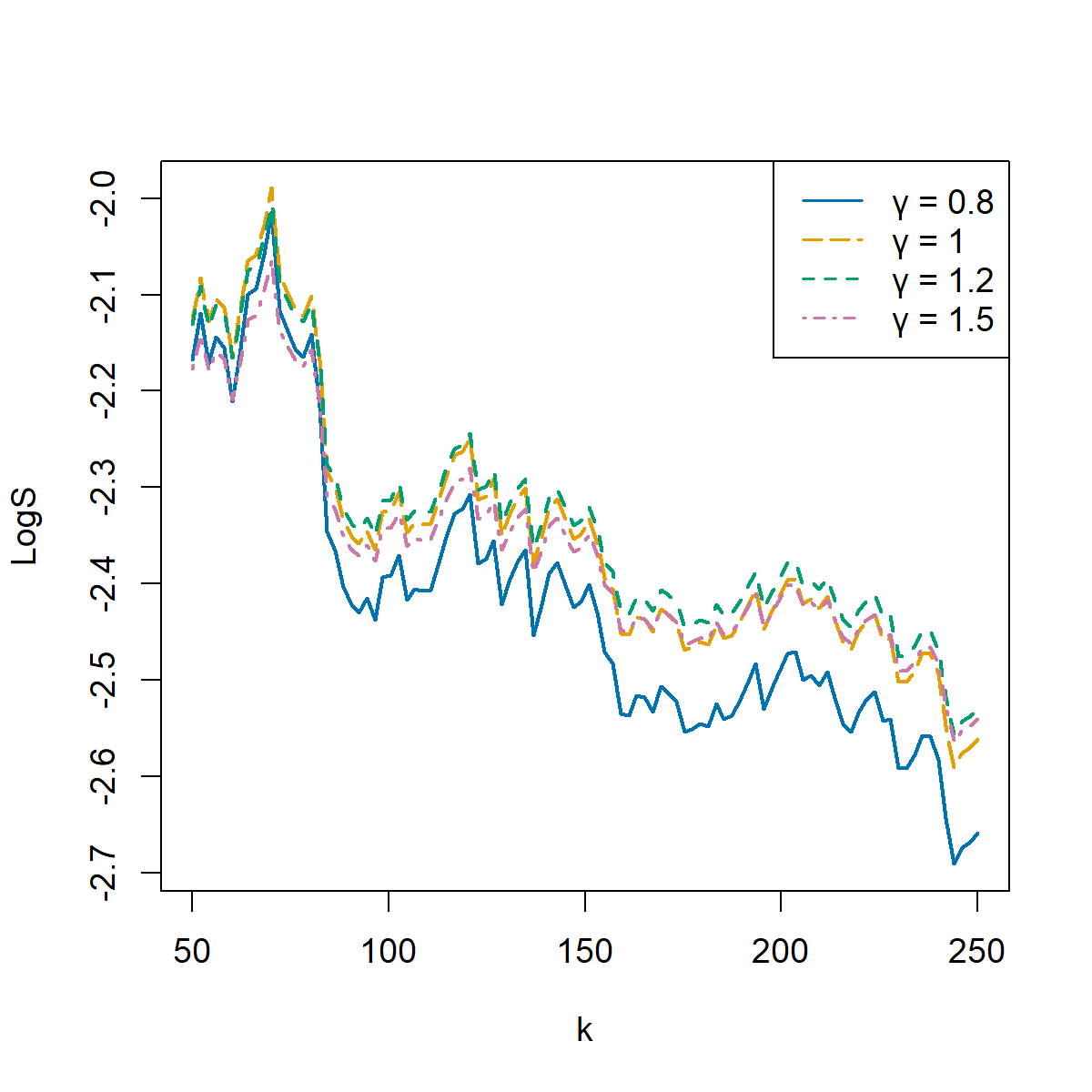}
    \caption{Linear scaling $X_i=X_i^1$, $n = 10^3$}
\end{subfigure}
\begin{subfigure}{0.40\textwidth}
    \includegraphics[width=\linewidth, trim= 0.3in 0.6in 0.3in 0.6in,clip]{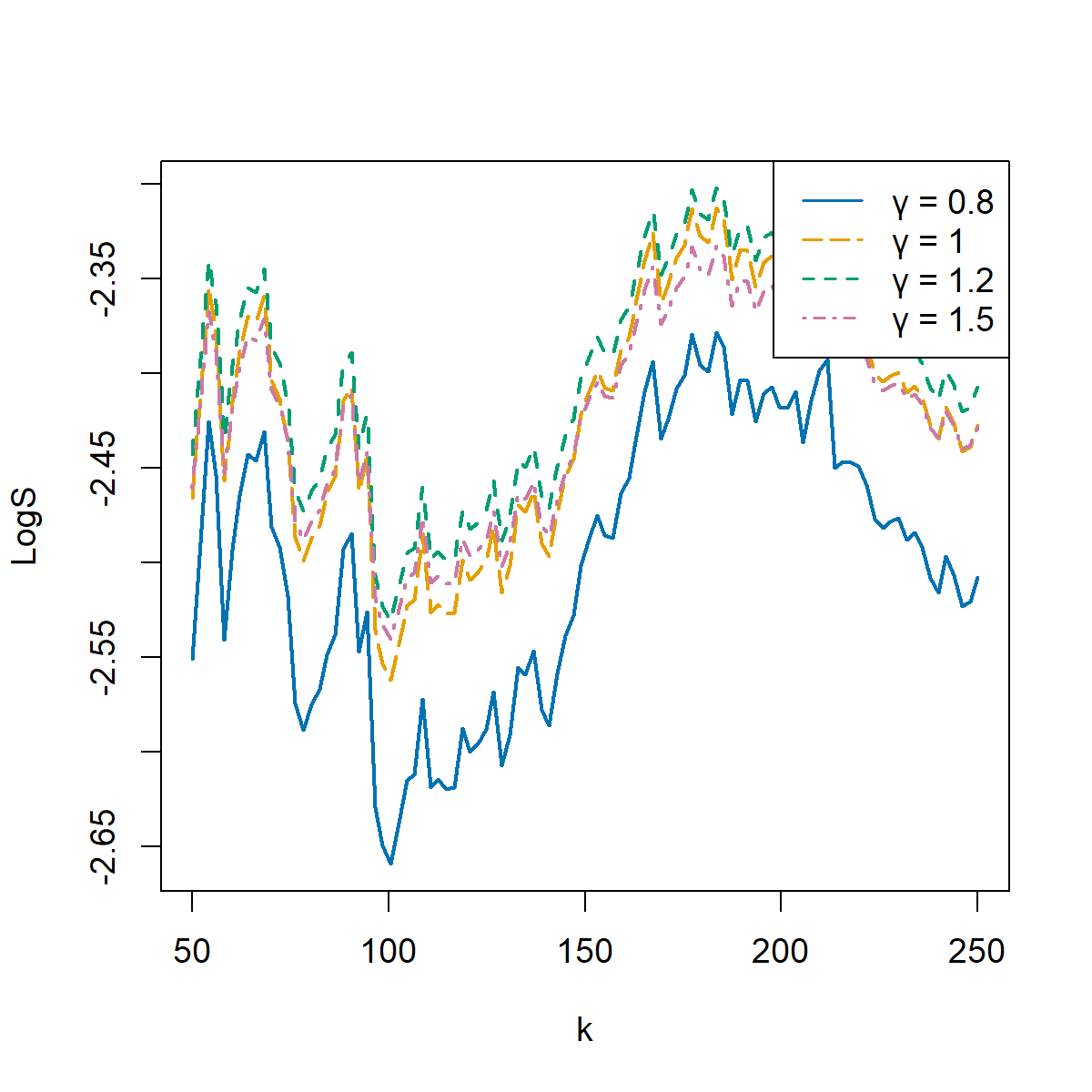}
    \caption{Sinusoidal scaling $X_i=X_i^2$, $n = 10^3$}
\end{subfigure}\\
\begin{subfigure}{0.40\textwidth}
    \includegraphics[width=\linewidth, trim= 0.3in 0.6in 0.3in 0.6in,clip]{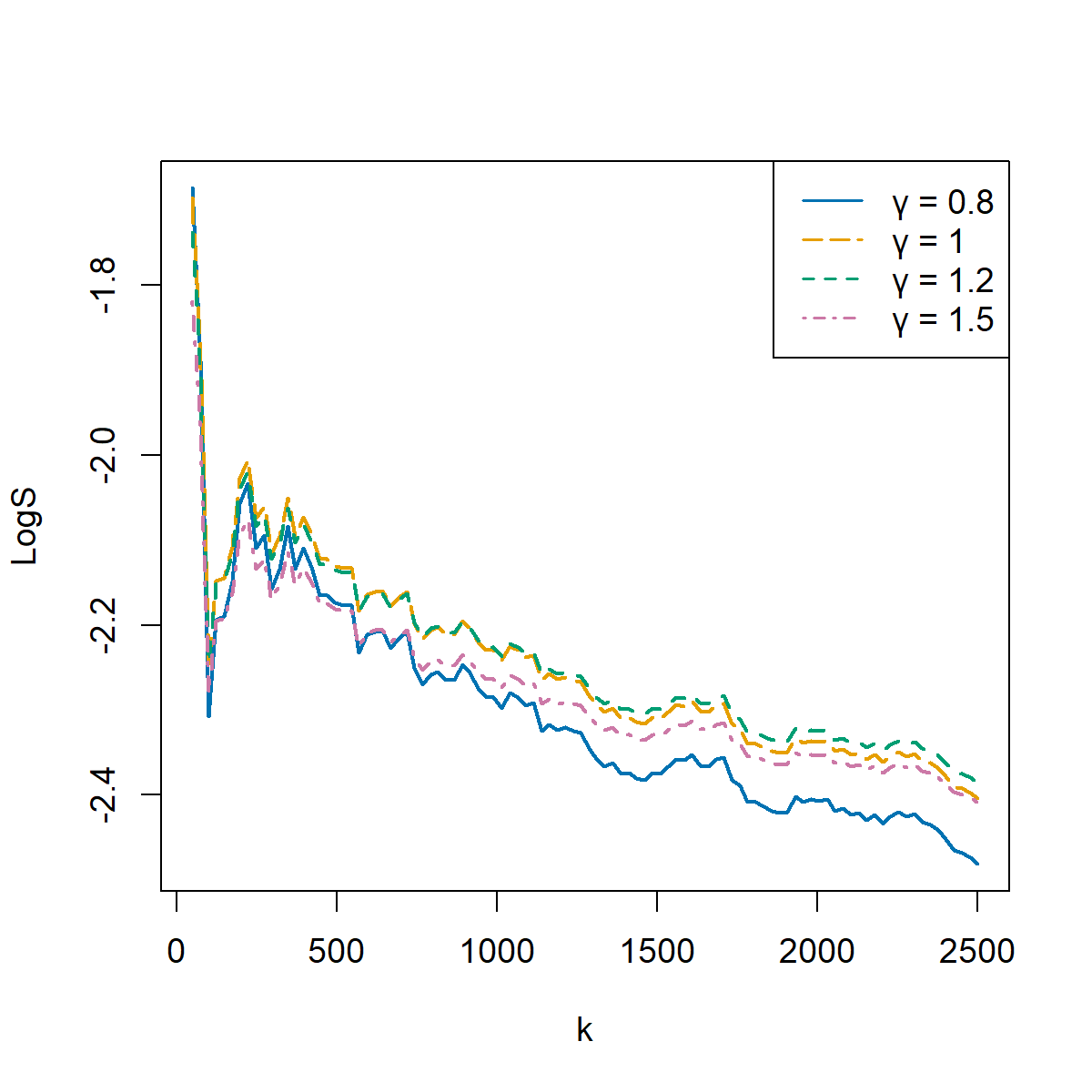}
    \caption{Linear scaling $X_i=X_i^1$, $n = 10^4$}
\end{subfigure}
\begin{subfigure}{0.40\textwidth}
    \includegraphics[width=\linewidth, trim= 0.3in 0.6in 0.3in 0.6in,clip]{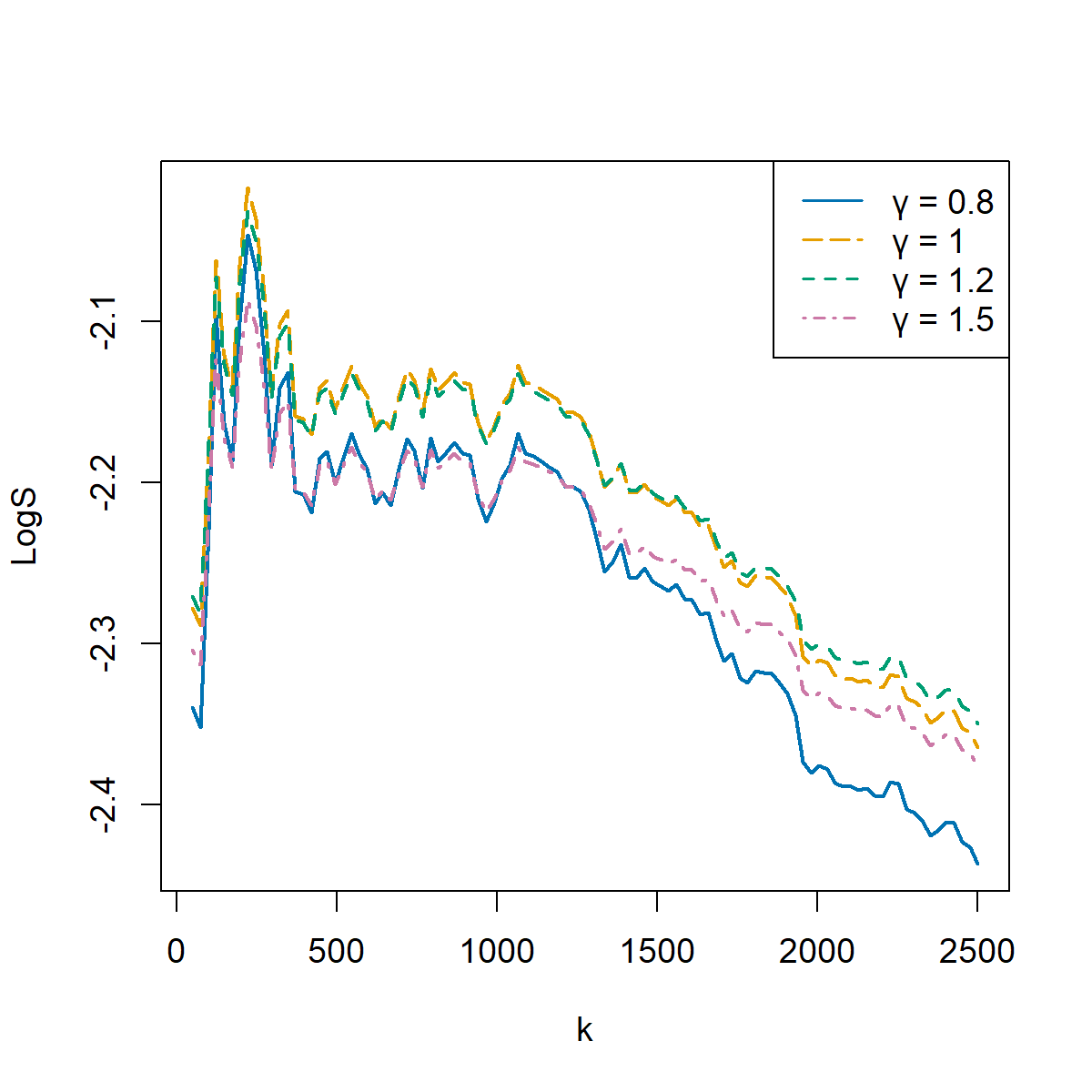}
    \caption{Sinusoidal scaling $X_i=X_i^2$, $n = 10^4$}
\end{subfigure}\\
\begin{subfigure}{0.40\textwidth}
    \includegraphics[width=\linewidth, trim= 0.3in 0.6in 0.3in 0.6in,clip]{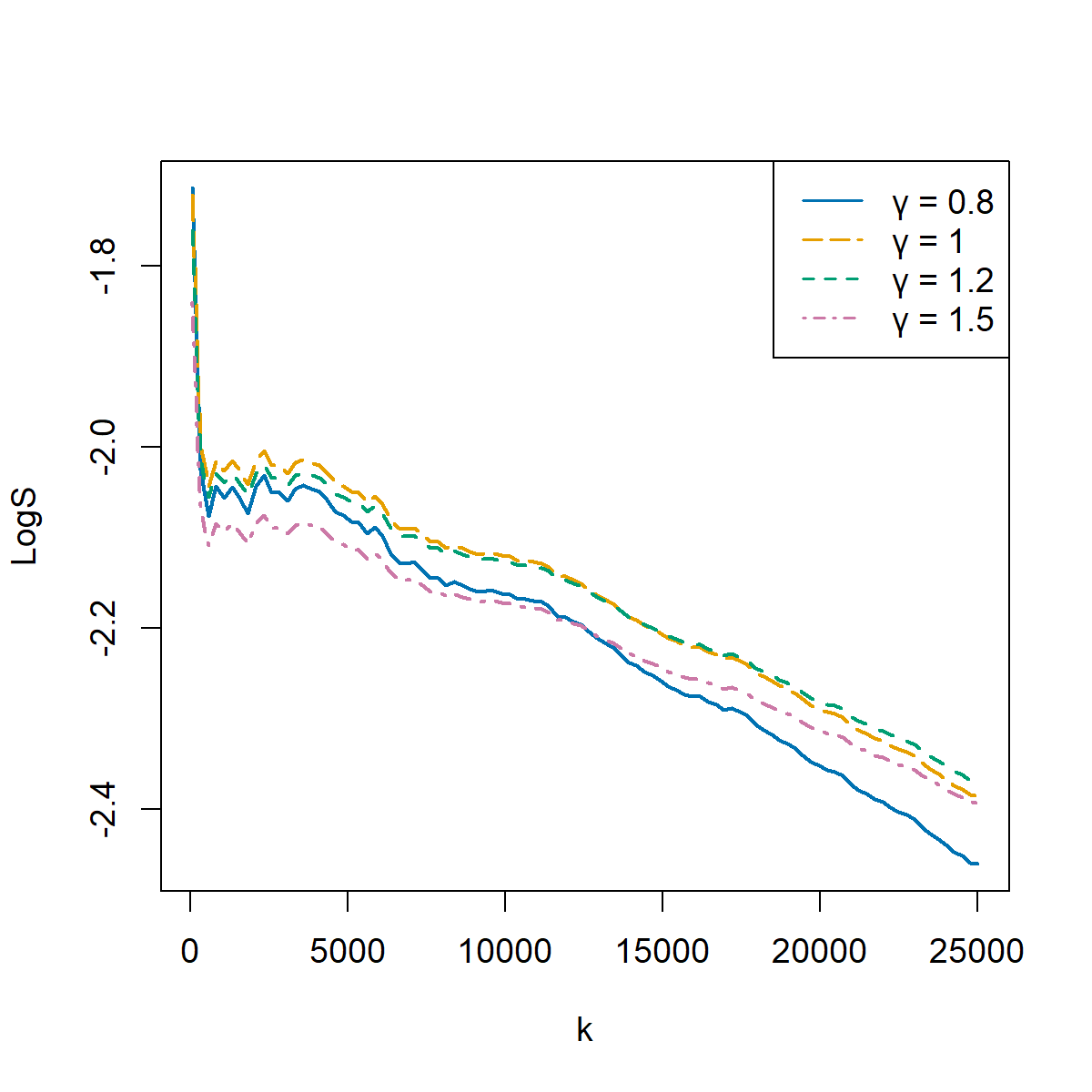}
    \caption{Linear scaling $X_i=X_i^1$, $n = 10^5$}
\end{subfigure}
\begin{subfigure}{0.40\textwidth}
    \includegraphics[width=\linewidth, trim= 0.3in 0.6in 0.3in 0.6in,clip]{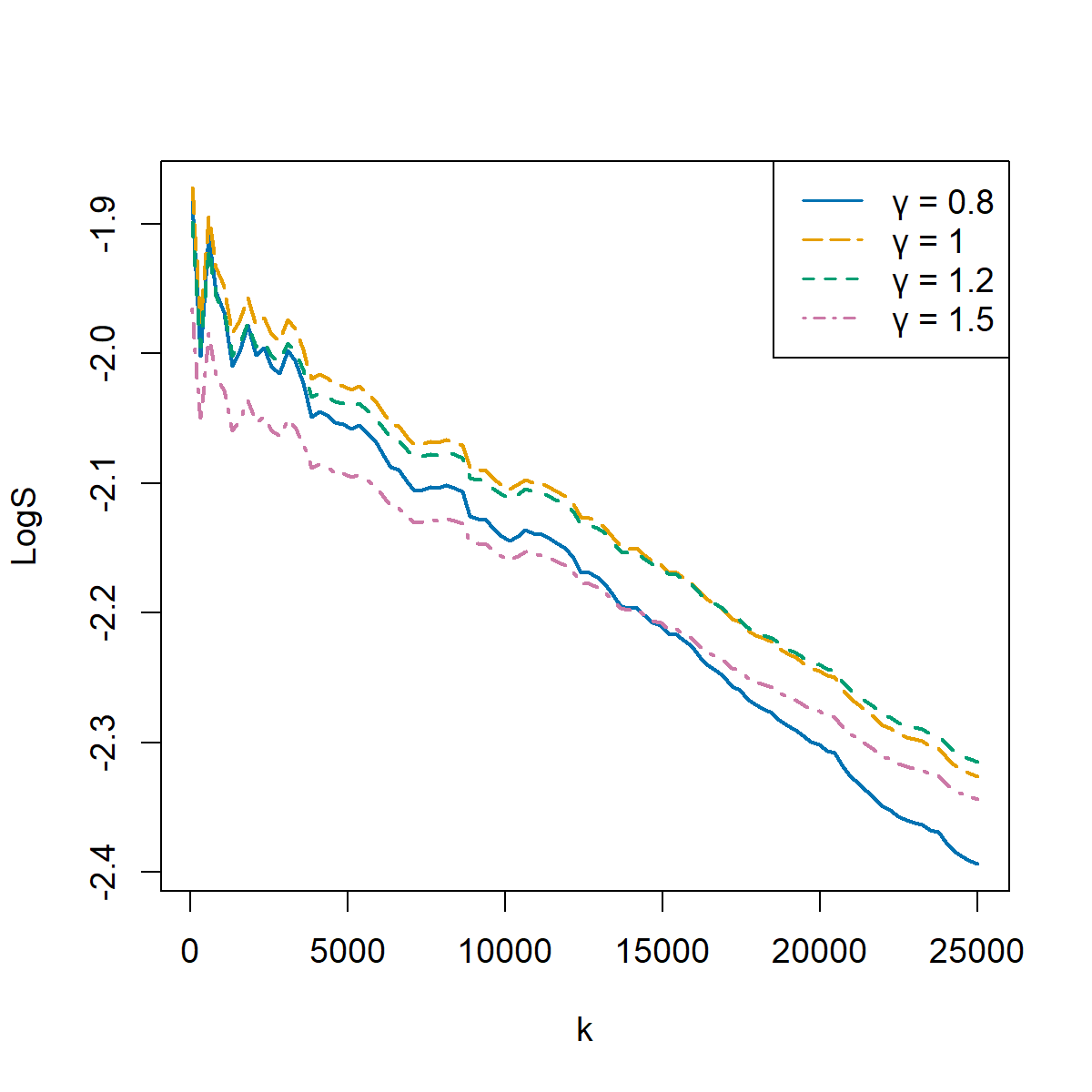}
    \caption{Sinusoidal scaling $X_i=X_i^2$, $n = 10^5$}
\end{subfigure}

\caption{
Empirical logarithmic scores (vertical axis) versus $k$ (horizontal axis) for Burr baseline samples with heterogeneous scaling, $Y_i=X_iZ_i$, where $Z_i$ has tail index $\gamma_G=1$. Candidate tail indices are $\gamma \in \{0.8, 1, 1.2, 1.5\}$. Left panels use $X_i^1={1+}i/n$, right panels use $X_i^2=1.5+0.5\sin(6\pi i/n)$, and rows correspond to $n=10^3,10^4,10^5$.
}
\label{fig:LS_burr_hetero}
\end{figure}

\begin{figure}[H]
\centering
\begin{subfigure}{0.40\textwidth}
    \includegraphics[width=\linewidth, trim= 0.3in 0.6in 0.3in 0.6in,clip]{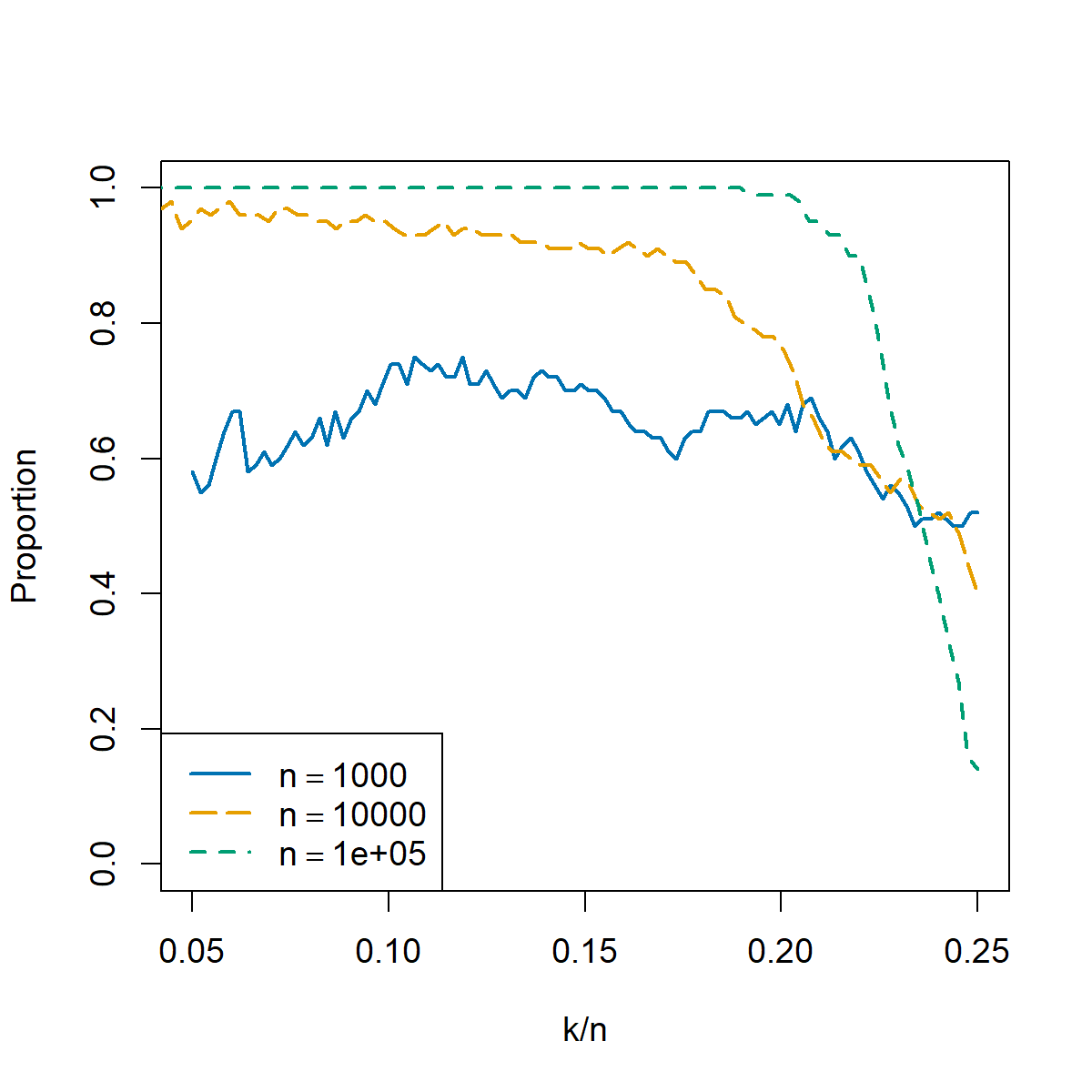}
    \caption{Fréchet DGP, Linear scaling $X_i=X_i^1$}
\end{subfigure}
\begin{subfigure}{0.40\textwidth}
    \includegraphics[width=\linewidth, trim= 0.3in 0.6in 0.3in 0.6in,clip]{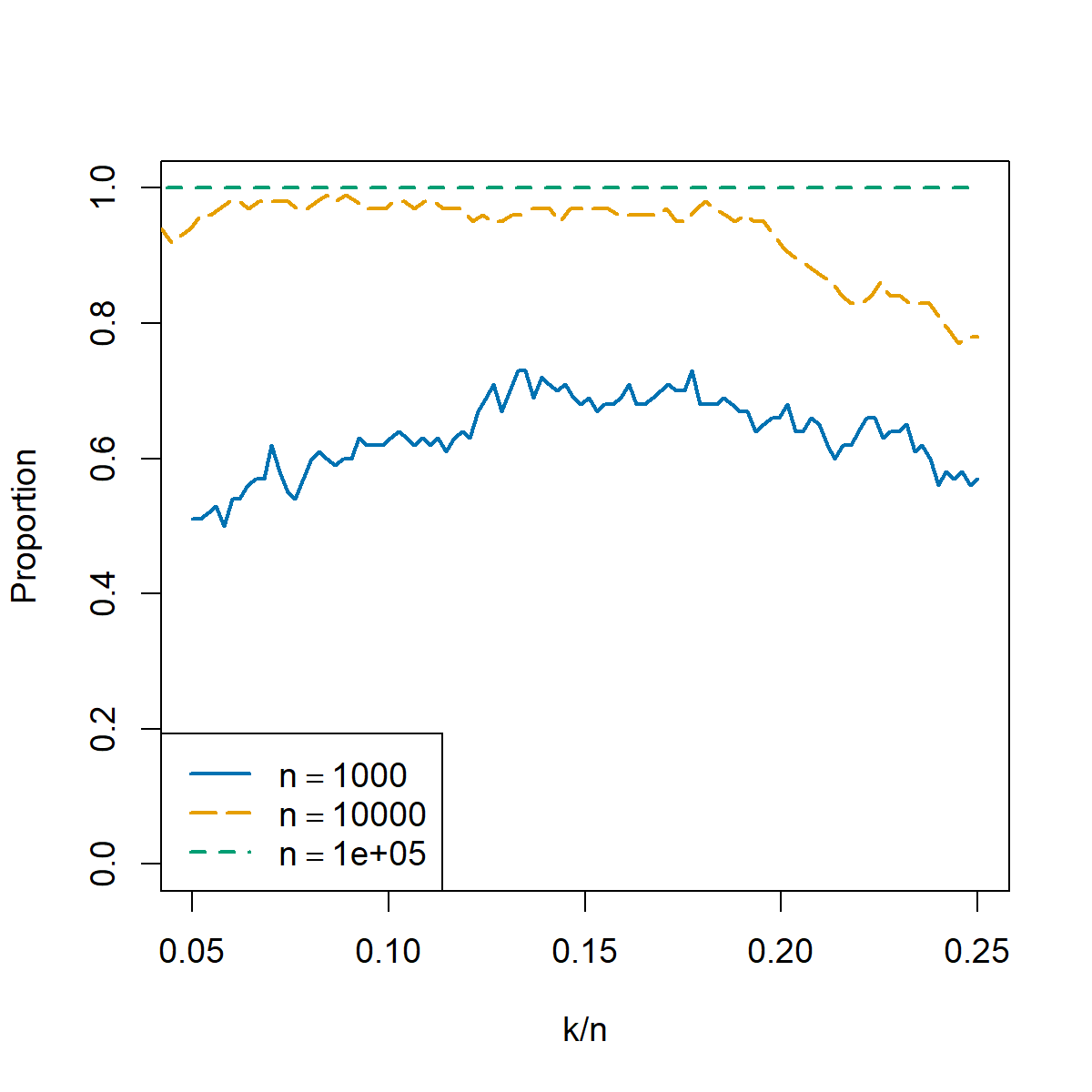}
    \caption{Fréchet DGP, Sinusoidal scaling $X_i=X_i^2$}
\end{subfigure}
\begin{subfigure}{0.40\textwidth}
    \includegraphics[width=\linewidth, trim= 0.3in 0.6in 0.3in 0.6in,clip]{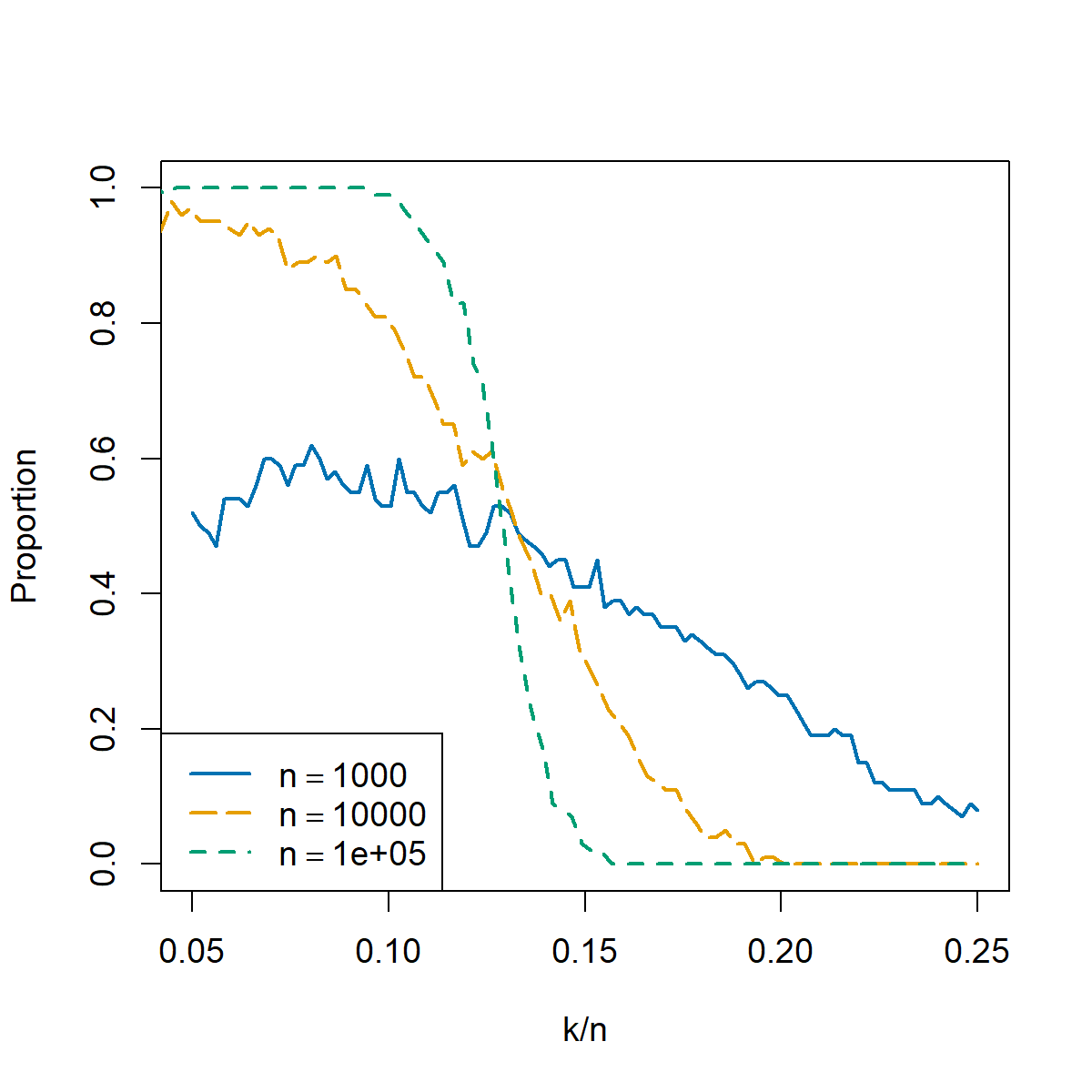}
    \caption{Burr DGP, Linear scaling $X_i=X_i^1$}
\end{subfigure}
\begin{subfigure}{0.40\textwidth}
    \includegraphics[width=\linewidth, trim= 0.3in 0.6in 0.3in 0.6in,clip]{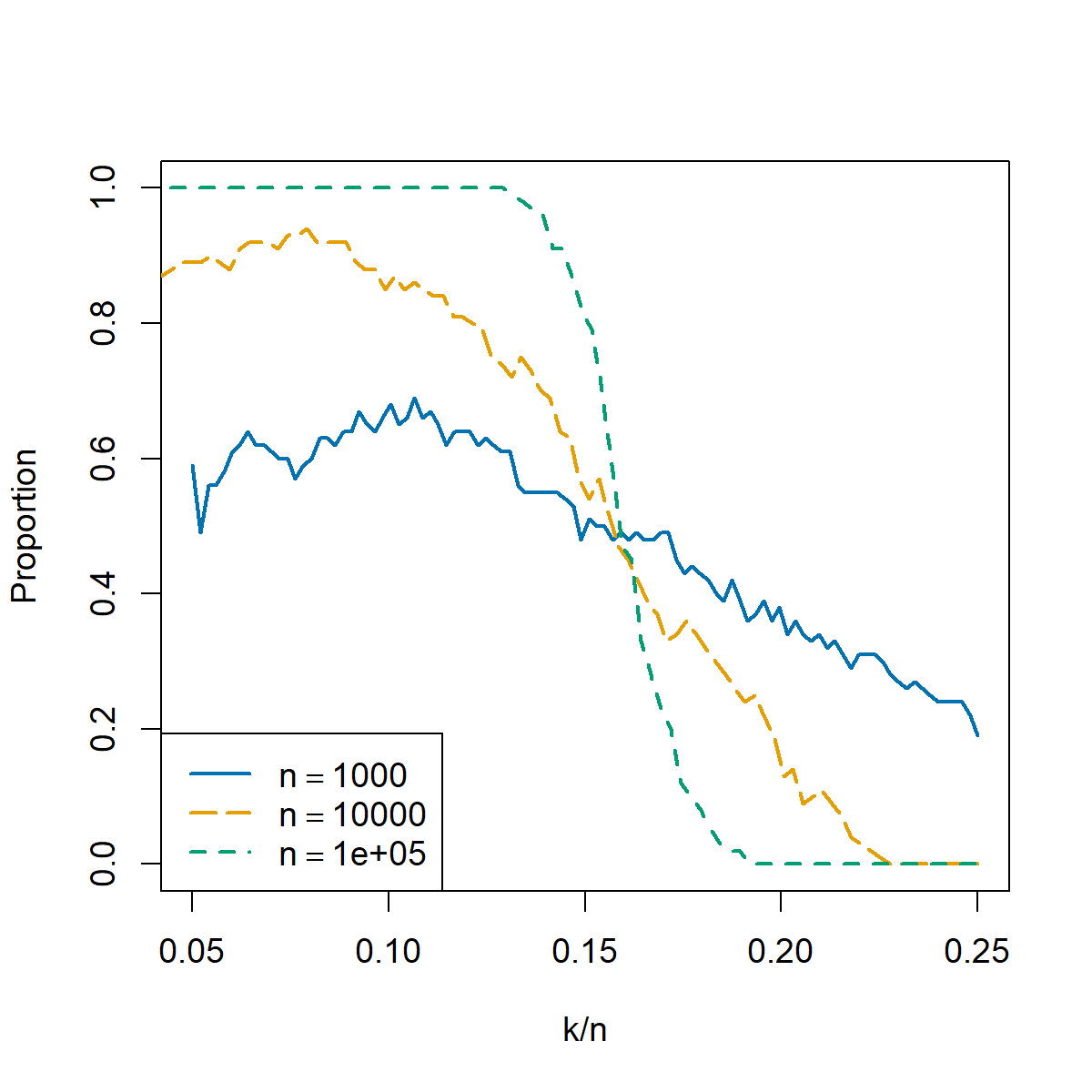}
    \caption{Burr DGP, Sinusoidal scaling $X_i=X_i^2$}
\end{subfigure}\\
\caption{
Proportion of simulations (based on 100 Monte Carlo replications) in which the empirical logarithmic score in \eqref{eq: estimator} is maximized at the true tail index $\gamma=1$ among candidate values $\gamma \in \{0.8, 1, 1.2, 1.5\}$, plotted against the relative number of upper order statistics $k/n$. The top panels correspond to Fr\'echet data-generating distributions and the bottom panels to Burr distributions. The left panels use a linearly varying scale, while the right panels use a sinusoidally varying scale. Curves correspond to sample sizes $n=10^3,10^4,10^5$. Higher values indicate that the logarithmic score more frequently identifies the true tail index.}
\label{fig:LS_mean_scale}
\end{figure}

\subsection{Finite-Sample Performance of Score-Optimization Estimators} 

\begin{figure}[H]
    \centering
    \begin{subfigure}[b]{0.40\textwidth}
        \centering
        \includegraphics[width=\textwidth, trim= 0.3in 0.6in 0.3in 0.6in,clip]{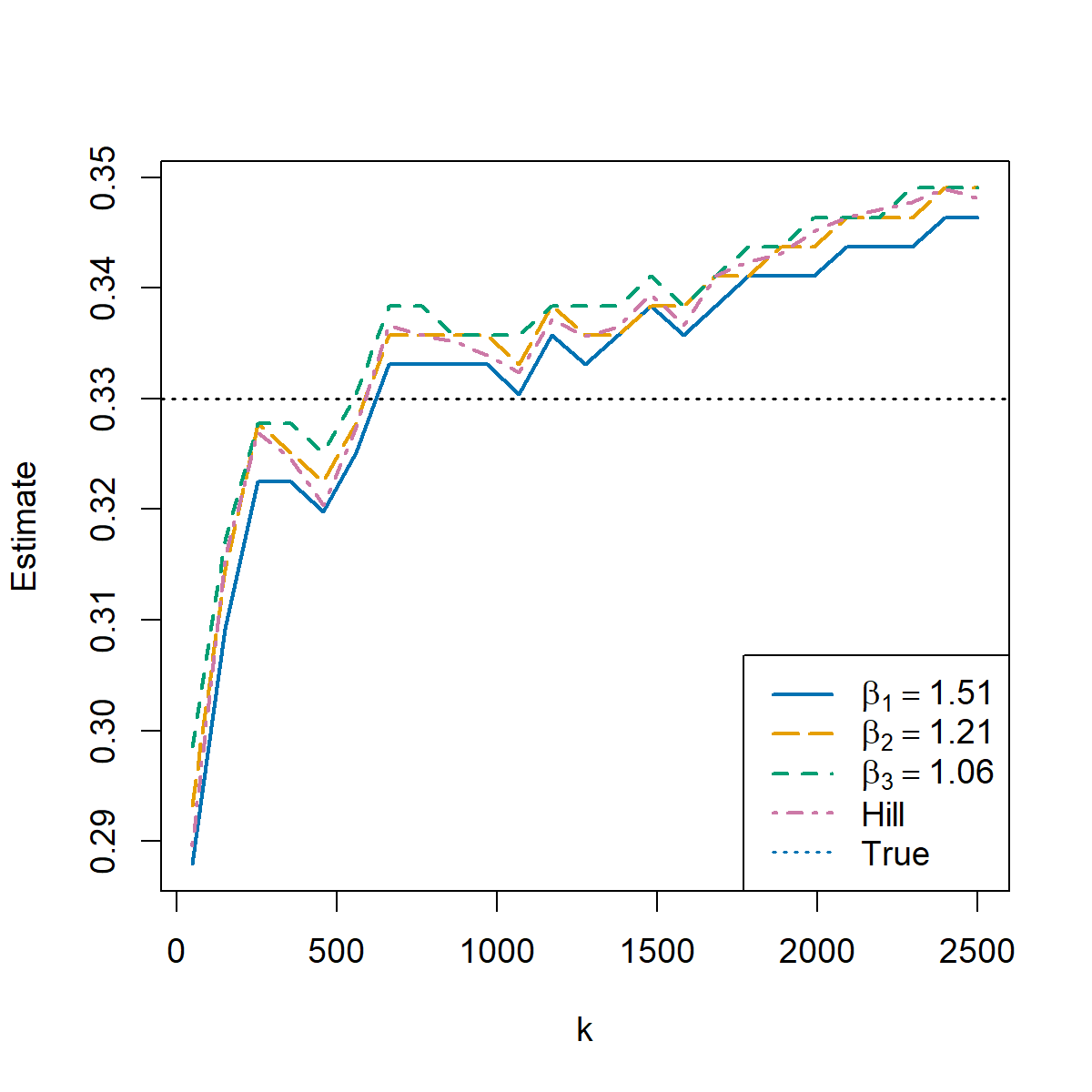}
        \caption{True tail index $\gamma_G=0.33$}
        \label{fig:Ener_fre1}
    \end{subfigure}
    \begin{subfigure}[b]{0.40\textwidth}
        \centering
        \includegraphics[width=\textwidth, trim= 0.3in 0.6in 0.3in 0.6in,clip]{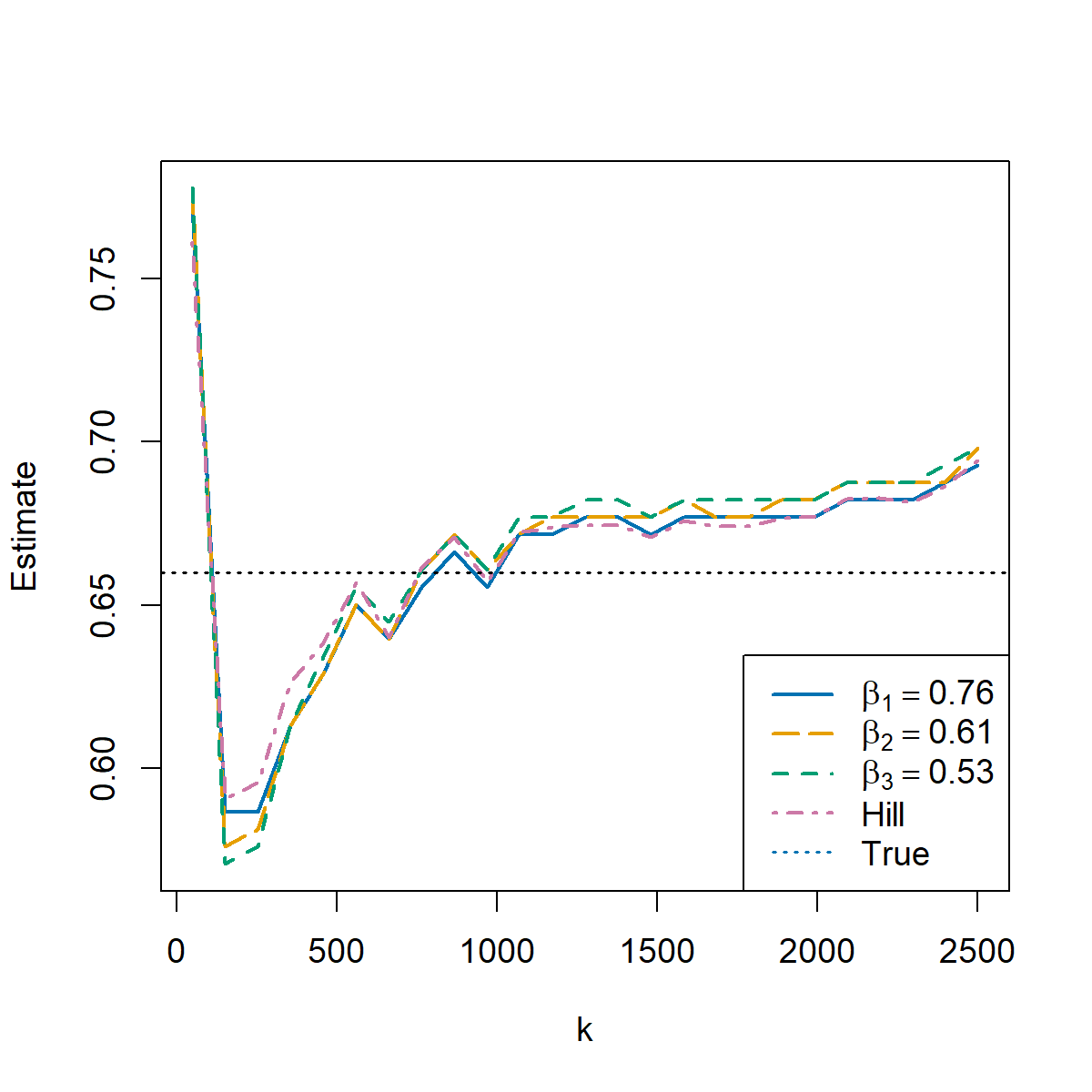}
        \caption{True tail index $\gamma_G=0.66$}
        \label{fig:Ener_fre4}
    \end{subfigure}
    \begin{subfigure}[b]{0.40\textwidth}
        \centering
        \includegraphics[width=\textwidth, trim= 0.3in 0.6in 0.3in 0.6in,clip]{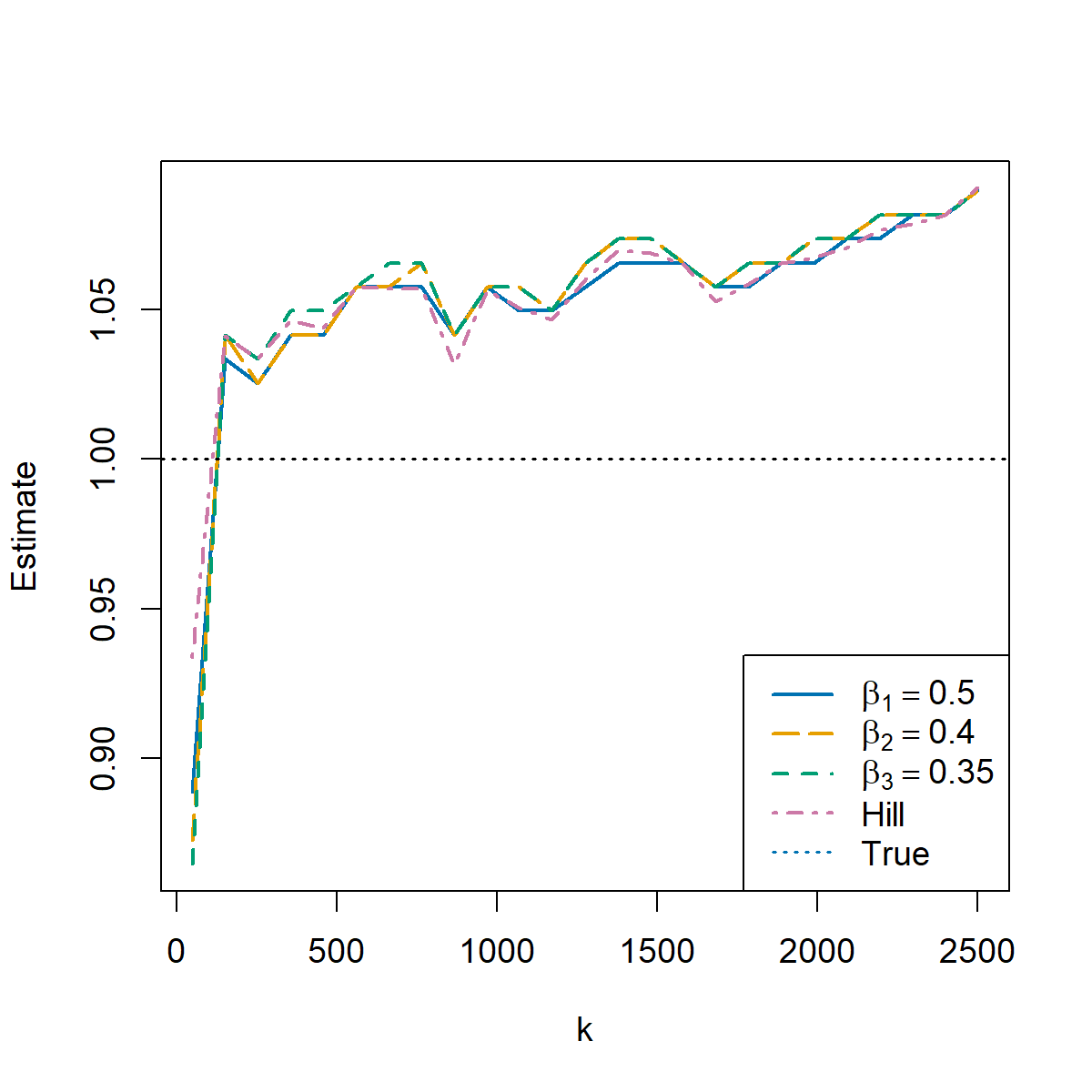}
        \caption{True tail index $\gamma_G=1$}
        \label{fig:Ener_fre5}
    \end{subfigure}
    \begin{subfigure}[b]{0.40\textwidth}
        \centering
        \includegraphics[width=\textwidth, trim= 0.3in 0.6in 0.3in 0.6in,clip]{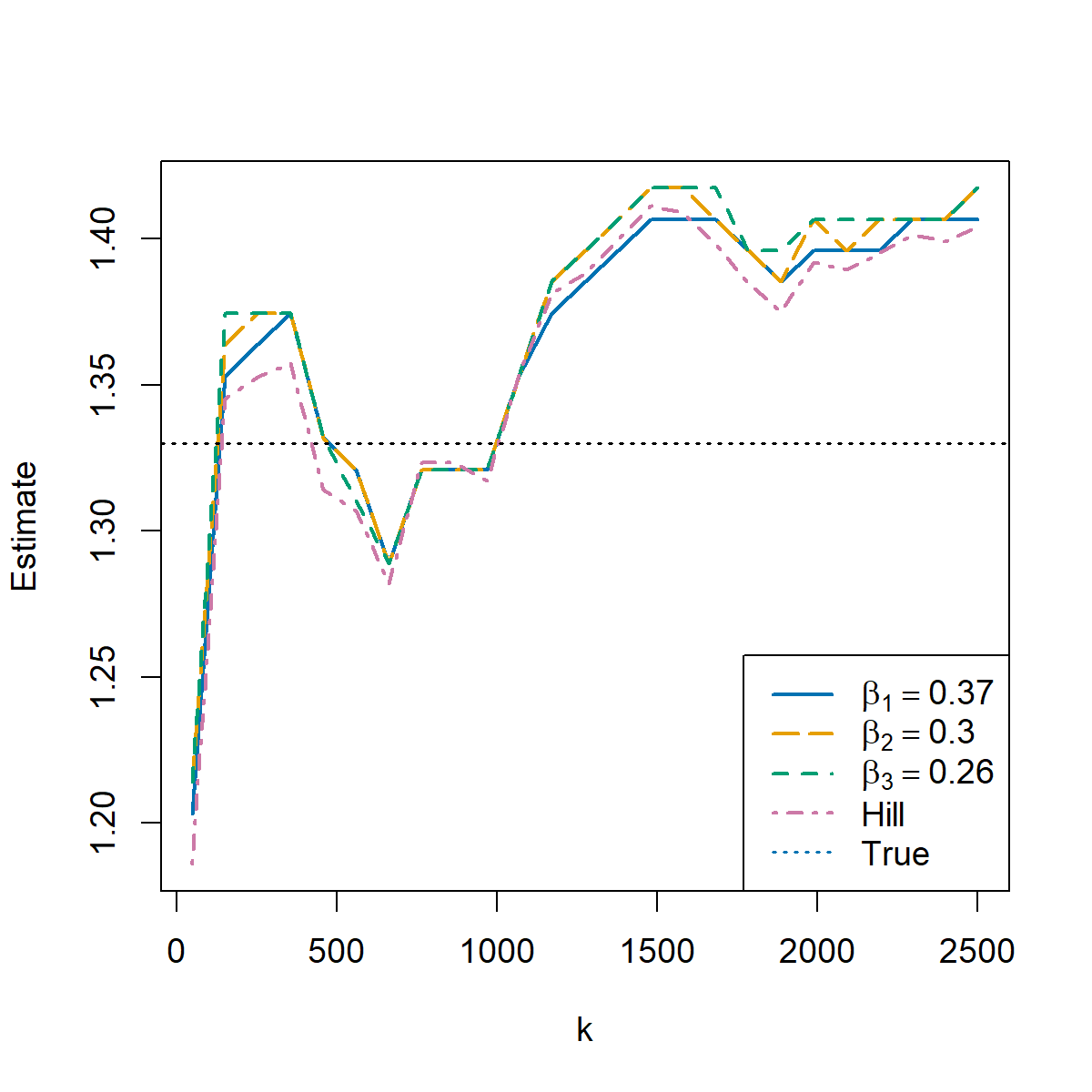}
        \caption{True tail index $\gamma_G=1.33$}
        \label{fig:Ener_fre6}
    \end{subfigure}
    \caption{Tail-index estimates $\hat{\gamma}_k$ (vertical axis) as functions of the number of upper order statistics $k$ (horizontal axis) for Fréchet samples with $n=10^4$. Panels vary the true tail index $\gamma_G$. Each panel compares the Hill estimator with three Energy-score estimators based on $\beta_1,\beta_2,\beta_3$; the horizontal reference line marks $\gamma_G$. 
    }
    \label{fig:Ener_hill_fre}
\end{figure}

We now examine the finite-sample behavior of tail-index estimators based on the Energy score, $\hat{\gamma}_k(ES_\beta)$, {to illustrate the finite-sample behavior of
Energy-score-based estimation.} Their performance is compared with that of the classical Hill estimator over a range of tail indices. Recall that the parameter $\beta$ must satisfy $\beta < 1/\gamma^U$. 

Samples of size $n = 10^4$ are generated from Fr\'echet distributions. Companion Burr runs were also conducted but are not shown, as they display the same qualitative sensitivity in \(k\) and \(\beta\) and were omitted for brevity. As in the previous simulation studies, the remaining parameters are chosen analogously. The number of upper order statistics $k$ is varied over 25 equally spaced values between $50$ and $n/4$.
The Energy score was evaluated over the compact parameter space
\[ 
\Gamma=[\gamma_L,\gamma_U]=[\,0.8\gamma_G,\,2\gamma_G\,],
\]
using a grid of 150 equidistant candidate values. {The interval is a heuristic choice, included to contain the true value while also allowing the estimator to select substantially heavier-tailed alternatives.} To study the sensitivity of the estimator to the parameter $\beta$ in the score, three different values were considered,
\[
\beta_1 = \frac{1}{{\gamma_U}} - 0.001, \qquad
\beta_2 = 0.8\beta_1, \qquad
\beta_3 = 0.7\beta_1.
\]
For each tail index $\gamma_G$, the estimated values $\hat{\gamma}_k$ obtained from the three energy-score-based estimators and the Hill estimator are plotted against $k$. The true value $\gamma_G$ is shown as a horizontal reference line.

{Figure~\ref{fig:Ener_hill_fre} reports the Fréchet results for a single
realization. Across all panels, the choice of $\beta$ has only a limited
impact on the resulting Energy-score estimator. Moreover, the three
Energy-score estimators are almost indistinguishable from the Hill estimator
over most of the considered range of $k$. Thus, in this setting, optimizing
the Energy score leads to essentially the same finite-sample behaviour as the
classical Hill estimator.
This should be viewed in light of Remark~\ref{col:LogS} and
Remark~\ref{rem:var_Logs}. These results give the logarithmic-score estimator,
and hence the Hill estimator, a natural benchmark role. Hence, one should not
expect the Energy-score estimators to provide a clear efficiency improvement in
this experiment. At the same time, the near-overlap of the curves indicates
that using the Energy score does not lead to a noticeable deterioration of the
estimates.
Table~\ref{tab:var_bias} supports the same conclusion. The reported bias and
variance are estimated from 100 independent simulation runs using the usual
sample mean and sample variance of the estimates. For all four values of
$\gamma_G$ and all three choices of $k$, the Energy-score estimators and the
Hill estimator have very similar bias and variance. The differences between
the three choices of $\beta$ are also small, indicating that the finite-sample
performance is not particularly sensitive to the precise value of $\beta$ in
this simulation design.}

\begin{table}[H] 
\centering
\begin{tabular}{c|cc|cc|cc}
\hline
 & \multicolumn{2}{c|}{$k=0.05n$} & \multicolumn{2}{c|}{$k=0.15n$} & \multicolumn{2}{c}{$k=0.25n$} \\
 & {Variance} & {Bias} & {Variance} & {Bias} & {Variance} & {Bias} \\
\hline

\multicolumn{7}{c}{$\gamma = 0.33$} \\
\hline
Hill      
& 2.08e-04 & 3.66e-03 
& 8.15e-05 & 1.279e-02 
& 5.57e-05 & 2.307e-02 \\
$\beta_1$ 
& 1.98e-04 & 3.05e-03 
& 7.58e-05 & 1.121e-02 
& 5.33e-05 & 2.019e-02 \\
$\beta_2$ 
& 2.09e-04 & 3.66e-03 
& 8.25e-05 & 1.227e-02 
& 5.68e-05 & 2.194e-02 \\
$\beta_3$ 
& 2.17e-04 & 3.90e-03 
& 8.50e-05 & 1.275e-02 
& 5.77e-05 & 2.287e-02 \\

\hline
\multicolumn{7}{c}{$\gamma = 0.66$} \\
\hline
Hill      
& 8.72e-04 & 6.10e-03 
& 3.03e-04 & 2.672e-02 
& 2.38e-04 & 4.826e-02 \\
$\beta_1$ 
& 8.71e-04 & 5.40e-03 
& 3.07e-04 & 2.624e-02 
& 2.36e-04 & 4.729e-02 \\
$\beta_2$ 
& 9.21e-04 & 5.78e-03 
& 3.16e-04 & 2.784e-02 
& 2.48e-04 & 5.000e-02 \\
$\beta_3$ 
& 9.33e-04 & 5.72e-03 
& 3.24e-04 & 2.885e-02 
& 2.55e-04 & 5.128e-02 \\

\hline
\multicolumn{7}{c}{$\gamma = 1$} \\
\hline
Hill      
& 1.64e-03 & 1.710e-02 
& 5.48e-04 & 4.377e-02 
& 3.82e-04 & 7.129e-02 \\
$\beta_1$ 
& 1.64e-03 & 1.640e-02 
& 5.85e-04 & 4.411e-02 
& 4.01e-04 & 7.125e-02 \\
$\beta_2$ 
& 1.67e-03 & 1.681e-02 
& 5.75e-04 & 4.548e-02 
& 3.99e-04 & 7.447e-02 \\
$\beta_3$ 
& 1.72e-03 & 1.705e-02 
& 6.36e-04 & 4.644e-02 
& 4.21e-04 & 7.681e-02 \\

\hline
\multicolumn{7}{c}{$\gamma = 1.33$} \\
\hline
Hill      
& 3.38e-03 & 1.640e-02 
& 1.24e-03 & 5.424e-02 
& 7.36e-04 & 9.525e-02 \\
$\beta_1$ 
& 3.54e-03 & 1.785e-02 
& 1.24e-03 & 5.481e-02 
& 7.28e-04 & 9.562e-02 \\
$\beta_2$ 
& 3.56e-03 & 1.914e-02 
& 1.27e-03 & 5.716e-02 
& 8.09e-04 & 1.0076e-01 \\
$\beta_3$ 
& 3.63e-03 & 1.914e-02 
& 1.36e-03 & 5.898e-02 
& 8.07e-04 & 1.0290e-01 \\

\hline
\end{tabular}
\caption{Bias and variance of the Hill estimator and three Energy-score based estimators corresponding to $\beta_1,\beta_2,$ and $\beta_3$, computed from 100 simulated samples drawn from a Fr\'echet distribution with sample size $n = 10^4$. Results are reported for different values of the extreme value index $\gamma$ and for three choices of the number of upper order statistics $k$.
}
\label{tab:var_bias}
\end{table}

\section{Empirical Application: USAutoBI Claim Severity} \label{sec:Data}

The empirical analysis is based on the \texttt{usautoBI} (USAutoBI) automobile bodily injury claims dataset from the \texttt{CASdatasets} package in \textsf{R}. The dataset contains $1{,}340$ observations and $8$ variables, including the economic loss amount (\texttt{LOSS}), measured in thousands of U.S.\ dollars. The data were collected in 2002 by the Insurance Research Council. The variable \texttt{LOSS} is observed for all observations, whereas some observations contain missing values in other variables. Accordingly, all observations are retained in analyses that do not rely on subsetting. When subsetting on variables with missing values, observations with missing entries are excluded. No additional data transformations are applied.  We let $Y$ denote the economic loss amount. In contrast to the previous sections, the score is evaluated for all integer values of $k\ge10$, rather than on a predefined grid.

Following the methodology in Section~\ref{sec:method}, the empirical workflow is: (i) assess Fréchet-domain plausibility using a tail QQ diagnostic; (ii) define a finite set of candidate Pareto tail models; (iii) compute score curves over \(k\); and {(iv) base the final ordering on a stable lower-\(k\) region, while using uncertainty bands to assess the robustness of the resulting ranking.} 

{For the empirical comparison, we focus on the logarithmic score. In the present
application this is the more practical choice, since it avoids the additional choice
of the parameter \(\beta\) required by Energy-score-based estimation. The resulting
empirical comparison is summarized in Figure~\ref{fig:realdata}.} Panel~(a) shows a Pareto quantile--quantile plot for the upper tail, and the near-linear pattern supports Fréchet-domain behavior. {This diagnostic suggests that the Pareto approximation is reasonable over the
upper-tail region used in the score comparison, even though the simulation study
shows that samples of this size may not always produce a stable ranking.} Panels~(b)--(c) display logarithmic score curves as functions of $k$ for five Pareto candidates with $\gamma\in\{0.3,0.5,0.8,1,1.3\}$, where panel~(c) focuses on the lower $25\%$ of the $k$-range. 

{Within this lower \(k\)-range, the rankings vary only slightly with \(k\).
To obtain a single summary ranking, we follow the procedure described in
Section~\ref{sec:method} and rank the models based on \(\bar S\), with
\(\mathcal{K}_{\mathrm{stab}}=\{10,11,\ldots,335\}\). The score curves are
relatively stable over this range, and the corresponding thresholds remain
sufficiently far into the tail. The models with \(\gamma=0.8\) and \(\gamma=1\) yield nearly indistinguishable average scores and are therefore ranked jointly first;} \(\gamma=1.3\) is ranked second, \(\gamma=0.5\) third, and \(\gamma=0.3\) fourth. {Alternative approaches for selecting a single threshold value \(k\), rather than a stability range, are available in
the extreme value literature. See, for example, \cite{bladt2020threshold} for an overview.
These methods could be adapted to the present score-based framework.} 

{To assess the uncertainty in these score comparisons, we consider score
differences relative to the benchmark \(\gamma=1\). This benchmark is chosen
because it is one of the two jointly highest-ranked models in the summary
ranking above. For each alternative \(\gamma\), a positive expected score
difference indicates that \(F_\gamma\) attains a higher tail score than \(F_1\),
whereas a negative value favours the benchmark. Confidence intervals for these
differences therefore indicate whether the observed differences are clearly
separated from zero, or whether the corresponding models should be viewed as
statistically indistinguishable at the given threshold level. Let
\begin{align*}
h_\gamma(y)=\text{LogS}(F_\gamma,y)-\text{LogS}(F_1,y)=\log(1/\gamma)+\left(1-\frac{1}{\gamma}\right)\log y, \qquad y\ge 1
\end{align*}
be the score difference function relative to the reference model $\gamma=1$. Panels~(d)--(f) then report 
\begin{align*}
    \frac{1}{k}\sum_{i=1}^kh_\gamma\left(\frac{Y_{n,n-i+1}}{Y_{n,n-k}}\right)
\end{align*}
together with pointwise 95\% confidence intervals for different $\gamma$'s.  
The confidence intervals are based on Theorem~\ref{thm:normality}, applied to the score difference function \(h_\gamma\). The function \(h_\gamma\) is absolutely continuous with derivative proportional to \(y^{-1}\). Hence, for any \(\rho \in (0,(2\gamma_G)^{-1})\), the derivative condition in Theorem~\ref{thm:normality} is satisfied. Consequently,  an approximate \(95\%\) confidence interval for
\(
\mathbb{E}\!\left[h_\gamma(Y^\circ)\right]
\)
is given by
\begin{align*}
     \frac{1}{k}\sum_{i=1}^kh_\gamma\left(\frac{Y_{n,n-i+1}}{Y_{n,n-k}}\right)\pm z_{0.975}\frac{\hat\sigma_k(\gamma)}{\sqrt{k}},
\end{align*}
where 
\[
\hat\sigma_k^2(\gamma)
=
\frac{1}{k-1}\sum_{i=1}^k
\left[
h_\gamma\left(\frac{Y_{n,n-i+1}}{Y_{n,n-k}}\right)
-
\frac{1}{k}\sum_{i=1}^kh_\gamma\left(\frac{Y_{n,n-i+1}}{Y_{n,n-k}}\right)
\right]^2 .
\]
The pointwise confidence intervals are then calculated for each $k$. Using the sample variance avoids substituting an estimate of the unknown tail index
\(\gamma_G\), although the resulting intervals should be interpreted with some
caution for small values of \(k\).}

{The score-difference plots are consistent with the identification of
\(\gamma=0.8\) and \(\gamma=1\) as the leading candidates, and indicate that this
comparison is stable across the selected lower \(k\)-range. The model
\(\gamma=0.5\) has negative score differences over most of this range, indicating
weaker performance than the reference model. In contrast, the differences for
\(\gamma=0.8\) are close to zero, with intervals that largely overlap zero, which
supports treating \(\gamma=0.8\) and \(\gamma=1\) as jointly best. The model
\(\gamma=1.3\) also has mostly negative score differences relative to
\(\gamma=1\).}

\begin{figure}[H]
    \centering
    \begin{subfigure}[b]{0.40\textwidth}
        \centering
        \includegraphics[width=\textwidth, trim= 0.3in 0.6in 0.3in 0.6in,clip]{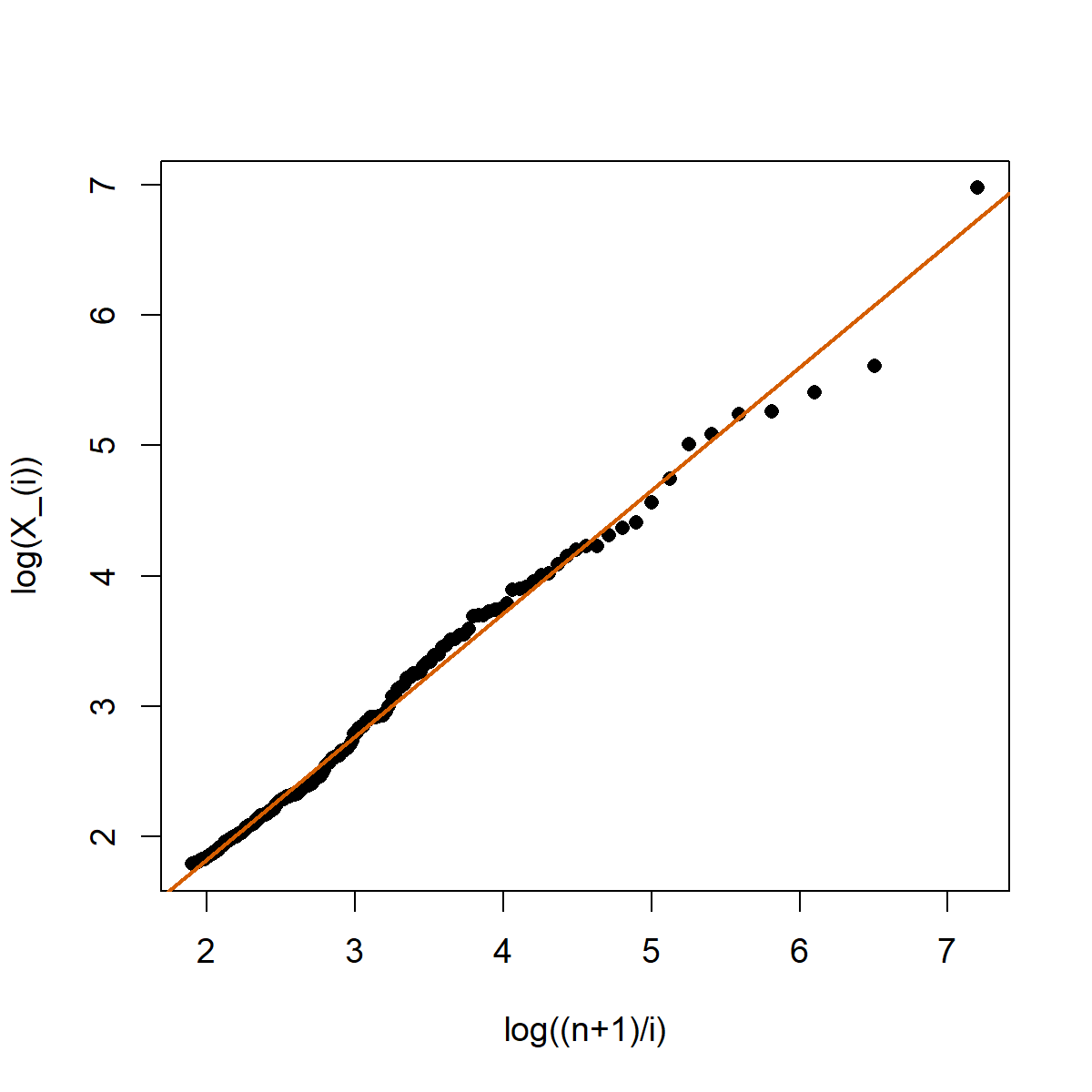}
        \caption{Upper-tail Pareto QQ-plot}
        \label{fig:qqpareto}
    \end{subfigure}
    \begin{subfigure}[b]{0.40\textwidth}
        \centering
        \includegraphics[width=\textwidth, trim= 0.3in 0.6in 0.3in 0.6in,clip]{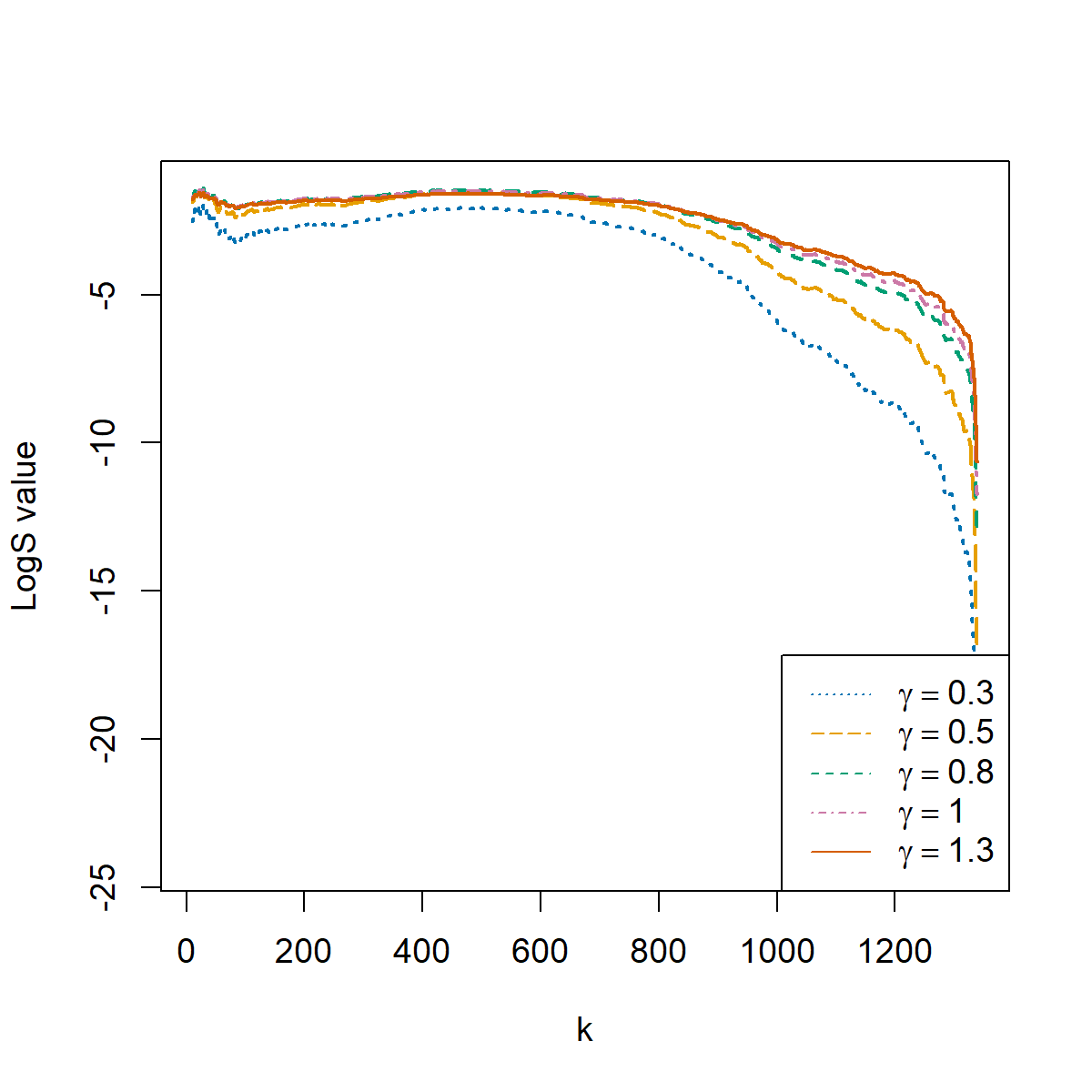}
        \caption{LogS versus $k$ over the full range}
        \label{fig:realdata_full}
    \end{subfigure}
    \begin{subfigure}[b]{0.40\textwidth}
        \centering
        \includegraphics[width=\textwidth, trim= 0.3in 0.6in 0.3in 0.6in,clip]{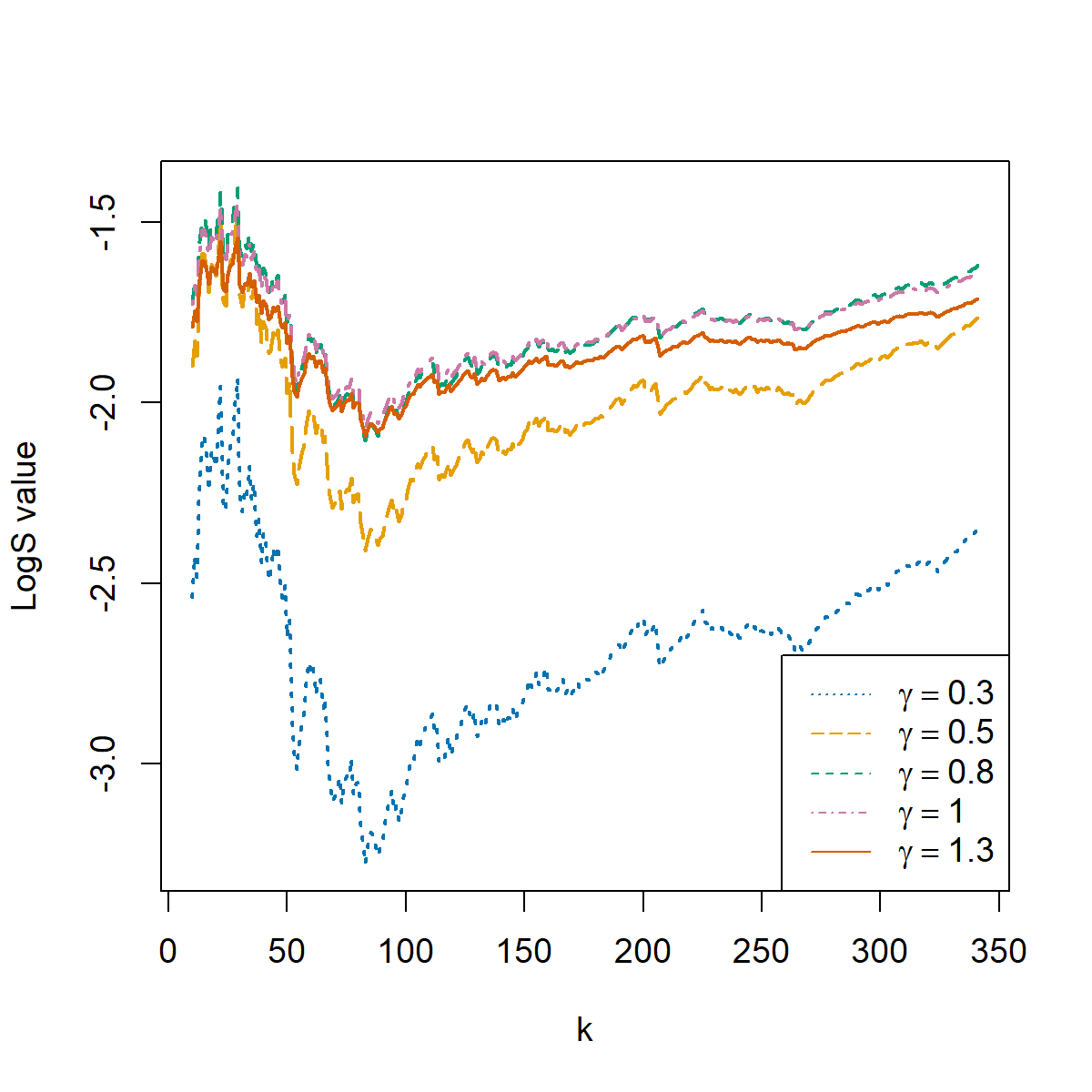}
        \caption{LogS versus $k$ for the lower $25\%$ of the range}
        \label{fig:realdata_zoom}
    \end{subfigure}
    \begin{subfigure}[b]{0.4\textwidth}
        \centering
        \includegraphics[width=\textwidth, trim= 0.3in 0.6in 0.3in 0.6in,clip]{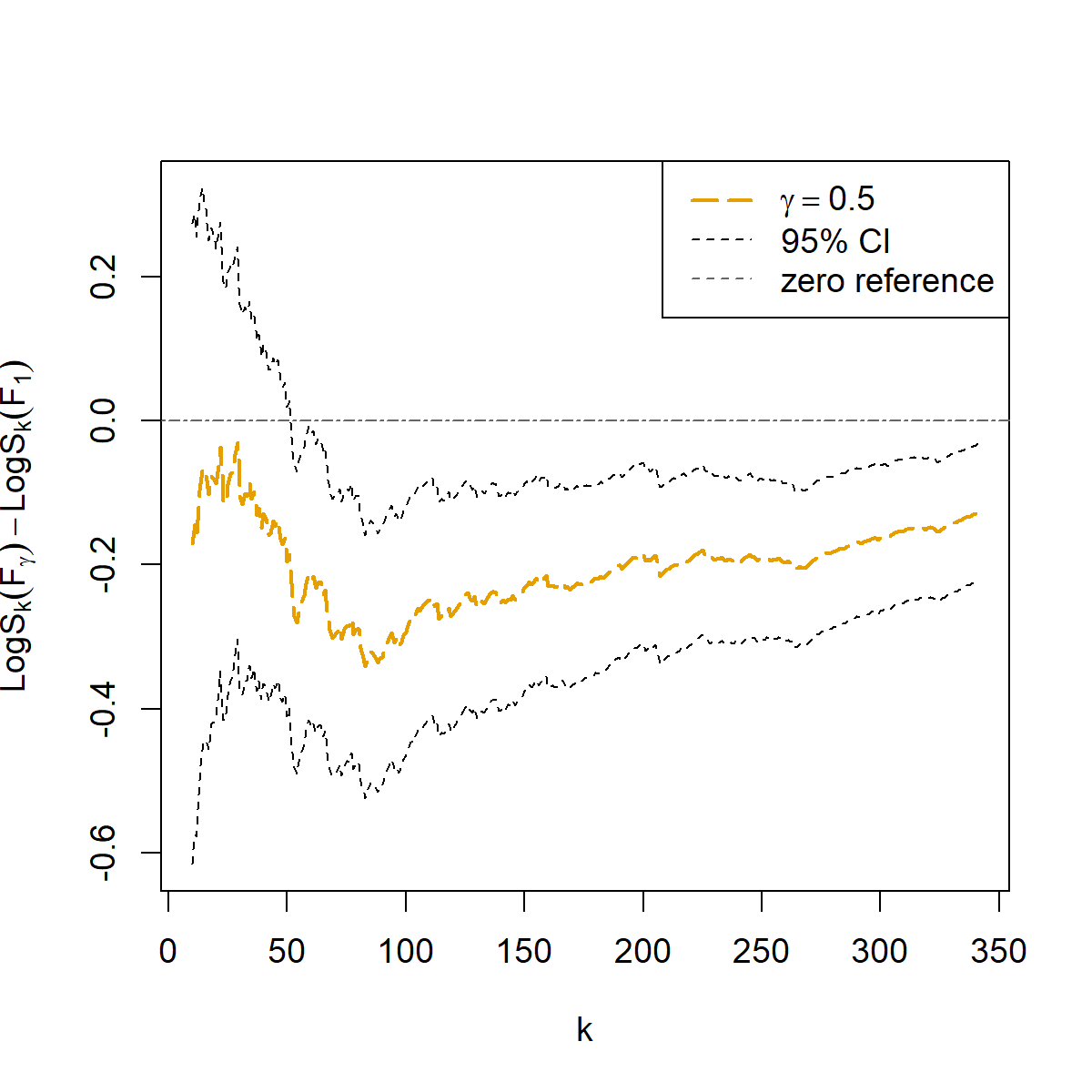}
        \caption{\(S_k(F_{0.8})-S_k(F_1)\) versus $k$.
        }
        \label{fig:realdata_CI_05}
    \end{subfigure}
    \begin{subfigure}[b]{0.40\textwidth}
        \centering
        \includegraphics[width=\textwidth, trim= 0.3in 0.6in 0.3in 0.6in,clip]{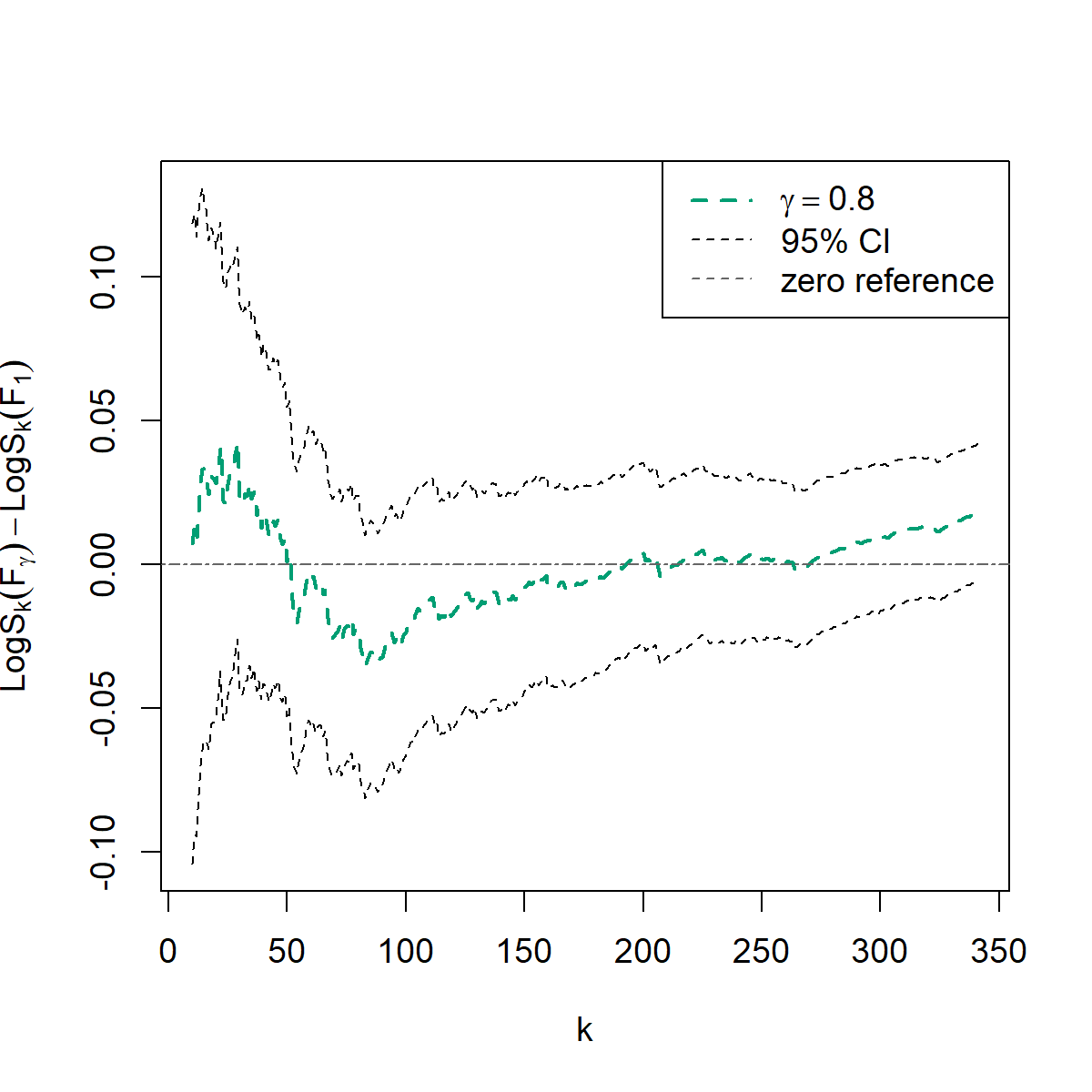}
        \caption{\(S_k(F_{0.5})-S_k(F_1)\) versus $k$. 
        }
        \label{fig:realdata_CI_08}
    \end{subfigure}
    \begin{subfigure}[b]{0.4\textwidth}
        \centering
        \includegraphics[width=\textwidth, trim= 0.3in 0.6in 0.3in 0.6in,clip]{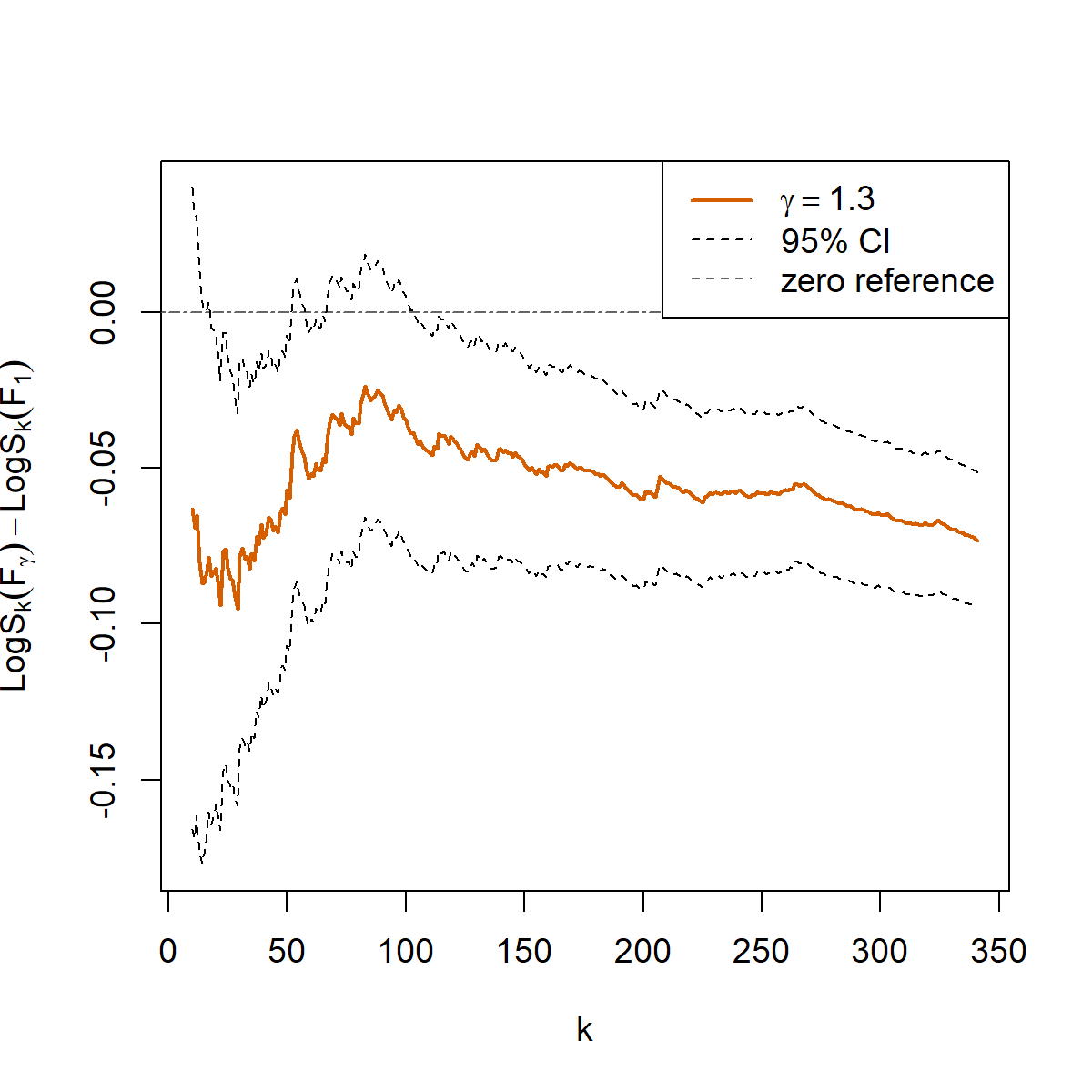}
        \caption{\(S_k(F_{1.3})-S_k(F_1)\) versus $k$. 
        }
        \label{fig:realdata_CI_13}
    \end{subfigure}
    \caption{USAutoBI claim severity analysis ($n=1{,}340$). Panel~(a) compares theoretical Pareto upper-tail quantiles (horizontal axis) with empirical quantiles (vertical axis). Panels~(b)--(c) plot logarithmic score (vertical axis) against the number of upper order statistics $k$ (horizontal axis) for five Pareto predictive distributions with $\gamma\in\{0.3,0.5,0.8,1,1.3\}$. Panel~(c) is a zoom of panel~(b) over the lower $25\%$ of the $k$-range. {Panels~(d)--(f) show score differences relative to the \(\gamma=1\) model with pointwise 95\% confidence intervals.}
    }
    \label{fig:realdata}
\end{figure}

To assess subgroup variation, additional insight is obtained by repeating the score analysis on subsamples. Figure~\ref{fig:realdata_subset} reports splits by sex (Figures~\ref{fig:real_female} and \ref{fig:real_men}) and attorney involvement (Figures~\ref{fig:real_A0} and \ref{fig:real_A1}), using the same relative lower-\(k\) range as in Figure~\ref{fig:realdata_zoom}. The sex-based split is defined using the variable \texttt{CLMSEX}, where \texttt{F} and \texttt{M} denote female and male, respectively; observations with missing values (12 in total) are excluded. The split by attorney involvement is based on the variable \texttt{ATTORNEY}, where $1$ indicates the presence of an attorney and $0$ otherwise.
For the sex split, the ranking is broadly consistent with the full sample:
$\gamma=0.8$ and $\gamma=1$ attain the highest scores. In contrast, the attorney split yields a different ordering for claims without attorney representation (Figure~\ref{fig:real_A0}), where $\gamma=0.5$ scores highest and $\gamma=0.3$ becomes competitive. This suggests a lighter tail in the non-attorney subsample, which is plausible if smaller claims are less likely to involve legal representation. As usual for observational subgroup analyses, this interpretation is descriptive and does not isolate causal effects of covariates.

\begin{figure}[H]
    \centering
    \begin{subfigure}[b]{0.40\textwidth}
        \centering
        \includegraphics[width=\textwidth, trim= 0.3in 0.6in 0.3in 0.6in,clip]{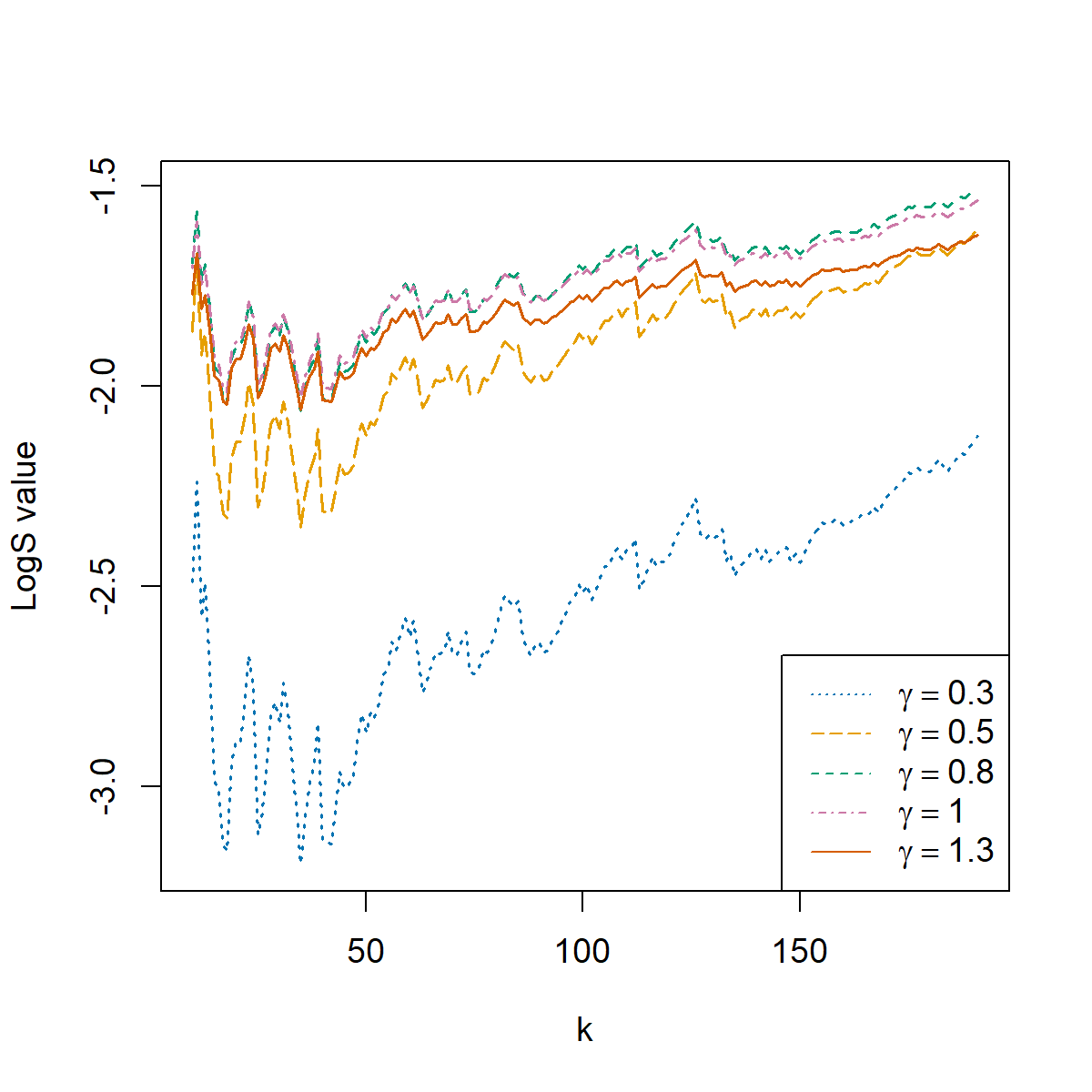}
        \caption{Female subsample ($n=742$)}
        \label{fig:real_female}
    \end{subfigure}
    \begin{subfigure}[b]{0.40\textwidth}
        \centering
        \includegraphics[width=\textwidth, trim= 0.3in 0.6in 0.3in 0.6in,clip]{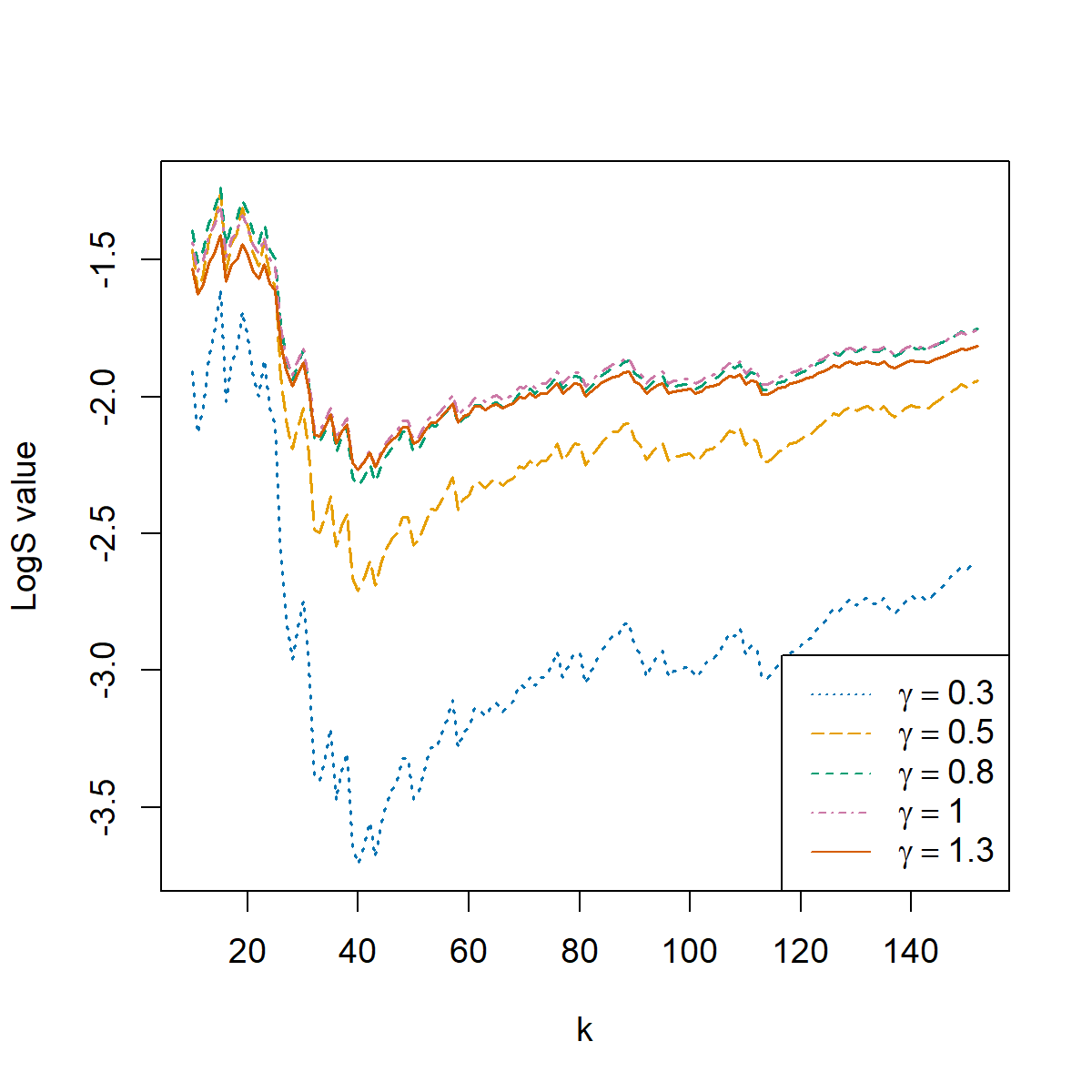}
        \caption{Male subsample ($n=586$)}
        \label{fig:real_men}
    \end{subfigure}
    \begin{subfigure}[b]{0.40\textwidth}
        \centering
        \includegraphics[width=\textwidth, trim= 0.3in 0.6in 0.3in 0.6in,clip]{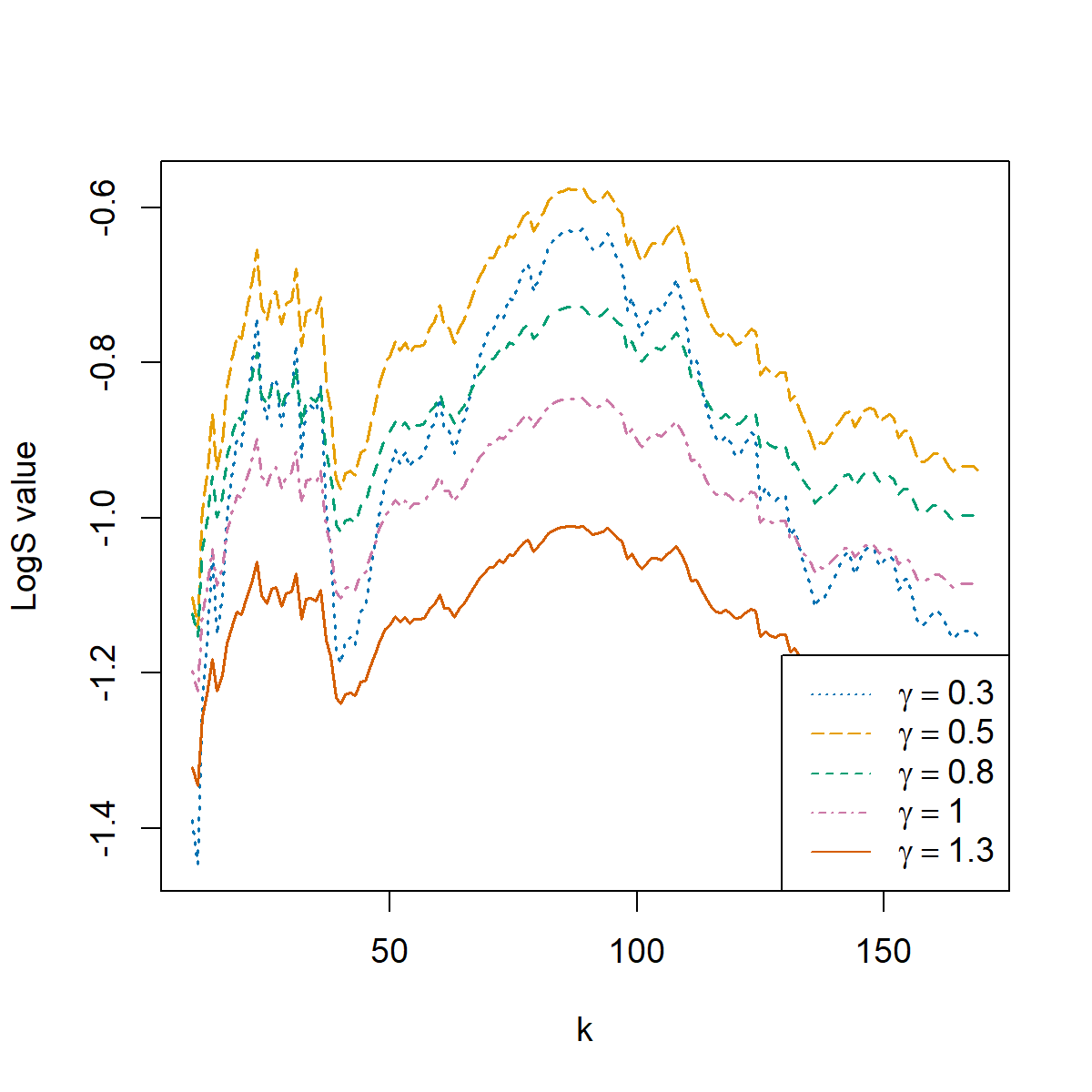}
        \caption{No-attorney subsample ($n=655$)}
        \label{fig:real_A0}
    \end{subfigure}
    \begin{subfigure}[b]{0.4\textwidth}
        \centering
        \includegraphics[width=\textwidth, trim= 0.3in 0.6in 0.3in 0.6in,clip]{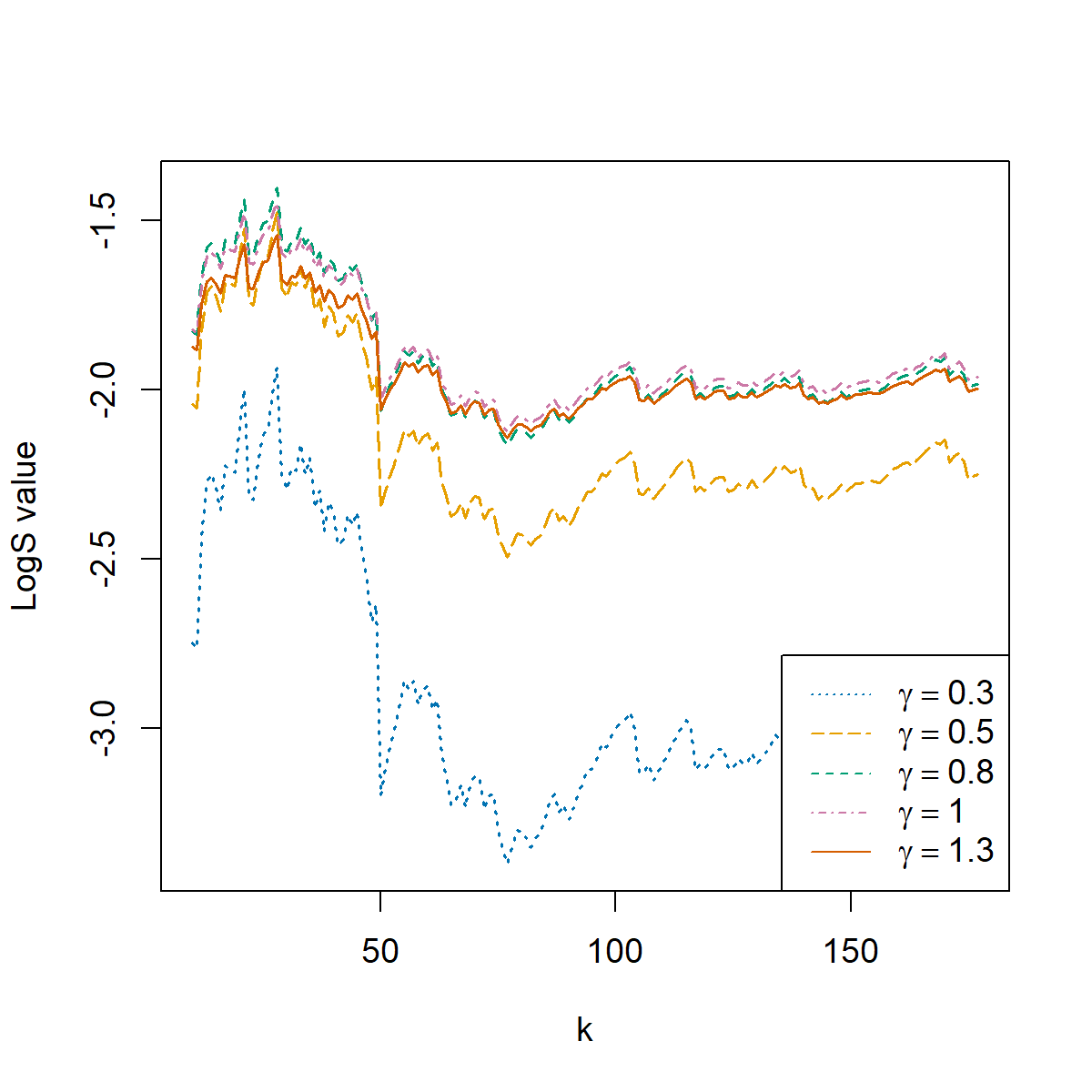}
        \caption{Attorney subsample ($n=685$)}
        \label{fig:real_A1}
    \end{subfigure}
    \caption{Subset analyses of the USAutoBI claim severity data. Each panel plots logarithmic score (vertical axis) against the number of upper order statistics $k$ (horizontal axis) over the lower $25\%$ of the $k$-range, for the same five Pareto candidates used in Figure~\ref{fig:realdata}. Panels~(a)--(b) split by sex; panels~(c)--(d) split by attorney involvement.
    }
    \label{fig:realdata_subset}
\end{figure}

\section{Concluding Remarks}

This paper proposes a scoring-rule-based framework for ranking predictive distributions in the Fréchet domain based on normalized order statistics. By embedding extreme value limit theory into a scoring-rule perspective, we show how proper scoring rules can be applied to compare tails in a principled manner. Within this framework, we derive conditions under which the logarithmic score and the CRPS (as a special case of the Energy score) are well defined and applicable.

We further show that optimizing scoring rules yields consistent tail-index estimators and that the classical Hill estimator arises as a special case.

The first two simulation studies show that the proposed approach can successfully distinguish between different tail indices in finite samples and that
its performance remains stable under systematic scale variation. {In the third simulation study, we observe that estimation based on Energy-score optimization yields results that are very close to those of the classical Hill estimator.} Lastly, we conduct an empirical analysis of automobile claim severity data, in which five competing Pareto tail models are evaluated and ranked using the proposed scoring-rule framework. The analysis illustrates how the method can be applied in practice, how the resulting ranking supports model selection, and how differences in tail behavior across data partitions can be identified.

{In this paper, the score-specific conditions are verified only for the logarithmic and Energy scores, making the study of further scoring rules a natural direction for future work. It would also be useful to study how existing threshold-selection methods can be adapted to the present score-based framework. Finally, extending the asymptotic theory to covariate-dependent or otherwise non-identically distributed predictive distributions would broaden the scope of the proposed methodology.
}

\paragraph{Acknowledgments}
AI tools were used exclusively for language editing. All results, analysis, and conclusions are the authors’ own.

\paragraph{Funding Statement}
The first author was supported by the Carlsberg Foundation, grant CF23-1096.

\paragraph{Competing Interests}
None to declare.

\section*{Acknowledgements}

AI tools were used exclusively for language editing. All results, analysis, and conclusions are the authors’ own.

\section*{Competing Interest}

None to declare.

\section*{Funding}
The first author was supported by the Carlsberg Foundation, grant CF23-1096.

\newpage
\bibliographystyle{imsart-nameyear} 
\bibliography{main.bib}       

\end{document}